\documentclass{article}
\usepackage{amsfonts}
\usepackage{amssymb}
\usepackage{graphicx}
\usepackage{amsmath}
\usepackage{amsthm}
\usepackage{mathrsfs}
\usepackage{pifont}
\usepackage{comment}
\usepackage[colorlinks=true,linkcolor=blue,citecolor= red,bookmarksopen=true]{hyperref}

\nonstopmode
\newtheorem{theorem}{Theorem}[section]
\newtheorem{proposition}[theorem]{Proposition}
\newtheorem{lemma}[theorem]{Lemma}

\newtheorem{claim}[theorem]{Claim}

\newtheorem{corollary}[theorem]{Corollary}

\newtheorem{definition}[theorem]{Definition}

\newtheorem*{standingassumptionI*}{Standing Assumption I}
\newtheorem*{standingassumptionII*}{Standing Assumption II}
\newtheorem{assumption}[theorem]{Assumption}

\theoremstyle{remark}
\newtheorem{example}[theorem]{Example}
\newtheorem{remark}[theorem]{Remark}
\newcommand{\probp}{{\mathbb P}}
\newcommand{\probq}{{\mathbb Q}}

\DeclareMathOperator*{\esssup}{ess\,sup}
\DeclareMathOperator*{\essinf}{ess\,inf}

\newcommand{\cG}{\mathcal{G}}

\begin{document}

\title{Conditional Systemic Risk Measures}
\date{\today}
\author{Alessandro Doldi\thanks{%
Dipartimento di Matematica, Universit\`{a} degli Studi di Milano, Via
Saldini 50, 20133 Milano, Italy, $\,\,$\emph{alessandro.doldi@unimi.it}. }
\and Marco Frittelli\thanks{%
Dipartimento di Matematica, Universit\`a degli Studi di Milano, Via Saldini
50, 20133 Milano, Italy, \emph{marco.frittelli@unimi.it}. }}
\maketitle

\begin{abstract}
\noindent We investigate to which extent the relevant features of (static)
Systemic Risk Measures can be extended to a conditional setting. After
providing a general dual representation result, we analyze in greater detail
Conditional Shortfall Systemic Risk Measures. In the particular case of
exponential preferences, we provide explicit formulas that also allow us to
show a time consistency property. Finally, we provide an interpretation of
the allocations associated to Conditional Shortfall Systemic Risk Measures
as suitably defined equilibria. Conceptually, the generalization from static
to conditional Systemic Risk Measures can be achieved in a natural way, even
though the proofs become more technical than in the unconditional framework.
\end{abstract}

\parindent0em \noindent 

\bigskip \textbf{Acknowledgments:} the authors thank two anonymous referees
for precious advice and comments. %

\section{Introduction}


We provide a natural extension of static Systemic Risk Measures to a
dynamic, conditional setting, and we study related concepts of time
consistency and equilibrium.

To put the principal findings of this paper into prospective, we briefly
review the literature pertaining to Systemic Risk Measures. We let $%
X=[X^{1},\ldots ,X^{N}]\in (L^{0}(\Omega ,\mathcal{F},\mathbb{P}))^{N}$ be a
vector of $N$ $\mathbb{P}$-a.s. finite random variables on the probability
space $(\Omega ,\mathcal{F},\mathbb{P}),$ representing a configuration of
risky (financial) factors at a future time $T$ associated to a system of $N$
financial institutions/banks.

A traditional approach to evaluate the risk of each institution $j\in
\{1,\dots,N\}$ is to apply a \textit{univariate monetary Risk Measure} $\eta
^{j}$ to the single financial position $X^{j}$, yielding $\eta ^{j}(X^{j})$.
Let $L$ be a subspace of $L^{0}(\Omega ,\mathcal{F},\mathbb{P})$. A monetary
Risk Measure (see \cite{FollmerSchied2}) is a map $\eta :$ $L\rightarrow 
\mathbb{R}$ that can be interpreted as the minimal capital needed to secure
a financial position with payoff $Z\in L$, i.e., the minimal amount $m\in 
\mathbb{R}$ that must be added to $Z$ in order to make the resulting
(discounted) payoff at time $T$ acceptable 
\begin{equation}
\eta (Z):=\inf \{m\in \mathbb{R}\mid Z+m\in \mathbb{A}\},  \label{DynRM1}
\end{equation}%
where the acceptance set $\mathbb{A}\subseteq L^{0}(\Omega ,\mathcal{F},%
\mathbb{P})$ is assumed to be monotone, i.e., $Z\geq Y\in \mathbb{A}$
implies $Z\in \mathbb{A}$. Then $\eta $ is monotone decreasing and satisfies
the cash additivity property 
\begin{equation}
\eta (Z+m)=\eta (Z)-m\text{, for all }m\in \mathbb{R}\text{ and }Z\in L.
\label{DynRMCAdd}
\end{equation}%
Under the assumption that the set $\mathbb{A}$ is convex (resp. is a convex
cone) the maps in \eqref{DynRM1} are convex (resp. convex and positively
homogeneous) and are called \textit{convex} (resp. coherent) \textit{Risk
Measures}, see Artzner et al. (1999) \cite{ADEH99}, F\"{o}llmer and Schied
(2002) \cite{FollmerSchied}, Frittelli and Rosazza Gianin (2002) \cite%
{FrittelliRosazza}. Once the risk $\eta ^{j}(X^{j})$ of each institution $%
j\in \{1,...,N\}$ has been determined, the quantity 
\begin{equation*}
\rho (X):=\sum_{j=1}^{N}\eta ^{j}(X^{j})
\end{equation*}%
could be used as a very preliminary and naive assessment of the risk of the
entire system.

\subsection{Static Systemic Risk Measures}

The approach sketched above does not clearly capture systemic risk of an
interconnected system, and the design of more adequate Risk Measures for
financial systems is the topic of a vast literature on systemic risk. Let $%
L_{\mathcal{F}}$ be a vector subspace of $(L^{0}(\Omega ,\mathcal{F},{%
\mathbb{P}}))^{N}.$ For example, we may take as $L_{\mathcal{F}}$\ the space 
$(L^{p}(\Omega ,\mathcal{F},{\mathbb{P}}))^{N}$, $p\in \lbrack 1,\infty ],$
of (equivalence classes of) $p$-integrable (or essentially bounded if $%
p=\infty $) $N$-dimensional vectors of random variables on $(\Omega ,%
\mathcal{F},{\mathbb{P}})$. A Systemic Risk Measure is a map $\rho :L_{%
\mathcal{F}}\rightarrow \mathbb{R}$ that evaluates the risk $\rho (X)$ of
the complete system $X\in L_{\mathcal{F}}$ and satisfies additionally
financially reasonable properties. The subspace $L_{\mathcal{F}}$ of $%
(L^{0}(\Omega ,\mathcal{F},\mathbb{P}))^{N}$ may represent possible
additional integrability or boundedness requirements.

\paragraph{\textbf{First aggregate then allocate.}}

In Chen et al. (2013) \cite{ChenIyengarMoallemi} the authors investigated
under which conditions a Systemic Risk Measure could be written in the form 
\begin{equation}
\rho (X)=\eta (U(X))=\inf \{m\in \mathbb{R}\mid U(X)+m\in \mathbb{A}\},
\label{DynRM2}
\end{equation}%
for some univariate monetary Risk Measure $\eta $ and some aggregation rule%
\begin{equation*}
U:\mathbb{R}^{N}\rightarrow \mathbb{R}
\end{equation*}
that aggregates the $N$-dimensional risk factors into a univariate risk
factor. We also refer to Kromer et al. \cite{Kromer} (2013) for extension to
general probability space.

Such systemic risk might again be interpreted as the minimal cash amount
that secures the system when it is added to the total aggregated system loss 
$U(X)$, given that $U(X)$ allows for a monetary loss interpretation. Note,
however, that in (\ref{DynRM2}) systemic risk is the minimal capital added
to secure the system \textit{after aggregating individual risks.}

\paragraph{\textbf{First allocate then aggregate.}}

A second approach consisted in measuring systemic risk as the minimal cash
that secures the aggregated system by adding the capital into the single
institutions \textit{before aggregating their individual risks}. This way of
measuring systemic risk can be expressed by 
\begin{equation}
\rho (X):=\inf \left\{ \sum_{j=1}^{N}m^{j}\mid m=[m^{1},\dots ,m^{N}]\in 
\mathbb{R}^{N},\,U(X+m)\in \mathbb{A}\right\} \,.  \label{DynRM3}
\end{equation}%
Here, the amount $m^{j}$ is added to the financial position $X^{j}$ of
institution $j\in \{1,\dots ,N\}$ before the corresponding total loss $%
U(X+m) $ is computed. Such Systemic Risk Measures were introduced and
analyzed by Biagini et al. (2019) \cite{BFFMB}. Feinstein et al. (2017) \cite%
{FeinsteinRudloffWeber} introduced a similar approach for set-valued Risk
Measures. We refer to Armenti et al. (2018) \cite{Drapeau} and Biagini et
al. (2020) \cite{bffm} for a detailed study of Shortfall Systemic Risk
Measures - a relevant subclass of Risk Measures in the form \eqref{DynRM3} -
and their dual representations. More recently, dual representations of
Systemic Risk Measures based on acceptance sets have been studied in Arduca
et al. (2019) \cite{ArducaKochMedinaMunari19} for the real-valued case, in
Ararat and Rudloff (2020) \cite{AraratRudloff20} in the set-valued case.

\paragraph{\textbf{Scenario dependent allocations.}}

The \textquotedblleft \textit{first allocate and then aggregate}%
\textquotedblright\ approach was then extended in Biagini et al. (2019) \cite%
{BFFMB} and (2020) \cite{bffm} by adding to $X$ not merely a vector $%
m=[m_{1},\dots ,m_{N}]\in \mathbb{\ R}^{N}$ of deterministic amounts but,
more generally, a random vector $Y\in \mathcal{C}$, for some given class $%
\mathcal{C}$. In particular, one main example considered in \cite{bffm} is
given by the class $\mathcal{C}$ such that 
\begin{equation}
\mathcal{C}\subseteq \mathcal{C}_{\mathbb{R}}\cap L_{\mathcal{F}}\text{,
where \ }\mathcal{C}_{\mathbb{R}}:=\left\{ Y\in (L^{0}(\Omega ,\mathcal{F},%
\mathbb{P}))^{N}\mid \sum_{j=1}^{N}Y^{j}\in \mathbb{R}\right\} .
\label{DynRMdefC}
\end{equation}
Here, the notation $\sum_{j=1}^{N}Y^{j}\in \mathbb{R}$ means that $%
\sum_{j=1}^{N}Y^{j}$ \textit{is equal to some deterministic constant in }$%
\mathbb{R}$, \textit{even though each single }$Y^{j}$, $j=1,\dots ,N$, 
\textit{is a random variable}. It is possible to model additional
constraints on the allocation $Y\mathbf{\in }\mathcal{C}_{\mathbb{R}}$ by
requiring $Y\in \mathcal{C}\subseteq \mathcal{C}_{\mathbb{R}}$. The set $%
\mathcal{C}$ represents the class of feasible allocations and it is assumed
that ${\mathbb{R}}^{N}\subseteq \mathcal{C}$.

Under these premises the Systemic Risk Measure considered in \cite{bffm}
takes the form%
\begin{equation}
\rho (X):=\inf \left\{ \sum_{j=1}^{N}Y^{j}\mid Y\in \mathcal{C},\text{ }%
U(X+Y)\in \mathbb{A}\right\}  \label{DynRMSRM}
\end{equation}%
and can still be interpreted, since $\mathcal{C}\subseteq \mathcal{C}_{%
\mathbb{R}}$, as the minimal total cash amount $\sum_{j=1}^{N}Y^{j}\in 
\mathbb{R}$ needed today to secure the system by distributing the cash at
the future time $T$ among the components of the risk vector $X$. However,
while the total capital requirement $\sum_{j=1}^{N}Y^{j}$ is determined
today, contrary to (\ref{DynRM3}) the individual allocation $Y^{j}(\omega )$
to institution $j$ does not need to be decided today but in general depends
on the scenario $\omega $ realized at time $T$. As explained in details in 
\cite{bffm}, this total cash amount $\rho (X)$ can be composed today through
the formula 
\begin{equation}
\sum_{j=1}^{N}a^{j}(X)=\rho (X),  \label{DynRMaa}
\end{equation}%
where each cash amount $a^{j}(X)\in \mathbb{R}$ can be interpreted as a%
\textit{\ risk allocation} of bank $j$. The exact formula for the risk
allocation $a^{j}(X)$ will be introduced later in (\ref{DynRMa}).

We remark that by selecting $\mathcal{C}=\mathbb{\ R}^{N}$ in (\ref{DynRMSRM}%
), one recovers the deterministic case (\ref{DynRM3}); while when $\mathcal{C%
}=\mathcal{C}_{\mathbb{R}}$ no further requirements are imposed on the set
of feasible allocations.

Under minimal simple properties on the sets $\mathcal{C}\subseteq \mathcal{C}%
_{\mathbb{R}},$\ $\mathbb{A\subseteq }L^{0}(\Omega ,\mathcal{F},\mathbb{P})$
and on the aggregator $U:\mathbb{R}^{N}\rightarrow \mathbb{R}$ the Systemic
Risk Measures in (\ref{DynRMSRM}) satisfy the key properties of: (i)
decreasing monotonicity, (ii) convexity, (iii) systemic cash additivity:%
\begin{equation*}
\rho (X\mathbf{+}c)=\rho (X)-\sum_{j=1}^{N}c^{j}\text{\qquad for all }c\in {%
\mathbb{R}}^{N}\text{ and }X\in L_{\mathcal{F}}.
\end{equation*}

\paragraph{Shortfall Systemic Risk Measures.}

A special, relevant case of Systemic Risk Measures of the form (\ref%
{DynRMSRM}) \textquotedblleft \textit{first allocate and then aggregate,
with scenario dependent allocation}\textquotedblright\ is given by the class
of Shortfall Systemic Risk Measures, where the acceptance set has the form $%
\mathbb{A}=\{Z\in L^{1}(\Omega ,\mathcal{F},\mathbb{P})\mid \,\mathbb{E}_{%
\mathbb{P}}[Z]\geq B\}$ for a given constant $B\in \mathbb{R}$, namely:%
\begin{equation}
\rho (X):=\inf \left\{ \sum_{j=1}^{N}Y^{j}\mid Y\in \mathcal{C}\text{, }%
\mathbb{E}_{\mathbb{P}}\left[ U(X+Y)\right] \geq B\right\} .  \label{DynRMRM}
\end{equation}%
For the financial motivation behind these choices and for a detailed study
of this class of measures, we refer to \cite{BFFMB} and \cite{bffm} when $%
\mathcal{C}\subseteq \mathcal{C}_{\mathbb{R}}$, and to Armenti et al. (2018) 
\cite{Drapeau} for the analysis of such Risk Measures in the special case $%
\mathcal{C}=\mathbb{\ R}^{N},$ i.e. when only deterministic allocations are
allowed.

The choice of the aggregation functions $U:\mathbb{R}^{N}\rightarrow \mathbb{%
R}$ is also a key ingredient in the construction of $\rho $ and we refer to
Acharia et al. (2017) \cite{Acharya}, Adrian and Brunnermeier (2016) \cite%
{AdrianBrunnermeier16}, Huang and Zhou (2009) \cite{HuangZhou09}, Lehar
(2005) \cite{Lehar}, Brunnermeier and Cheridito (2019) \cite%
{BrunnermeierCheridito}, Biagini et al. (2019) \cite{BFFMB}, and (2020) \cite%
{bffm} for the many examples of aggregators adopted in literature. In order
to obtain more specific and significant properties of $\rho ,$\ \cite{bffm}
selected the aggregator 
\begin{equation}
U(x)=\sum_{j=1}^{N}u_{j}(x^{j}),\quad x\in \mathbb{R}^{N},  \label{DynRMUsum}
\end{equation}%
for strictly concave increasing utility functions $u_{j}:\mathbb{R}%
\rightarrow \mathbb{R},$ for each $j=1,\dots ,N$.

Systemic Risk Measures can be applied not only to determine the overall risk 
$\rho (X)$ of the system, but also to establish the riskiness of each
individual financial institution. As explained in \cite{bffm} it is possible
to determine the risk allocations $a^{j}(X)\in \mathbb{R}$ of each bank $\ j$
that satisfy (\ref{DynRMaa}) and additional meaningful properties. It was
there shown that, with the choice (\ref{DynRMUsum}), a \textit{fair risk
allocation} of bank $j$ is given by: 
\begin{equation}
a^{j}(X):=\mathbb{E}_{\mathbb{Q}^{j}(X)}\left[ Y^{j}(X)\right] \text{,\quad }%
j=1,\dots ,N,  \label{DynRMa}
\end{equation}%
where the vector $Y(X)$ is the optimizer in (\ref{DynRMRM}) and $\mathbb{Q}%
(X)=[\mathbb{Q}^{1}(X),...,\mathbb{Q}^{N}(X)]$ is the vector of probability
measures that optimizes the dual problem associated to $\rho (X)$.

In this paper we will adopt the generalization of the aggregation function (%
\ref{DynRMUsum}) defined by 
\begin{equation}
U(x)=\sum_{j=1}^{N}u_{j}(x_{j})+\Lambda (x),\quad x\in \mathbb{R}^{N},
\label{DynRMU}
\end{equation}%
where the multivariate aggregator $\Lambda :\mathbb{R}^{N}\rightarrow 
\mathbb{R}$ is concave and increasing (not necessarily in a strict sense).
Thus the selection $\Lambda =0$ is possible and hence, in this case, (\ref%
{DynRMU}) reduces to (\ref{DynRMUsum}). The term $\Lambda $ allows
additionally for modeling interdependence among agents also from the point
of view of the preferences.

Just to mention a few examples (see also the examples in \cite{DF19}), any
of the following multivariate utility functions satisfy our assumptions: 
\begin{equation}
U(x):=\sum_{j=1}^{N}u_{j}(x^{j})+u\left( \sum_{j=1}^{N}\beta
_{j}x^{j}\right) ,\text{\quad with }\text{ }\beta _{j}\geq 0\text{, for all }%
j,  \notag
\end{equation}%
where $u:\mathbb{R\rightarrow R}$, for some $p>1,$ is any one of the
following functions:%
\begin{equation*}
u(x):=1-\exp \left( -px\right) ;\text{\quad }u(x):=\,%
\begin{cases}
\,p\,\frac{x}{x+1} & \,\,\,x\geq 0 \\ 
\,1-\left\vert x-1\right\vert ^{p} & \,\,\,x<0%
\end{cases}%
;\text{\quad }u(x):=\,%
\begin{cases}
\,p\,\arctan (x) & \,\,\,x\geq 0 \\ 
\,1-\left\vert x-1\right\vert ^{p} & \,\,\,x<0%
\end{cases}%
\end{equation*}%
and $u_{1},\dots ,u_{N}$ are exponential utility functions ($%
u_{j}(x^{j})=1-\exp {(-\alpha _{j}x^{j})},\,\alpha _{j}>0$) for any choice
of $u$ as above. As shown in this paper, a fairness property for the risk
allocation of each bank can be established also in a conditional setting and
for the aggregator expressed by (\ref{DynRMU}).


\subsection{Conditional Systemic Risk Measures}

\label{DynRMsectcondsystintro} The temporal setting in the approaches
described above is static, meaning that the Risk Measures do not allow for
possible dynamic elements, such as additional information, or the
possibility of risk monitoring in continuous time, or the possibility of
intermediate payoffs and valuation. In order to model \ the conditional
setting we then assume that $\mathcal{G}\subseteq \mathcal{F} $ is a sub- $%
\sigma $-algebra of $\mathcal{F}$ and we consider Risk Measures $\rho _{%
\mathcal{G}}$ with range in $L^{0}(\Omega ,\mathcal{G},\mathbb{P})$ and
interpret $\rho _{\mathcal{G}}\left( X\right) $ as the risk of the whole
system $X$ given the information $\mathcal{G}$.

Conditional Risk Measures have mostly been studied in the framework of 
\textit{univariate }dynamic Risk Measures, where one adjusts the risk
measurement in response to the flow of information that is revealed when
time elapses. The conditional coherent case was treated in Riedel (2004) 
\cite{Riedel04}. Detlefsen and Scandolo (2006) \cite{DetlefsenScandolo05}
was one of the first contributions in the study of conditional convex Risk
Measures and since then a vast literature appeared. Among the early works on
the topic we refer to Barrieu and El Karoui (2005) \cite{BarrieuElKaroui05},
Tutsch (2008) \cite{Tutsch08}, Weber (2006) \cite{Weber06}. Several results
have been obtained for the case of quasi-convex conditional maps and Risk
Measures, see Frittelli and Maggis (2011) \cite{FrittelliMaggis11}, and
Frittelli and Maggis (2014) \cite{FrittelliMaggis}, \cite{FrittelliMaggis14a}%
. Conditional counterparts to classical static results (e.g. dual
representation and separation properties) have been obtained exploiting the
theory of $L^{0}$-modules. Among the many contributions in this stream of
research we mention Filipovi\'{c} at al. (2009) \cite%
{FilipovicKupperVogelpoth09} and (2012) \cite{FilipovicKupperVogelpoth12},
Drapeau et al. (2016) \cite{DrapeauJamneshanKarliczekKupper16}, Drapeau et
al. (2019) \cite{DrapeauJamneshanKupper19}, Guo (2010) \cite{Guo10} and the
references therein. Overall, the fact that the natural conditional
counterparts hold for static results is not so surprising. The two are
intrinsically related by a Boolean Logic principle. As seen in Carl and
Jamneshan (2018) \cite{CarlJamneshan18}, traditional theorems carry over to
the conditional setup assuming that suitable concatenation properties hold.
Time consistency properties have been considered, in the univariate case, in
Roorda and Schumacher (2007) \cite{RoordaSchumacher07}, (2013) \cite%
{RoordaSchumacher13} and (2016) \cite{RoordaSchumacher16}.

We refer the reader to \cite{FollmerSchied2} Chapter 11 for a good overview
on univariate dynamic Risk Measures. We observe that such a conditional and
dynamic framework generated a florilegium of interesting ramification in
different fields, including the relationships with BSDEs (Barrieu and El
Karoui (2005) \cite{BarrieuElKaroui05}, Rosazza Gianin (2006) \cite%
{Rosazza06}, Bion-Nadal (2008) \cite{BionNadal08}, Delbaen et al. (2011) 
\cite{DelbaenPengRosazza10}) and Non Linear Expectations (Peng (2004) \cite%
{Peng04}).

\bigskip

A conditional Systemic Risk Measure is a map $\rho _{\mathcal{G}}:L_{%
\mathcal{F}}\rightarrow L^{0}(\Omega ,\mathcal{G},\mathbb{P})$ that
associates to a $N$-dimensional risk factor $X\mathbf{\in }L_{\mathcal{F}%
}\subseteq (L^{0}(\Omega ,\mathcal{F},\mathbb{P}))^N$ a $\mathcal{G}$%
-measurable random variable. A conditional Systemic Risk Measure thus models
the risk of a system as new information arises in the course of time. The
study of \ conditional Systemic (multivariate) Risk Measures was initiated
by Hoffmann et al. (2016) \cite{HoffmannMeyerbrandisSvindland16} and (2018) 
\cite{HoffmannMeyerbrandisSvindland18}. 
However, as pointed out in \cite{HoffmannMeyerbrandisSvindland18}, in the
context of multivariate Risk Measures, a second interesting and important
dimension of conditioning arises, besides dynamic conditioning: a risk
measurement, of the $N$-dimensional vector $X$, conditional on some specific
substructure of the system, for example induced by mergers and acquisitions.
In this paper we will not elaborate on this topic and refer the reader to
Follmer (2014) \cite{Follmer14} or Follmer and Kluppelberg (2014) \cite%
{FollmerKluppelberg14} for some details. However, in order to allow both
interpretations, a general $\sigma $-algebra $\mathcal{G}\subseteq \mathcal{F%
}$ will be considered in the sequel.

The papers \cite{HoffmannMeyerbrandisSvindland16}, \cite%
{HoffmannMeyerbrandisSvindland18}, as well as Kromer et al. (2019) \cite%
{KromerOverbeckZilch19}, consider only the conditional extension of (static)
Systemic Risk Measures of the \textquotedblleft first aggregate, then
allocate\textquotedblright\ form expressed by \eqref{DynRM2} and study
related consistency issues. Multivariate/Systemic and set-valued conditional
Risk Measures, and related time consistency aspects, have also been analyzed
in Feinstein and Rudloff (2013) \cite{FeinsteinRudloff13}, (2015) \cite%
{FeinsteinRudloff15a}, (2017) \cite{FeinsteinRudloff17} and (2021) \cite%
{FeinsteinRudloff21}, Tahar and Lepinette (2014) \cite{TaharLepinette14},
Chen and Hu (2018) \cite{ChenHu18}. Although apparently similar, our
approach in the present work is significantly different. Once we clarify our
setup, we will elaborate more on this after Theorem \ref%
{DynRMmainthmrhoinfty} and in Remark \ref{DynRMremarksetvalued2}.

\paragraph{Contribution and outline of the paper.}

Our aim in the present paper is the study of general conditional convex
Systemic Risk Measures and the detailed analysis of Conditional Shortfall
Systemic Risk Measures. Our findings show that most properties of Shortfall
Systemic Risk Measures carry over to the conditional setting, even if the
proofs become more technical, and that a new vector type consistency, with
respect to sub $\sigma $-algebras $\mathcal{H\subseteq G}\subseteq \mathcal{F%
}$, replaces the scalar recursiveness property of univariate Risk Measures.

\bigskip

More precisely, we define axiomatically a\textbf{\ }\textit{Conditional
Systemic Risk Measure} (CSRM) on $L_{\mathcal{F}}$ as a map $\rho _{\mathcal{%
G}}:L_{\mathcal{F}}\rightarrow L^{0}(\Omega ,\mathcal{G},{\mathbb{P}})$
satisfying monotonicity, conditional convexity and the conditional monetary
property (see Definition \ref{DynRMdefdynsistrm}). Our first result (Theorem %
\ref{DynRMcorcashadd}) shows, under fairly general assumptions, that: (i) $%
\rho _{\mathcal{G}}$ admits the conditional dual representation 
\begin{equation}
\rho _{\mathcal{G}}\left( X\right) =\esssup_{{\mathbb{Q}}\in \mathscr{Q}_{%
\mathcal{G}}}\left( \sum_{j=1}^{N}\mathbb{E}_{\mathbb{Q}^{j}}\left[ -X^{j}%
\middle|\mathcal{G}\right] -\alpha ({\mathbb{Q}})\right) ,\text{ }{\mathbb{P}%
}\text{-a.s., for }X\in L_{\mathcal{F}},  \label{DynRMsup}
\end{equation}%
where $\mathscr{Q}_{\mathcal{G}}$, defined in Equation %
\eqref{DynRMequatintrointro1}, is a set of vectors of probability measures
and the penalty $\alpha ({\mathbb{Q}})\in L^{0}(\Omega ,\mathcal{G},{\mathbb{%
P}})$ is defined in Equation \eqref{DynRMequatintro2}; (ii) the supremum in (%
\ref{DynRMsup}) is attained.

We then specialize our analysis by considering the \textit{Conditional
Shortfall Systemic Risk Measure, associated to multivariate utility
functions }$U$\textbf{\ }of the form (\ref{DynRMU}), defined by%
\begin{equation}
\rho _{\mathcal{G}}\left( X\right) :=\essinf\left\{ \sum_{j=1}^{N}Y^{j}\mid
Y\in \mathscr{C}_{\mathcal{G}},\text{ }\mathbb{E}_{\mathbb{P}}\left[ U\left(
X+Y\right) \middle|\mathcal{G}\right] \geq B\right\} ,
\label{DynRMdefrcondintro}
\end{equation}%
where $B$ is now a random variable in $L^{\infty }(\Omega ,\mathcal{G},{%
\mathbb{P}})$ and the set of \ $\mathcal{G}$-admissible allocations is 
\begin{equation}
\mathscr{C}_{\mathcal{G}}\subseteq \left\{ Y\in (L^{1}(\Omega ,\mathcal{F},{%
\mathbb{P}}))^{N}\text{ such that }\sum_{j=1}^{N}Y^{j}\in L^{\infty }(\Omega
,\mathcal{G},{\mathbb{P}})\right\} \,.  \notag
\end{equation}%
Thus, with these definitions that mimic those in \eqref{DynRMdefC} and in %
\eqref{DynRMRM}, the same motivations, mutatis mutandis, explained in the
unconditional setting remain true in the conditional one.

Observe that even for the trivial selection $\mathcal{G}=\left\{ \varnothing
,\Omega \right\} $, for which conditional Risk Measures reduce to static
ones, this paper extends the results in \cite{bffm} to the more general
aggregator $U$ in the form (\ref{DynRMU}).

\bigskip

In Theorem \ref{DynRMmainthmrhoinfty} we prove the main properties of the
Conditional Shortfall Systemic Risk Measure $\rho _{\mathcal{G}}$ and, in
particular, we show that (i) $\rho _{\mathcal{G}}$ is continuous from above
and from below; (ii) the essential infimum in \eqref{DynRMdefrcondintro} is
attained by a vector $Y(\mathcal{G},X)=[Y^{1}(\mathcal{G},X),...,Y^{N}(%
\mathcal{G},X)]\in \mathscr{C}_{\mathcal{G}}$; (iii) $\rho _{\mathcal{G}}$
admits the dual representation described in \eqref{DynRMeqdualreprshorfall};
(iv) the supremum in the dual formulation \eqref{DynRMeqdualreprshorfall} of 
$\rho _{\mathcal{G}}$ is attained by a vector ${\mathbb{Q}}(\mathcal{G},X)=[{%
\mathbb{Q}}^{1}(\mathcal{G},X),...,\allowbreak {\mathbb{Q}}^{N}(\mathcal{G}%
,X)]$ of probability measures satisfying:%
\begin{equation*}
\sum_{j=1}^{N}\mathbb{E}_{\mathbb{Q}^{j}(\mathcal{G},X)}\left[ Y^{j}(%
\mathcal{G},X)\middle| \mathcal{G}\right] =\sum_{j=1}^{N}Y^{j}(\mathcal{G}%
,X)=\rho _{\mathcal{G}}\left( X\right) \,\,\,\,{\mathbb{P}}-\text{a.s.}\,
\end{equation*}
In the same spirit of \cite{bffm}, we will then interpret the quantity 
\begin{equation*}
a^{j}(\mathcal{G},X):=\mathbb{E}_{\mathbb{Q}^{j}(\mathcal{G},X)}\left[ Y^{j}(%
\mathcal{G},X)\middle| \mathcal{G}\right]
\end{equation*}
as a \textit{fair risk allocation of institution} $j$, \textit{given} $%
\mathcal{G}.$

\bigskip

Section 6 is then devoted to the particular case of exponential utility
functions $u_{j}(x^{j}):=-e^{-\alpha _{j}x^{j}},$ $\alpha _{j}>0,$ $%
j=1,...,N,$ and with $\Lambda =0$. As in the static case (see\ \cite{bffm}),
also in the conditional case it is possible to find the explicit formulas
for: (i) the value of the Conditional Shortfall Systemic Risk Measure\textit{%
\ }$\rho _{\mathcal{G}}\left( X\right) ;$ (ii) the optimizer $Y(\mathcal{G}%
,X)$ in \eqref{DynRMdefrcondintro} of $\rho _{\mathcal{G}}\left( X\right) $;
(iii) the vector $\mathbb{Q}(\mathcal{G},X)$ of probability measures that
attains the supremum in the dual formulation. Such formulas provide a
conditional counterpart to the results in \cite{bffm}.

Finally, for sub $\sigma $-algebras $\mathcal{H\subseteq G}\subseteq 
\mathcal{F}$ we prove a particular consistency property, which does not have
a counterpart in the univariate case. Indeed, a recursive property of the
type $\rho _{\mathcal{H}}(-\rho _{\mathcal{G}}\left( X\right) )=\rho _{%
\mathcal{H}}\left( X\right) $ is not even well defined in the systemic
setting, as $\rho _{\mathcal{G}}\left( X\right) $ is a random variable but
the argument of $\rho _{\mathcal{H}}$ is a vector of random variables.
However, we explain that consistency properties may be well defined for: (i)
the vector optimizers $Y(\mathcal{G},X)$ of $\rho _{\mathcal{G}}\left(
X\right) $ and $Y(\mathcal{H},-Y(\mathcal{G},X))$ of $\rho _{\mathcal{H}%
}\left( -Y(\mathcal{G},X)\right) $; (ii) the fair risk allocations vectors $%
[a(\mathcal{G},X)]_{j}:=[\mathbb{E}_{{\mathbb{Q}}^{j}(\mathcal{G},X)}[Y^{j}(%
\mathcal{G},X)|\mathcal{G}]]_{j}$ of $\rho _{\mathcal{G}}\left( X\right) $
and $a(\mathcal{H},-a(\mathcal{G},X))$ of $\rho _{\mathcal{H}}\left( -a(%
\mathcal{G},X)\right) $. The consistency properties are shown in %
\eqref{DynRMconsisty} and \eqref{DynRMconsista} and proven in Theorem \ref%
{DynRMthmconsistency} for the entropic Conditional Systemic Risk Measure. In
Remark \ref{DynRMremarksetvalued2} we compare this consistency properties
with the ones for the set-valued case from \cite{FeinsteinRudloff15a} and 
\cite{ChenHu18}.

\bigskip

In a final Section we elaborate on the concept of a Systemic Optimal Risk
Transfer Equilibrium, a notion introduced in \cite{BDFFM}. We defer the
interested reader to \cite{BDFFM} for the economic motivation and for the
applications of this equilibrium. Based on the results on the Conditional
Shortfall Systemic Risk Measure developed in Section 5, we are able to
provide in Theorem \ref{DynRMthmallocissorte} a direct extension of this
equilibrium in the conditional setting. At the same time, we show that the
optimal allocations for Shortfall Systemic Risk Measures, in both the static
and dynamic cases, admit an interpretation in the sense of a suitably
defined equilibrium. By the choice of the trivial $\mathcal{G}%
=\{\emptyset,\Omega\}$ our findings in this last section cover the static
setup, and provide an explicit link between \cite{BDFFM} and \cite{bffm}.

\section{Static Systemic Risk Measures\label{Sec2}}

We fix a probability space $(\Omega ,\mathcal{F},{\mathbb{P}})$. We take two
vector subspaces $L_{\mathcal{F}},\,L^{\ast }$ of $(L^{1}(\Omega ,\mathcal{F}%
,{\mathbb{P}}))^{N}$, for $N\geq 1$.

\begin{definition}
A functional $\rho _{0}:L_{\mathcal{F}}\rightarrow {\mathbb{R}}$ will be
called \ a \textbf{(Static) Convex Systemic Risk Measure} if it satisfies: 
\textbf{Monotonicity}, that is $X\leq Y\text{ componentwise }\,\,\Rightarrow
\,\,\rho _{0}(X)\geq \rho _{0}(Y)$ , \textbf{Convexity}, that is $0\leq
\lambda \leq 1\,\,\Rightarrow \rho _{0}(\lambda X+(1-\lambda )Y)\leq \lambda
\rho _{0}(X)+(1-\lambda )\rho _{0}(Y)$ and the \textbf{Monetary property}
(or Cash Additivity), that is $X\in L_{\mathcal{F}},$ $c\in {\mathbb{R}}%
^{N}\,\,\,\Rightarrow \,\,\,\rho _{0}(X+c)=\rho _{0}({X)}%
-\sum_{j=1}^{N}c^{j}.$
\end{definition}

For $L_{\mathcal{F}}=(L^{\infty }(\Omega ,\mathcal{F},{\mathbb{P}}))^{N}$,
we also say that $\rho _{0}:(L^{\infty }(\Omega ,\mathcal{F},{\mathbb{P}}%
))^{N}\rightarrow {\mathbb{R}}$ is continuous from below (resp. from above)
if for any sequence $(X_{n})_{n}$ such that $X_{n}\in (L^{\infty }(\Omega ,%
\mathcal{F},{\mathbb{P}}))^{N}$ and $X_{n}\uparrow _{n}X\in (L^{\infty
}(\Omega ,\mathcal{F},{\mathbb{P}}))^{N}$ (resp. $X_{n}\downarrow _{n}X\in
(L^{\infty }(\Omega ,\mathcal{F},{\mathbb{P}}))^{N}$) we have $\rho
_{0}(X)=\lim_{n}\rho _{0}(X_{n})$. If ${\mathbb{Q}}=[{\mathbb{Q}}^{1},\dots ,%
{\mathbb{Q}}^{N}]$ is a vector of probability measures on $(\Omega ,\mathcal{%
F)}$, we write ${\mathbb{Q\ll P}}$ for ${\mathbb{Q}}^{j}\ll {\mathbb{P}}%
\,\,\forall \,j=1,\dots ,N$ and use the notation $\frac{\mathrm{d}{\mathbb{Q}%
}}{\mathrm{d}{\mathbb{P}}}:=\left[ \frac{\mathrm{d}{\mathbb{Q}}^{1}}{\mathrm{%
d}{\mathbb{P}}},\dots ,\frac{\mathrm{d}{\mathbb{Q}}^{N}}{\mathrm{d}{\mathbb{P%
}}}\right] .$ We set 
\begin{equation*}
\mathscr{Q}:=\left\{ {\mathbb{Q}}=[{\mathbb{Q}}^{1},\dots ,{\mathbb{Q}}^{N}]{%
\mathbb{\ll P}}\mid \frac{\mathrm{d}{\mathbb{Q}}}{\mathrm{d}{\mathbb{P}}}\in
L^{\ast }\right\} .
\end{equation*}

\begin{definition}
\label{DynRMdefnicelyrepr} We say that a (Static) Convex Systemic Risk
Measure $\rho _{0}:L_{\mathcal{F}}\rightarrow {\mathbb{R}}$ is \textbf{%
nicely representable (}with respect to the\textbf{\ }$\sigma (L_{\mathcal{F}%
},L^{\ast })$\textbf{\ }topology\textbf{) } if 
\begin{equation}
\rho _{0}(X)=\max_{{\mathbb{Q}}\in \mathscr{Q}}\left( \sum_{j=1}^{N}\mathbb{E%
}_{{\mathbb{Q}}_{j}}\left[ -X^{j}\right] -\alpha _{0}({\mathbb{Q}})\right)
,\quad X\in L_{\mathcal{F}}  \label{DynRMdualrep}
\end{equation}%
where%
\begin{equation}
\alpha _{0}({\mathbb{Q}}):=\rho _{0}^{\ast }\left( -\frac{\mathrm{d}{\mathbb{%
Q}}}{\mathrm{d}{\mathbb{P}}}\right) =\sup_{X\in L_{\mathcal{F}}}\left(
\sum_{j=1}^{N}\mathbb{E}_{{\mathbb{Q}}_{j}}\left[ -X^{j}\right] -\rho
_{0}(X)\right) ,\quad {\mathbb{Q}}\in \mathscr{Q}\text{.}  \label{rho**}
\end{equation}
and $\rho_0^*$ is the convex conjugate of $\rho_0$.
\end{definition}

\begin{remark}
\label{DynRMremsufficient}For univariate ($N=1$) Convex Risk Measures, there
are well known sufficient conditions for nice representability, which can be
split in two categories: either continuity conditions (order upper
semicontinuity or continuity from below); or structural properties of the
vector spaces. In particular:

\begin{enumerate}
\item If $L_{\mathcal{F}}=L^{\infty }(\Omega ,\mathcal{F},{\mathbb{P}}),$ $%
L^{\ast }=L^{1}(\Omega ,\mathcal{F},{\mathbb{P}})$ and if $\rho
_{0}:L^{\infty }(\Omega ,\mathcal{F},{\mathbb{P}})\rightarrow {\mathbb{R}}$
is continuous from below then it is nicely representable (see \cite{bfnam}
Lemma 7 or \cite{FollmerSchied2} Corollary 4.35). This in turns implies $%
\sigma (L^{\infty }(\Omega ,\mathcal{F},{\mathbb{P}}),L^{1}(\Omega ,\mathcal{%
F},{\mathbb{P}})$-lower semicontinuity and continuity from above.

\item If $L_{\mathcal{F}}=L^{p}(\Omega ,\mathcal{F},{\mathbb{P}})$, $L^{\ast
}=L^{q}(\Omega ,\mathcal{F},{\mathbb{P}})$ with $p\in \lbrack 1,+\infty )$
and $q$ the conjugate exponent, or if $L_{\mathcal{F}}=M^{\Phi }(\Omega ,%
\mathcal{F},{\mathbb{P}})\neq \emptyset $, $L^{\ast }=L^{\Phi ^{\ast
}}(\Omega ,\mathcal{F},{\mathbb{P}})$ (see Section \ref{DynRMsecintegr} and
Equation \eqref{DynRMMphi} for the definitions), then any univariate Convex
Systemic Risk Measure $\rho _{0}:L_{\mathcal{F}}\rightarrow {\mathbb{R}}$ is
nicely representable, due to the Extended Namioka-Klee Theorem in \cite%
{bfnam}.
\end{enumerate}
\end{remark}

We will now extend one dimensional classical results to our systemic setup.
Only slight modifications are needed in the proofs, but we add them in the
Appendix, Section \ref{DynRMsecproofstatic} for the sake of completeness.

\begin{theorem}
\label{DynRMthmdual}i) If $\rho _{0}:(L^{\infty }(\Omega ,\mathcal{F},{%
\mathbb{P}}))^{N}\rightarrow {\mathbb{R}}$ is a (Static) Convex Systemic
Risk Measure continuous from below then it is nicely representable with $%
L^{\ast }=(L^{1}(\Omega ,\mathcal{F},{\mathbb{P}}))^{N}$ and therefore it is 
$\sigma ((L^{\infty }(\mathcal{F}))^{N},(L^{1}(\mathcal{F}))^{N})$\textbf{\ }%
lower semicontinuous and continuous from above. ii) If $L_{\mathcal{F}%
}\subseteq (L^{1}(\Omega ,\mathcal{F},{\mathbb{P}}))^{N}$ is a Banach
lattice with order continuous topology and if $L^{\ast }\subseteq
(L^{1}(\Omega ,\mathcal{F},{\mathbb{P}}))^{N}$ is its topological dual space
then any (Static) Convex Systemic Risk Measure $\rho _{0}:L_{\mathcal{F}%
}\rightarrow {\mathbb{R}}$ is nicely representable.
\end{theorem}

%
Obviously item ii) in the theorem covers the multidimensional versions of
the cases described in item 2 in Remark \ref{DynRMremsufficient}.

\section{Conditional Systemic Risk Measures\label{Sec3}}

We now present the conditional framework which acts as a counterpart to the
static one presented before. We introduce the conditional versions of usual
properties of Systemic Risk Measures (convexity, additivity), provide the
general definition of Conditional Systemic Risk Measure and related
continuity concepts we will use in the following. We also present a general
duality result in Section \ref{DynRMsecdualrepgeneral}.

\subsection{Setup and notation\label{Sec31}}

We let $\mathcal{G}\subseteq \mathcal{F}$ be a sub $\sigma $-algebra and
recall that $L_{\mathcal{F}}\subseteq L^{1}(\Omega ,\mathcal{F},{\mathbb{P}}%
) $. Throughout all the paper we will often need to change underlying $%
\sigma $-algebras. In order to avoid unnecessarily heavy notation, we will
explicitly specify the one or the other only when some confusion might
arise. For example, $L^{\infty }(\mathcal{F})$, $L^{\infty }(\mathcal{G})$
stand for $L^{\infty }(\Omega ,\mathcal{F},{\mathbb{P}})$ and $L^{\infty
}(\Omega ,\mathcal{G},{\mathbb{P}})$ respectively. Unless differently
stated, all inequalities between random variables hold $\mathbb{P}$-a.s.. If 
$A\in \mathcal{F}$ we write $A^{c}$ for its complement.

\begin{remark}
In the following (MON) and (DOM) are references to Monotone and Dominated
convergence Theorem respectively. (cMON) and (cDOM) refer to their
conditional counterparts. We will use without explicit mention the
properties of essential suprema (and essential infima) collected in
Proposition \ref{DynRMrefesssup}.
\end{remark}

\begin{definition}
\label{DynRMdefdecomp}$L_{\mathcal{F}}$ is $\mathcal{G}$-decomposable if $%
(L^{\infty }(\mathcal{F}))^{N}\subseteq L_{\mathcal{F}}$ and if for any $%
Y\in (L^{\infty }(\mathcal{G}))^{N}$ and $X\in L_{\mathcal{F}}$ the random
vector $Z$ defined as $Z^{j}=X^{j}Y^{j},\,j=1,\dots ,N,$ belongs to $L_{%
\mathcal{F}}$.
\end{definition}

\begin{remark}
\label{DynRMremdecom} Observe that by decomposability whenever $A\in 
\mathcal{G}$ and $X,Y\in L_{\mathcal{F}}$ we also have $X1_{A}+Y1_{A^{c}}\in
L_{\mathcal{F}}$. We stress the fact that $\mathcal{G}$-decomposability is a
very mild requirement, which is clearly satisfied for example if $L_{%
\mathcal{F}}=L^{p}$ for some $p\in \lbrack 1,+\infty ]$ or $L_{\mathcal{F}}$
is an Orlicz Space (see Section \ref{DynRMsecintegr} and \cite{DF19} Section
2.1).
\end{remark}

\begin{definition}
A subset $\mathcal{C}\subseteq L_\mathcal{F}$ is:

\begin{itemize}
\item $\mathcal{G}-$\textbf{conditionally convex} if $\forall$ $\lambda \in
L^{0}(\mathcal{G}),\,0\leq \lambda \leq 1$ and $\forall$ $X,Y\in \mathcal{C}$%
, $\lambda X+(1-\lambda )Y\in \mathcal{C}$.

\item \textbf{a} $\mathcal{G}-$\textbf{conditional cone} if $\forall$ $0\leq
\lambda \in L^{\infty }(\mathcal{G})$ and $\forall$ $X\in \mathcal{C}$, $%
\lambda X\in \mathcal{C}$.

\item \textbf{closed under }$\mathcal{G}-$ \textbf{truncation} if $\forall$ $%
Y\in \mathcal{C}$ there exists $k_{Y}\in \mathbb{N}$ and a $Z_{Y}\in
(L^{\infty }(\mathcal{F}))^{N}$ such that $\sum_{j=1}^{N}Z_{Y}^{j}=%
\sum_{j=1}^{N}Y^{j}$ and $\forall$ $k\geq k_{Y},\,k\in \mathbb{N}$ 
\begin{equation}
Y_{(k)}:=Y1_{\bigcap_{j}\{\left\vert Y^{j}\right\vert <
k\}}+Z_{Y}1_{\bigcup_{j}\{\left\vert Y^{j}\right\vert \geq k\}}\in \mathcal{C%
}.  \label{DynRMdeftruncated}
\end{equation}
\end{itemize}

We will explicitly specify the $\sigma$-algebra ($\mathcal{G}$ in the
notation above) with respect to which the properties are required to hold
only when some confusion might arise.
\end{definition}

\begin{definition}
\label{DynRMdefdynsistrm}

A map $\rho _{\mathcal{G}}:L_{\mathcal{F}}\rightarrow L^{0}(\mathcal{G})$ is
a \textbf{Conditional Systemic Risk Measure} (CSRM) if it satisfies for all $%
X,Y\in L_{\mathcal{F}}$

\begin{enumerate}
\item \textbf{Monotonicity}, that is 
\begin{equation}
X\leq Y\text{ componentwise}\,\,\Rightarrow \,\,\rho _{\mathcal{G}}\left(
X\right) \geq \rho _{\mathcal{G}}\left( Y\right)  \label{DynRMcmon}
\end{equation}

\item \textbf{Conditional Convexity}, that is 
\begin{equation}
\rho _{\mathcal{G}}\left( \lambda X+(1-\lambda )Y\right) \leq \lambda \rho _{%
\mathcal{G}}\left( X\right) +(1-\lambda )\rho _{\mathcal{G}}\left( Y\right)
\quad \text{for all }\,0\leq \lambda \leq 1,\text{ }\lambda \in L^{\infty }(%
\mathcal{G})\,\,  \label{DynRMccnvx}
\end{equation}

\item \textbf{Conditional $\mathcal{G}$-Additivity} (or the conditional
monetary property), that is 
\begin{equation}
\rho _{\mathcal{G}}\left( X+Y\right) =\rho _{\mathcal{G}}\left( X\right)
-\sum_{j=1}^{N}Y^{j}\quad \text{if }Y\in \left( L^{\infty }(\mathcal{G}%
)\right) ^{N}\cap L_{\mathcal{F}}  \label{DynRMccmt}
\end{equation}
\end{enumerate}
\end{definition}

One may easily show, as in the one dimensional case, that a map $\rho _{%
\mathcal{G}}:\left( L^{\infty }(\mathcal{F})\right) ^{N}\rightarrow L^{0}(%
\mathcal{G})$ satisfying $\rho _{\mathcal{G}}(0)\in L^{\infty }(\mathcal{G})$%
, monotonicity and the conditional monetary property has range in $L^{\infty
}(\mathcal{G})$ and $|\rho _{\mathcal{G}}(X)-\rho _{\mathcal{G}}(0)|\leq
\sum_{j=1}^{N}\left\Vert X^{j}\right\Vert _{\infty }$ ${\mathbb{P}}$-${a.s..}%
\,$ For the Conditional Shortfall Systemic Risk Measure in Section \ref{Sec5}
we prove first that $\rho _{\mathcal{G}}(X)\in L^{\infty }(\mathcal{G})$ for
all $X\in \left( L^{\infty }(\mathcal{F})\right) ^{N}$ and then show all the
properties of the Risk Measure. When $L_{\mathcal{F}}\neq \left( L^{\infty }(%
\mathcal{F})\right) ^{N}$, in order to apply the scalarization procedure $%
\rho _{0}(\cdot )=\mathbb{E}_{\mathbb{P}}[\rho _{\mathcal{G}}(\cdot )]$ we
will need, in Theorem \ref{DynRMcorcashadd}, the assumption that the range
of $\rho _{\mathcal{G}}$ is contained in $L^{1}(\mathcal{G})$, and for this
we will require that $\rho _{\mathcal{G}}:L_{\mathcal{F}}\,\rightarrow L^{0}(%
\mathcal{G})\cap L_{\mathcal{F}}$.

\begin{definition}
\label{DynRMdefcont} For the particular choice $L_{\mathcal{F}}=(L^{\infty }(%
\mathcal{F}))^{N}$ we say that a CSRM $\rho _{\mathcal{G}}:(L^{\infty }(%
\mathcal{F}))^{N}\rightarrow L^{0}(\mathcal{G})$ is

\begin{itemize}
\item \textbf{continuous from above} if for any sequence $%
(X_{n})_{n}\subseteq (L^{\infty }(\mathcal{F}))^{N}$ and $X\in (L^{\infty }(%
\mathcal{F}))^{N}$ such that for each $j=1,\dots ,N$ $X_{n}^{j}\downarrow
_{n}X^{j}$ we have $\rho _{\mathcal{G}}\left( X_{n}\right) \uparrow _{n}\rho
_{\mathcal{G}}\left( X\right) \,\,{\mathbb{P}}-$a.s.

\item \textbf{continuous from below} if for any sequence $%
(X_{n})_{n}\subseteq (L^{\infty }(\mathcal{F}))^{N}$ and $X\in (L^{\infty }(%
\mathcal{F}))^{N}$ such that for each $j=1,\dots ,N$ $X_{n}^{j}\uparrow
_{n}X^{j}$ we have $\rho _{\mathcal{G}}\left( X_{n}\right) \downarrow
_{n}\rho _{\mathcal{G}}\left( X\right) \,\,{\mathbb{P}}-$a.s.

\item \textbf{Lebesgue continuous} (or that $\rho _{\mathcal{G}}\left( \cdot
\right) $ has the Lebesgue property) if for any sequence $%
(X_{n})_{n}\subseteq (L^{\infty }(\mathcal{F}))^{N}$ and $X\in (L^{\infty }(%
\mathcal{F}))^{N}$ such that for each $j=1,\dots ,N$ $\sup_{n}\left\Vert
X_{n}^{j}\right\Vert _{\infty }<+\infty $ and $X_{n}^{j}\rightarrow
_{n}X^{j}\,\,{\mathbb{P}}-$a.s. we have $\rho _{\mathcal{G}}\left(
X_{n}\right) \rightarrow _{n}\rho _{\mathcal{G}}\left( X\right) \,\,{\mathbb{%
P}}-$a.s.
\end{itemize}
\end{definition}

\begin{remark}
\label{DynRMLebesgue} Observe that continuity from above and continuity from
below of a CSRM $\rho _{\mathcal{G}}$ yield the Lebesgue property, by simple
computations similar to the univariate case. 
\end{remark}

\subsection{Dual representation of Conditional Systemic Risk Measures}

\label{DynRMsecdualrepgeneral} This section follows the lines of the
scalarization procedure in \cite{DetlefsenScandolo05} and \cite%
{OrihuelaZapata} and we defer the proofs to the Appendix, Section \ref%
{DynRMsecauxildualreprcond}. 

Define the following subset of $\mathscr{Q}$: 
\begin{equation}
\mathscr{Q}_{\mathcal{G}}:=\left\{ {\mathbb{Q}}=[{\mathbb{Q}}^{1},\dots ,{%
\mathbb{Q}}^{N}]\ll {\mathbb{P}}\mid \,\frac{\mathrm{d}{\mathbb{Q}}^{j}}{%
\mathrm{d}{\mathbb{P}}}\in L^{\ast },\,\mathbb{E}_{\mathbb{P}}\left[ \frac{%
\mathrm{d}{\mathbb{Q}}^{j}}{\mathrm{d}{\mathbb{P}}}\middle|\mathcal{G}\right]
=1\quad \forall \,j=1,\dots ,N\right\} ,  \label{DynRMequatintrointro1}
\end{equation}%
and set 
\begin{equation}
\alpha ({\mathbb{Q}}):=\esssup_{X\in L_{\mathcal{F}},\rho _{\mathcal{G}%
}\left( X\right) \leq 0}\sum_{j=1}^{N}\mathbb{E}_{\mathbb{Q}^{j}}\left[ -X^j%
\middle|\mathcal{G}\right] ,\quad {\mathbb{Q\in }}\mathscr{Q}_{\mathcal{G}}.
\label{DynRMequatintro2}
\end{equation}

\begin{remark}
\label{DynRMWellDefined} Observe that each component ${\mathbb{Q}}^{j}$ of
elements in $\mathscr{Q}_{\mathcal{G}}$ satisfies $\mathbb{E}_{\mathbb{P}}%
\left[ \frac{\mathrm{d}{\mathbb{Q}}^{j}}{\mathrm{d}{\mathbb{P}}}\middle|%
\mathcal{G}\right] =1$. In this case, ${\mathbb{Q}}^{j}=\mathbb{P}$ on $%
\mathcal{G}$ and $\mathbb{E}_{\mathbb{Q}^{j}}\left[ X^{j}\middle|\mathcal{G}%
\right] =\mathbb{E}_{\mathbb{P}}\left[ \frac{\mathrm{d}{\mathbb{Q}}^{j}}{%
\mathrm{d}{\mathbb{P}}}X^{j}\middle|\mathcal{G}\right] $ is defined not only 
${\mathbb{Q}}^{j}$-a.s. but also ${\mathbb{P}}$-a.s.. Hence (\ref%
{DynRMequatintro2}) and (\ref{DynRMdualreprgeneralcor})\ are well defined ${%
\mathbb{P}}$-a.s..
\end{remark}

\begin{theorem}
\label{DynRMcorcashadd} Suppose that $L_{\mathcal{F}}$ is $\mathcal{G}$%
-decomposable and that for any $X\in L_{\mathcal{F}},\,Z\in L^{\ast }$ we
have $\sum_{j=1}^{N}X^{j}Z^{j}\in L^{1}(\mathcal{F})$. Let $\rho _{\mathcal{G%
}}:L_{\mathcal{F}}\,\rightarrow L^{0}(\mathcal{G})\cap L_{\mathcal{F}}$
satisfy monotonicity, conditional convexity and conditional additivity (that
is, let $\rho _{\mathcal{G}}$ be a CSRM) and let $\rho _{0}(\cdot ):=\mathbb{%
E}_{\mathbb{P}}\left[ \rho _{\mathcal{G}}\left( \cdot \right) \right] :L_{%
\mathcal{F}}\rightarrow {\mathbb{R}}$ be nicely representable \textbf{(}with
respect to the\textbf{\ }$\sigma (L_{\mathcal{F}},L^{\ast })$\textbf{\ }%
topology\textbf{)}. Then $\rho _{\mathcal{G}}$ admits the following dual
representation:%
\begin{equation}
\rho _{\mathcal{G}}\left( X\right) =\esssup_{{\mathbb{Q}}\in \mathscr{Q}_{%
\mathcal{G}}}\left( \sum_{j=1}^{N}\mathbb{E}_{\mathbb{Q}^{j}}\left[ -X^{j}%
\middle|\mathcal{G}\right] -\alpha ({\mathbb{Q}})\right), \quad X\in L_{%
\mathcal{F}}.  \label{DynRMdualreprgeneralcor}
\end{equation}%
Furthermore, there exists $\widehat{{\mathbb{Q}}}\in \mathscr{Q}_{\mathcal{G}%
}$ such that $\rho _{\mathcal{G}}\left( X\right) =\sum_{j=1}^{N}\mathbb{E}_{%
\widehat{{\mathbb{Q}}}^{j}}\left[ -X^{j}\middle|\mathcal{G}\right] -\alpha (%
\widehat{{\mathbb{Q}}})$ and $\rho _{\mathcal{G}}$ is continuous from above.
\end{theorem}

\section{Multivariate utility functions}

\label{DynRMSecmultiut} We will now turn our attention to Conditional
Systemic Risk Measures of shortfall type, which consider as eligible for
securing the system those terminal time allocations which produce a utility
(for the system) above a given threshold. Before formulating the precise
definition of such a risk measurement regime, we need to specify a model for
preferences of the agents in the system. To this end, we exploit
multivariate utility functions. This allows for modeling the fact that a
single agent's preferences might depend on the actions of the others.

\begin{definition}
\label{DynRMdefmultiutil} We say that $U:{\mathbb{R}}^{N}\rightarrow {%
\mathbb{R}}$ is a \textbf{multivariate utility function} if it is strictly
concave and increasing with respect to the partial componentwise order. When 
$N=1$ we will use the term univariate utility function instead.
\end{definition}

The following assumption, as well as Standing Assumption II below, holds
true throughout the paper without further mention.

\begin{standingassumptionI*}
We will consider multivariate utility functions in the form 
\begin{equation}
U(x):=\sum_{j=1}^{N}u_{j}(x^{j})+\Lambda (x)  \label{DynRMdefU}
\end{equation}%
where $u_{1},\dots ,u_{j}:{\mathbb{R}}\rightarrow {\mathbb{R}}$ are
univariate utility functions and $\Lambda :{\mathbb{R}}^{N}\rightarrow {%
\mathbb{R}}$ is concave, increasing with respect to the partial
componentwise order and bounded from above. Inspired by Asymptotic
Satiability as defined in Definition 2.13 \cite{Campi} we will furthermore
assume that for every $\varepsilon >0$ there exist a point $z_\varepsilon\in{%
\mathbb{R}}^N$ and a selection $\nu_\varepsilon\in\partial
\Lambda(z_\varepsilon)$, such that $\sum_{j=1}^{N}\left\vert
\nu_\varepsilon\right\vert <\varepsilon\,.$

For each$\,j=1,\dots ,N\,,$ we also assume the Inada conditions 
\begin{equation*}
\lim_{x\rightarrow +\infty }\frac{u_{j}(x)}{x}=0\,\,\,\,\,\,\,\text{ and }%
\,\,\,\,\,\,\,\lim_{x\rightarrow -\infty }\frac{u_{j}(x)}{x}=+\infty
\end{equation*}%
and that, without loss of generality, $u_{j}(0)=0$. 
\end{standingassumptionI*}

Observe that such a multivariate utility function is split in two
components: the sum of single agent utility functions and a universal part $%
\Lambda $ that could be either selected upon agreement by all the agents or
could be imposed by a regulatory institution. \textit{As }$\Lambda $\textit{%
\ is not necessarily strictly convex nor strictly increasing}, we may choose 
$\Lambda =0$, which corresponds to the case analyzed in \cite{bffm} for the
non conditional case.

\begin{remark}
\label{DynRMnullinzeros} $U$ defined in (\ref{DynRMdefU}) is a multivariate
utility function since it inherits strict concavity and strict monotonicity
from $u_{1},\dots ,u_{N}$. We may assume without loss of generality that $%
u_{j}(0)=0\,\,\,\,\forall \,j=1,\dots ,N$, since we can always write 
\begin{equation*}
U(x)=\sum_{j=1}^{N}\left( u_{j}(x^{j})-u_{j}(0)\right) +\left( \Lambda
(x)+\sum_{j=1}^{N}u_{j}(0)\right)
\end{equation*}%
Thus, we can always redefine the univariate utilities and the multivariate
one, without affecting other assumptions, in such a way that univariate
utilities are null in $0$.

We will make use of the following properties, without explicit mention: for
every $f:{\mathbb{R}}\rightarrow {\mathbb{R}}$ nondecreasing and such that $%
f(0)=0$ it holds that 
\begin{equation}
f(x)=f(x^{+})+f(-x^{-}),\,\,\,\,\,(f(x))^{+}=f\left( x^{+}\right)\,.  \notag
\end{equation}

Moreover, if $u_1,\dots,u_N$ are all null in $0$ (w.l.o.g. by the argument
above), for any $x^1,\dots,x^N\geq 0$ 
\begin{equation}  \label{DynRMgradest}
\sum_{j=1}^Nu_j(x^j)\leq \max_{j=1,\dots,N}\left(\frac{\mathrm{d}u_j}{%
\mathrm{d}x^j}(0) \right)\sum_{j=1}^Nx^j\,
\end{equation}
where $\frac{\mathrm{d}u_j}{\mathrm{d}x^j}(0)$, by an abuse of notation,
stands for any choice in $\partial u_j (0)$, i.e. in the subdifferential of $%
u_j$ at the point $0$, for each $j=1,\dots,N$. Inequality %
\eqref{DynRMgradest} can be showed observing that $\partial u_j (0)\neq
\emptyset $ by concavity and that $\partial u_j (0)\subseteq [0,+\infty)$
since $u_j$ is nondecreasing.
\end{remark}

\section{Conditional Shortfall Systemic Risk Measures on $(L^{\infty }(%
\mathcal{F}))^{N}\label{Sec5}$}

Once we fixed our model for the preferences in the system, we discuss the
set of allocations we admit for the terminal time exchanges, of scenario
dependent nature. This will allow us to formalize the problem which will be
the protagonist of our analysis in the subsequent parts of the paper, and to
state some of its main features in Theorem \ref{DynRMmainthmrhoinfty}.
\bigskip

Given a sub $\sigma$-algebra $\mathcal{G}\subseteq \mathcal{F}$ we introduce
the set 
\begin{equation}
\mathscr{D}_{\mathcal{G}}:=\left\{ Y\in (L^{0}(\Omega ,\mathcal{F},{\mathbb{P%
}}))^{N}\mid \sum_{j=1}^{N}Y^{j}\in L^{0}(\Omega ,\mathcal{G},{\mathbb{P}}%
)\right\} \,.  \label{DynRMdefDG}
\end{equation}

We would like to consider as the set of admissible allocations a subset 
\begin{equation*}
\mathscr{B}_\mathcal{G}\subseteq\mathscr{D}_\mathcal{G}
\end{equation*}
satisfying appropriate conditions (see the Standing Assumption II).

At the same time, we observe that the constraints in (\ref{DynRMdefDG}) can
be interpreted saying that the risk can be shared by all the agents in the
single group $\mathbf{I}:=\{1,\dots ,N\}$. This can be generalized by
introducing the set of constraints corresponding to a cluster of agents
conditional on the information in $\mathcal{G}$, inspired by an example in 
\cite{bffm} for the static case.

\begin{definition}
\label{DynRMdefChcond}

For $h\in \left\{ 1,\dots ,N\right\} ,$ let $\mathbf{I}:=(I_{m})_{m=1,...,h} 
$ be some partition of $\{1,\dots ,N\}$. Then we set {\small 
\begin{equation}
\mathcal{B}^{(\mathbf{I})}_\mathcal{G}:=\left\{ \mathbf{Y}\in (L^{0}(%
\mathcal{F}))^N\mid \ \exists \ d=[d_{1},\dots ,d_{h}]\in (L^0(\mathcal{G}%
))^{h}\mid \sum_{i\in I_{m}}Y^{i}=d_{m}\text{ for\ }m=1,\dots ,h\right\} \,,
\label{DynRMC0g}
\end{equation}
\begin{equation}
\mathcal{B}^{(\mathbf{I}),\infty}_\mathcal{G}:=\left\{ \mathbf{Y}\in (L^{0}(%
\mathcal{F}))^N\mid \ \exists \ d=[d_{1},\dots ,d_{h}]\in (L^\infty(\mathcal{%
G}))^{h}\mid \sum_{i\in I_{m}}Y^{i}=d_{m}\text{ for\ }m=1,\dots ,h\right\}
\,.  \label{DynRMC0ginfty}
\end{equation}%
}
\end{definition}

We stress that the family $\mathcal{B}^{(\mathbf{I})}_\mathcal{G}$ admits
two extreme cases:

\begin{itemize}
\item[(i)] when we have only one group $h=1$ then $\mathcal{B}^{(\mathbf{I}%
)}_\mathcal{G}=\mathscr{D}_\mathcal{G}$ is the largest possible class,
corresponding to risk sharing among all agents in the system;

\item[(ii)] on the opposite side, the strongest restriction occurs when $%
h=N, $ i.e., we consider exactly $N$ groups, and in this case $\mathcal{B}^{(%
\mathbf{I})}_\mathcal{G}=(L^0(\mathcal{G}))^N$ corresponds to no risk
sharing. This generalizes to a dynamic setting the case of deterministic
allocations, when no further information is available (i.e. $\mathcal{G}$ is
trivial). This case has been treated in the literature, especially in the
set-valued approach we mentioned in the Introduction. See also the comments
after Theorem \ref{DynRMmainthmrhoinfty} below for further details.
\end{itemize}

Suppose now a partition $\mathbf{I}$ has been fixed. We will consider a
subset 
\begin{equation*}
\mathscr{B}_{\mathcal{G}}\subseteq \mathcal{B}_{\mathcal{G}}^{(\mathbf{I})}
\end{equation*}%
and note that each component of $Y\in \mathscr{B}_{\mathcal{G}}$ is required
to be $\mathcal{F}$-measurable, while the sums $\sum_{i\in I_{m}}Y^{i}$ are $%
\mathcal{G}$-measurable, and so is consequently $\sum_{j=1}^{N}Y^{j}$. Thus $%
\mathscr{B}_{\mathcal{G}}\subseteq \mathcal{B}_{\mathcal{G}}^{(\mathbf{I}%
)}\subseteq \mathscr{D}_{\mathcal{G}}$.

We define 
\begin{equation}  \label{DynRMdefCG}
\mathscr{C}_\mathcal{G}:=\mathscr{B}_\mathcal{G}\cap\mathcal{B}^{(\mathbf{I}%
,\infty)}_\mathcal{G}\cap (L^1(\Omega,\mathcal{F},{\mathbb{P}}))^N\,.
\end{equation}

\begin{standingassumptionII*}
$\mathscr{B}_\mathcal{G}$ is closed in probability, it is conditionally
convex and it is a conditional cone. Moreover $\mathscr{B}_\mathcal{G}+ (L^0(%
\mathcal{G}))^N=\mathscr{B}_\mathcal{G}$ and the set $\mathscr{C}_{\mathcal{G%
}}$ is closed under $\mathcal{G}$-truncation. We finally consider a $B\in
L^\infty(\mathcal{G})$ with $\esssup{(B)}<\sup_{z\in{\mathbb{R}}^N}U(z)\leq
+\infty$.
\end{standingassumptionII*}

\begin{example}
\label{DynRMexChcond} It is easily seen that taking $\mathscr{B}_{\mathcal{G}%
}=\mathcal{B}_{\mathcal{G}}^{(\mathbf{I})}$ and consequently $\mathscr{C}_{%
\mathcal{G}}=\mathcal{B}_{\mathcal{G}}^{(\mathbf{I},\infty )}\cap
(L^{1}(\Omega ,\mathcal{F},{\mathbb{P}}))^{N}$ Standing Assumption II is
satisfied. Closedness under truncation in particular is verified as follows:
for $Y\in \mathscr{C}_{\mathcal{G}}$, for $j\in I_{m}$ we can take $%
Z_{Y}^{j}=\frac{1}{\left\vert I_{m}\right\vert }\sum_{i\in I_{m}}Y^{i}$
where $\left\vert I_{m}\right\vert $ is the cardinality of $I_{m}$. Then it
is easily verified that $Y_{(k)}$ defined as in \eqref{DynRMdeftruncated}
satisfies for every $m=1,\dots ,h$ 
\begin{equation*}
\sum_{i\in I_{m}}Y_{(k)}^{i}=\left( \sum_{i\in I_{m}}Y^{i}\right)
1_{\bigcap_{j}\{\left\vert Y^{j}\right\vert < k\}}+\left( \sum_{i\in
I_{m}}\left( \frac{1}{\left\vert I_{m}\right\vert }\sum_{i\in
I_{m}}Y^{i}\right) \right) 1_{\bigcup_{j}\{\left\vert Y^{j}\right\vert \geq
k\}}=\sum_{i\in I_{m}}Y^{i}\in L_{\mathcal{G}}^{\infty }
\end{equation*}%
which proves that $Y_{(k)}\in \mathcal{B}_{\mathcal{G}}^{(\mathbf{I}),\infty
}\cap (L^{\infty }(\mathcal{F}))^{N}\subseteq \mathscr{C}_{\mathcal{G}}$ and
that also 
\begin{equation*}
\sum_{j=1}^{N}Y_{(k)}^{j}=\sum_{m=1}^{h}\sum_{i\in
I_{m}}Y_{(k)}^{i}=\sum_{m=1}^{h}\sum_{i\in I_{m}}Y^{i}=\sum_{j=1}^{N}Y^{j}\,.
\end{equation*}%
Finally, we point out that we can cover the setup of \cite{bffm} in our
framework (clearly, here we work with bounded positions and not in an Orlicz
setup). Indeed, we may take the trivial partition $\mathbf{I}=\{\{1,\dots
,N\}\}$ and, to cover the static case, we may choose $\mathcal{G}%
=\{\emptyset ,\Omega \}$. Then we select the set $\mathscr{B}_{\mathcal{G}}$
equal to the set $\mathcal{C}_{0}$, defined in \cite{bffm}, which is assumed
to be closed under truncation in the sense of \cite{bffm} Definition 4.18.
Then our assumptions here are satisfied as well.
\end{example}

\begin{definition}
\label{DynRMdefexplicitrcondinfty}For each $X\in (L^{\infty }(\Omega ,%
\mathcal{F},{\mathbb{P}}))^{N}$ we set%
\begin{align}
\rho _{\mathcal{G}}^{\infty }\left( X\right) & :=\essinf\left\{
\sum_{j=1}^{N}Y^{j}\mid Y\in \mathscr{C}_{\mathcal{G}}\cap (L^{\infty }(%
\mathcal{F}))^{N},\text{ }\mathbb{E}_{\mathbb{P}}\left[ U\left( X+Y\right) %
\middle|\mathcal{G}\right] \geq B\right\} \,,  \label{DynRMdefrcondinfty} \\
\rho _{\mathcal{G}}\left( X\right) & :=\essinf\left\{
\sum_{j=1}^{N}Y^{j}\mid Y\in \mathscr{C}_{\mathcal{G}},\text{ }\mathbb{E}_{%
\mathbb{P}}\left[ U\left( X+Y\right) \middle|\mathcal{G}\right] \geq
B\right\} \,.  \label{DynRMdefrcond}
\end{align}%
and we call $\rho _{\mathcal{G}}\left( X\right) $ the \textbf{Conditional
Shortfall Systemic Risk Measure} associated to the multivariate utility
function $U$ and the set of allocations $\mathscr{C}_{\mathcal{G}}$.
\end{definition}

The difference between the two definitions only resides on the additional
constraint $Y\in (L^{\infty }(\mathcal{F}))^{N}$ appearing in $\rho _{%
\mathcal{G}}^{\infty }\left( X\right) $. As stated in our next main result,
the two Risk Measures coincide under our Standing Assumptions I and II. The
proof, which is quite lengthy, is split in separate results in the following
Section \ref{DynRMsubsecproofthmmain}.

For every ${\mathbb{Q}}=[{\mathbb{Q}}^{1},\dots ,{\mathbb{Q}}^{N}]\in %
\mathscr{Q}_{\mathcal{G}}$ defined in \eqref{DynRMequatintrointro1}, we set 
\begin{equation}
\alpha ^{1}({\mathbb{Q}}):=\esssup\left\{ \sum_{j=1}^{N}\mathbb{E}_{\mathbb{Q%
}^{j}}\left[ -Z^{j}\middle|\mathcal{G}\right] \mid Z\in (L^{\infty }(%
\mathcal{F}))^{N}\text{ and }\mathbb{E}_{\mathbb{P}}\left[ U\left( Z\right) %
\middle|\mathcal{G}\right] \geq B\right\}  \label{DynRMdefalpha1}
\end{equation}%
and we introduce the set 
\begin{equation}
\mathscr{Q}_{\mathcal{G}}^{1}:=\left\{ {\mathbb{Q}}\in \mathscr{Q}_{\mathcal{%
G}}\,\,\middle|\,\,\begin{aligned}&\alpha^1({\mathbb Q})\in
L^1(\mathcal{G})\,\,\,\,\text{ and }\\
&\sum_{j=1}^N\mathbb{E}_{\mathbb{Q}^j}
\left[Y^j\middle|\mathcal{G}\right]\leq \sum_{j=1}^NY^j,\,\forall
Y\in\mathscr{C}_\cG\cap(L^\infty(\mathcal{F}))^N\end{aligned}\right\} .
\label{DynRMDefq1cg}
\end{equation}%
As the set $\mathscr{Q}_{\mathcal{G}}^{1}$ is included in $\mathscr{Q}_{%
\mathcal{G}}$, the observation made in Remark \ref{DynRMWellDefined} on the
conditional expectation applies also here.

\begin{theorem}
\label{DynRMmainthmrhoinfty} Consider the maps $\rho _{\mathcal{G}}^{\infty
} $ and $\rho _{\mathcal{G}}$ defined in (\ref{DynRMdefrcondinfty}) and (\ref%
{DynRMdefrcond}).

\begin{enumerate}
\item $\rho _{\mathcal{G}}^{\infty }\left( X\right) \in L^{\infty }(\mathcal{%
G})$ for all $X\in (L^{\infty }(\mathcal{F}))^{N}$ and $\rho _{\mathcal{G}%
}^{\infty }$ is a Conditional Systemic Risk Measure as $\rho _{\mathcal{G}%
}^{\infty }$ is monotone \eqref{DynRMcmon}, conditionally convex %
\eqref{DynRMccnvx} and conditionally monetary \eqref{DynRMccmt}. It is also
continuous from above and from below in the sense of Definition \ref%
{DynRMdefcont}.

\item For every $X\in (L^{\infty }(\mathcal{F}))^{N}$ 
\begin{equation*}
\rho _{\mathcal{G}}^{\infty }\left( X\right) =\rho _{\mathcal{G}}\left(
X\right)
\end{equation*}

and the essential infimum in (\ref{DynRMdefrcond}) is attained.

\item The CSRM $\rho _{\mathcal{G}}^{\infty }$ admits the following dual
representation:%
\begin{equation}
\rho _{\mathcal{G}}^{\infty }\left( X\right) =\esssup_{{\mathbb{Q}}\in %
\mathscr{Q}_{\mathcal{G}}^{1}}\left( \sum_{j=1}^{N}\mathbb{E}_{\mathbb{Q}%
^{j}}\left[ -X^{j}\middle|\mathcal{G}\right] -\alpha ^{1}({\mathbb{Q}}%
)\right) ,\,\,\,\,\,\forall \,X\in (L^{\infty }(\mathcal{F}))^{N}.
\label{DynRMeqdualreprshorfall}
\end{equation}%
Furthermore, for every $X\in (L^{\infty }(\mathcal{F}))^{N}$ there exists $%
\widehat{{\mathbb{Q}}}\in \mathscr{Q}_{\mathcal{G}}^{1}$ such that 
\begin{equation*}
\rho _{\mathcal{G}}^{\infty }\left( X\right) =\sum_{j=1}^{N}\mathbb{E}_{%
\widehat{{\mathbb{Q}}}^{j}}\left[ -X^{j}\middle|\mathcal{G}\right] -\alpha
^{1}(\widehat{{\mathbb{Q}}})\,.
\end{equation*}
\end{enumerate}
\end{theorem}

As anticipated in the Introduction, several works have focused on the
set-valued theory for Systemic Risk Measures, in both the static and dynamic
case. A key difference with our approach is marked by our use of \emph{%
random allocations}. In \cite{ChenHu18}, \cite{FeinsteinRudloff13}, \cite%
{FeinsteinRudloff15}, \cite{FeinsteinRudloff15a}, \cite{FeinsteinRudloff17}, 
\cite{FeinsteinRudloff21}, \cite{TaharLepinette14} one associates, to each
risky position $X $,\textquotedblleft the set $R_{t}(X)$ of eligible
portfolios at time $t$ that cover the risk of the portfolio $X$%
\textquotedblright , quoting from \cite{FeinsteinRudloff15a}. The risk of $X$
is quantified in these works using a set of vectors which are measurable
with respect $\mathcal{F}_{t}$ (the information known at time $t$). Here,
instead, we are primarily interested in random allocations which happen at
terminal time. Taking $\mathcal{G}=\mathcal{F}_{t}$ in Definition \ref%
{DynRMdefexplicitrcondinfty} and Theorem \ref{DynRMmainthmrhoinfty} to
uniform notation, we stress once again that the amount $\rho _{\mathcal{F}%
_{t}}(X)$ is known once the information of $\mathcal{F}_{t}$ is known, but
this is not the case for $Y\in \mathscr{C}_{\mathcal{F}_{t}}$ since the
latter vectors are $\mathcal{F} $-measurable, hence known only at terminal
time. The label $\mathcal{F}_{t}$ in $\mathscr{C}_{\mathcal{F}_{t}}$ only
points out that $\sum_{j=1}^{N}Y^{j}$ is $\mathcal{F}_{t}$-measurable. An
evident consequence of this can be found in the dual representation result %
\eqref{DynRMeqdualreprshorfall}: the dual variables are taken in the set $%
\mathscr{Q}_{\mathcal{G}}^{1}$ and satisfy the fairness condition $%
\sum_{j=1}^{N}\mathbb{E}_{\mathbb{Q}^{j}}\left[ Y^{j}\middle|\mathcal{F}_{t}%
\right] \leq \sum_{j=1}^{N}Y^{j},\,\forall Y\in \mathscr{C}_{\mathcal{F}%
_{t}}\cap (L^{\infty }(\mathcal{F}))^{N}$. This would be a triviality taking
vectors $Y$ which are componentwise $\mathcal{F}_{t}$-measurable, but
becomes an additional characteristic feature in our setup. Additionally,
recall the scalarization procedure in \cite{FeinsteinRudloff21} for weights $%
[w^{1},\dots ,w^{N}]$ which are $\mathcal{F}_{t}$-measurable, namely 
\begin{equation*}
\rho _{\mathcal{F}_{t}}^{w}(X):=\essinf\left\{ \sum_{j=1}^{N}w^{j}Y^{j}\mid
Y\in R_{t}(X)\right\} \,.
\end{equation*}%
This is meaningful whenever the eligible allocations $Y\in R_{t}(X)$ are $%
\mathcal{F}_{t}$-measurable. We point out that our eligible assets for $\rho
_{\mathcal{F}_{t}}$ satisfy $Y\in \mathscr{C}_{\mathcal{F}_{t}}$, a
condition purposely designed for the valuation $Y\rightarrow
\sum_{j=1}^{N}Y^{j}$. Using any other type of weights $w\in L^{\infty }(%
\mathcal{F}_{t}))^{N}$ would produce an amount $\sum_{j=1}^{N}w^{j}Y^{j}$
which would be in general only $\mathcal{F}$-measurable. This would violate
the basic idea that the (scalar) measurement of risk, given the information
in $\mathcal{F}_{t}$, should only depend on the information in $\mathcal{F}%
_{t}$. 

\subsection{Proof of Theorem \protect\ref{DynRMmainthmrhoinfty}}

%
%
%
\label{DynRMsubsecproofthmmain} In the notation \eqref{DynRMdefrcond}, the
expression $\mathbb{E}_{\mathbb{P}}\left[ U\left( X+Y\right) \middle|%
\mathcal{G}\right] \geq B$ stands for a shortened version of the following
set of conditions: $U(X+Y)\in L^{1}(\Omega ,\mathcal{F},{\mathbb{P}})$ and $%
\mathbb{E}_{\mathbb{P}}\left[ U\left( X+Y\right) \middle|\mathcal{G}\right] $%
, which is well defined, is not ${\mathbb{P}}-$a.s smaller than the random
variable $B$. Recall also that for any random variable $W$ taking values in $%
[0,+\infty ]$ $\mathbb{E}_{\mathbb{P}}\left[ W\middle|\mathcal{G}\right] $
is always well defined via the Radon-Nikodym Theorem (see \cite{Bauer},
Theorems 17.10-11), and in this case the notation $\mathbb{E}_{\mathbb{P}}%
\left[ W\middle|\mathcal{G}\right] $ will be used with this meaning.

For technical reasons we first study the functional $\rho _{\mathcal{G}}$
defined in \eqref{DynRMdefrcond}. We will first prove in Claim \ref%
{DynRMthmalloc} that the range of $\rho _{\mathcal{G}}$ is $L^{\infty }(%
\mathcal{G})$ and then we will show all the properties in Theorem \ref%
{DynRMmainthmrhoinfty} Item 1, made exception for continuity from above and
below, and existence of an allocation for $\rho _{\mathcal{G}}$ . We will
then prove that $\rho _{\mathcal{G}}\equiv \rho _{\mathcal{G}}^{\infty }$ on 
$(L^{\infty }(\mathcal{F}))^{N}$ (Claim \ref{DynRMpropinftycoincide}), which
yields Theorem \ref{DynRMmainthmrhoinfty} Item 2, and move on proving
continuity from below and from above (Claim \ref{DynRMpropconbelow}).
Finally, in Claim \ref{DynRMthmdualreprshorfall} we prove Theorem \ref%
{DynRMmainthmrhoinfty} Item 3. 

\begin{claim}
\label{DynRMthmalloc} The functional $\rho _{\mathcal{G}}$ on $(L^{\infty }(%
\mathcal{F}))^{N}$ takes values in $L^{\infty }(\mathcal{G})$, the infimum
is attained by a $\widehat{Y}\in \mathscr{C}_{\mathcal{G}}$, $\rho _{%
\mathcal{G}}$ is monotone \eqref{DynRMcmon} conditionally convex %
\eqref{DynRMccnvx} and conditionally monetary \eqref{DynRMccmt}.
\end{claim}

\begin{proof}
$\,$

\textbf{STEP 1}: $\rho _{\mathcal{G}}$ takes values in $L^{\infty }(\mathcal{%
G})$.

First we see that the set over which we take the essential infimum defining $%
\rho _{\mathcal{G}}$ is nonempty. We have by monotonicity (for $m$ an $N$%
-dimensional deterministic vector) $\mathbb{E}_{\mathbb{P}}\left[ U\left(
X+m\right) \middle|\mathcal{G}\right] \geq U(-\left\Vert X\right\Vert
_{\infty }+m)$ where $\left\Vert X\right\Vert _{\infty }$ stands for the
vector $[\left\Vert X^{1}\right\Vert _{\infty },\dots ,\allowbreak
\left\Vert X^{N}\right\Vert _{\infty }]\in {\mathbb{R}}^{N}$. Since by
assumption 
\begin{equation*}
\sup_{m\in {\mathbb{R}}^{N}}U(-\left\Vert X\right\Vert _{\infty
}+m)=\sup_{z\in {\mathbb{R}}^{N}}U(z)> \esssup{(B)},
\end{equation*}%
we have consequently $\mathbb{E}_{\mathbb{P}}\left[ U\left( X+m\right) %
\middle|\mathcal{G}\right] \geq B$, for some $m\in {\mathbb{R}}^{N}$.

We claim that the set over which we take the essential infimum is downward
directed. To show this, suppose that $Z,Y\in (L^{1}(\mathcal{F}))^{N}$ are
such that $\sum_{j=1}^{N}Y^{j},\,\sum_{j=1}^{N}Z^{j}\in L^{\infty }(\mathcal{%
G})$ and 
\begin{equation*}
\mathbb{E}_{\mathbb{P}}\left[ U\left( X+Y\right) \middle|\mathcal{G}\right]
\geq B,\text{ }\mathbb{E}_{\mathbb{P}}\left[ U\left( X+Z\right) \middle|%
\mathcal{G}\right] \geq B
\end{equation*}%
Define the set $A:=\{\sum_{j=1}^{N}Y^{j}\leq \sum_{j=1}^{N}Z^{j}\}\in 
\mathcal{G}$ and the random variable $W:=1_{A}Y+1_{A^{c}}Z\in (L^{1}(%
\mathcal{F}))^{N}\cap \mathscr{B}_{\mathcal{G}}$ (observe that it belongs to 
$\mathscr{B}_{\mathcal{G}}$ since $\mathscr{B}_{\mathcal{G}}$ is
conditionally convex). It is easy to see that $\sum_{j=1}^{N}W^{j}=1_{A}%
\sum_{j=1}^{N}Y^{j}+1_{A^{c}}\sum_{j=1}^{N}Z^{j}=\min \left(
\sum_{j=1}^{N}Y^{j},\sum_{j=1}^{N}Z^{j}\right) \in L^{\infty }(\mathcal{G})$%
, so that the set is downward directed. Furthermore 
\begin{equation*}
\mathbb{E}_{\mathbb{P}}\left[ U\left( X+W\right) \middle|\mathcal{G}\right] =%
\mathbb{E}_{\mathbb{P}}\left[ U\left( X+W\right) \middle|\mathcal{G}\right]
1_{A}+\mathbb{E}_{\mathbb{P}}\left[ U\left( X+W\right) \middle|\mathcal{G}%
\right] 1_{A^{c}}=
\end{equation*}%
\begin{equation*}
=\mathbb{E}_{\mathbb{P}}\left[ U\left( X+Y\right) \middle|\mathcal{G}\right]
1_{A}+\mathbb{E}_{\mathbb{P}}\left[ U\left( X+Z\right) \middle|\mathcal{G}%
\right] 1_{A^{c}}\geq B1_{A}+B1_{A^{c}}=B
\end{equation*}%
which concludes the proof of our claim.

Since the set is downward directed, there exists a minimizing sequence $%
(Y_{n})_{n}\subseteq \mathscr{C}_{\mathcal{G}}$ such that $%
\sum_{j=1}^{N}Y_{n}^{j}\downarrow _{n}\rho _{\mathcal{G}}\left( X\right) $
and, having $\rho _{\mathcal{G}}\left( X\right) \leq
\sum_{j=1}^{N}Y_{1}^{j}\in L^{\infty }$, we conclude that $\left\Vert (\rho
_{\mathcal{G}}\left( X\right) )^{+}\right\Vert _{\infty }<+\infty $. Suppose
now by contradiction that for a sequence $k_{n}\uparrow +\infty $ we had ${%
\mathbb{P}}\left( A_{n}\right) >0$ for all $n$, where $A_{n}:=\{\rho _{%
\mathcal{G}}\left( X\right) \leq -k_{n}\}\in \mathcal{G}$. Since for all $%
M\in \mathbb{N}$ we have $-\left\Vert B\right\Vert _{\infty }\leq B\leq 
\mathbb{E}_{\mathbb{P}}\left[ U(X+Y_{M})\middle|\mathcal{G}\right] $ we
deduce:%
\begin{align*}
-\left\Vert B\right\Vert _{\infty }{\mathbb{P}}\left( A_{n}\right) &\leq 
\mathbb{E}_{\mathbb{P}}\left[ B1_{A_{n}}\right] \leq \mathbb{E}_{\mathbb{P}}%
\left[ \mathbb{E}_{\mathbb{P}}\left[ U(X+Y_{M})\middle|\mathcal{G}\right]
1_{A_{n}}\right] =\mathbb{E}_{\mathbb{P}}\left[ U\left( X+Y_{M}\right)
1_{A_{n}}\right] \\
&\overset{\text{Lemma }\ref{DynRMlemmacontrolwithline}.(ii)}{\leq }%
\sum_{j=1}^{N}\mathbb{E}_{\mathbb{P}}\left[ \left( a\left(
X^{j}+Y_{M}^{j}\right) +b\right) 1_{A_{n}}\right] \\
&\leq \left( a\sum_{j=1}^{N}\left\Vert X^{j}\right\Vert _{\infty }+b\right) {%
\mathbb{P}}\left( A_{n}\right) +a\mathbb{E}_{\mathbb{P}}\left[
\sum_{j=1}^{N}Y_{M}^{j}1_{A_{n}}\right] \text{, with }a>0.
\end{align*}%
Consequently 
\begin{align*}
-\left\Vert B\right\Vert _{\infty }{\mathbb{P}}\left( A_{n}\right) &\leq
\left( a\sum_{j=1}^{N}\left\Vert X^{j}\right\Vert _{\infty }+b\right) {%
\mathbb{P}}\left( A_{n}\right) +a\lim_{M}\mathbb{E}_{\mathbb{P}}\left[
\sum_{j=1}^{N}Y_{M}^{j}1_{A_{n}}\right] \\
&\overset{\text{(MON)}}{=}\left( a\sum_{j=1}^{N}\left\Vert X^{j}\right\Vert
_{\infty }+b\right) {\mathbb{P}}\left( A_{n}\right) +a\mathbb{E}_{\mathbb{P}}%
\left[ \rho _{\mathcal{G}}\left( X\right) 1_{A_{n}}\right] \\
&\leq \left( a\sum_{j=1}^{N}\left\Vert X^{j}\right\Vert _{\infty }+b\right) {%
\mathbb{P}}\left( A_{n}\right) -k_{n}a{\mathbb{P}}\left( A_{n}\right) \,.
\end{align*}%
Dividing by ${\mathbb{P}}\left( A_{n}\right) $ and sending $n$ to infinity
we would get a contradiction. This proves that $\left\Vert (\rho _{\mathcal{G%
}}\left( X\right) )^{-}\right\Vert _{\infty }<+\infty $. Recalling that we
already proved $\left\Vert (\rho _{\mathcal{G}}\left( X\right)
)^{+}\right\Vert _{\infty }<+\infty $, we obtain $\rho _{\mathcal{G}}\left(
X\right) \in L^{\infty }(\mathcal{G})$.

\medskip

\textbf{STEP 2}: the infimum in the definition of $\rho _{\mathcal{G}}$ is
attained.

For the minimizing sequence $(Y_{n})_{n}$, from the budget constraint $%
\mathbb{E}_{\mathbb{P}}\left[ U\left( X+Y\right) \middle|\mathcal{G}\right]
\geq B$ and the fact that $\sum_{j=1}^{N}\mathbb{E}_{\mathbb{P}}\left[
X^{j}+Y_{n}^{j}\right] $ is bounded in $n$ because of what we just proved ($%
L^{\infty }\ni \rho _{\mathcal{G}}\left( X\right) \leq
\sum_{j=1}^{N}Y_{n}^{j}\leq \sum_{j=1}^{N}Y_{1}^{j}\in L^{\infty }$) , we
obtain that the sequence $(Y_{n})_{n}$ is bounded in $(L^{1}(\mathcal{F}%
))^{N}$ using Lemma \ref{DynRMlemmakomlos}.

Applying Corollary \ref{DynRMcorkomplosmultidim} we can find a subsequence
and a $\widehat{Y}\in (L^{1}(\mathcal{F}))^{N}$ such that 
\begin{equation*}
W_{K}:=\frac{1}{K}\sum_{k=1}^{K}Y_{n_{k}}\xrightarrow[H\rightarrow
\infty]{\probp-\text{a.s.}}\widehat{Y}.
\end{equation*}%
Furthermore $\sum_{j=1}^{N}Y^{j}\in L^{1}(\mathcal{G})$, $W_{K}\in %
\mathscr{B}_{\mathcal{G}}$ by convexity of the set and $\widehat{Y}\in %
\mathscr{B}_{\mathcal{G}}$ since this set is closed in probability.
Additionally we have that%
\begin{equation}
\sum_{j=1}^{N}\widehat{Y}^{j}=\lim_{K}\frac{1}{K}\sum_{k=1}^{K}%
\sum_{j=1}^{N}Y_{n_{k}}^{j}\overset{\text{Rem.}\ref{DynRMremcesaro}}{=}%
\lim_{k}\sum_{j=1}^{N}Y_{n_{k}}^{j}=\rho _{\mathcal{G}}\left( X\right) \in
L^{\infty }(\mathcal{G})  \label{DynRMsumyisrho}
\end{equation}%
which yields that also $\sum_{j=1}^{N}\widehat{Y}^{j}\in L^{\infty }(%
\mathcal{G})$. To prove that $\widehat{Y}\in \mathscr{C}_{\mathcal{G}}$ we
need to show that $\sum_{i\in I_{m}}\widehat{Y}^{i}\in L^{\infty }(\mathcal{G%
})$ for every $m=1,\dots ,h$. This will be a consequence of Proposition \ref%
{DynRMsuminl0gimpliesitslinftyg}, once we show that $\mathbb{E}_{\mathbb{P}}%
\left[ U\left( X+\widehat{Y}\right) \middle|\mathcal{G}\right] \geq B$.
Hence we now focus on the latter inequality. We observe now that setting $%
Z_{K}:=X+\frac{1}{K}\sum_{k=1}^{K}Y_{n_{k}}$ and $Z=X+\widehat{Y}$ Items \ref%
{DynRMbudgetc} and \ref{DynRMconverge} in Lemma \ref{DynRMproputilfat} are
satisfied. Moreover if we take 
\begin{equation*}
\sum_{j=1}^{N}\mathbb{E}_{\mathbb{P}}\left[ Z_{K}^{j}\middle|\mathcal{G}%
\right] =\sum_{j=1}^{N}\mathbb{E}_{\mathbb{P}}\left[ X^{j}\middle|\mathcal{G}%
\right] +\frac{1}{K}\sum_{k=1}^{K}\sum_{j=1}^{N}Y_{n_{k}}^{j}
\end{equation*}%
we see that the first term in the sum in RHS does not depend on $K$, while
the C\'{esaro} means almost surely converge. Hence also Item \ref%
{DynRMpointwisebdd} in Lemma \ref{DynRMproputilfat} is satisfied, and we get
that $\mathbb{E}_{\mathbb{P}}\left[ U\left( Z\right) \middle|\mathcal{G}%
\right] =\mathbb{E}_{\mathbb{P}}\left[ U\left( X+\widehat{Y}\right) \middle|%
\mathcal{G}\right] \geq B$. As mentioned above, we now get also $\widehat{Y}%
\in \mathscr{C}_{\mathcal{G}}$. We finally recall from \eqref{DynRMsumyisrho}
that $\sum_{j=1}^{N}\widehat{Y}^{j}=\rho _{\mathcal{G}}\left( X\right) $ so
that the infimum is in fact attained at $\widehat{Y}$, which satisfies the
constraints for $\rho _{\mathcal{G}}\left( X\right) $.

\medskip \textbf{STEP 3}: $\rho_{\mathcal{G}}$ satisfies equations %
\eqref{DynRMcmon}, \eqref{DynRMccnvx}, \eqref{DynRMccmt}.

These have to be checked directly using definition of $\rho _{\mathcal{G}%
}\left( \cdot \right) $. We start with \eqref{DynRMcmon}: if $X\leq Z$
componentwise a.s. , for all $Y\in (L^{1}({\mathbb{P}}))^{N}$ such that $%
\mathbb{E}_{\mathbb{P}}\left[ U\left( X+Y\right) \middle|\mathcal{G}\right]
\geq B$ we have automatically (by monotonicity of $U$) that $\mathbb{E}_{%
\mathbb{P}}\left[ U\left( Z+Y\right) \middle|\mathcal{G}\right] \geq \mathbb{%
E}_{\mathbb{P}}\left[ U\left( X+Y\right) \middle|\mathcal{G}\right] \geq B$
so that%
\begin{equation*}
\left\{ \sum_{j=1}^{N}Y^{j}\mid Y\in \mathscr{C}_{\mathcal{G}},\mathbb{E}_{%
\mathbb{P}}\left[ U\left( X+Y\right) \middle|\mathcal{G}\right] \geq
B\right\} \subseteq \left\{ \sum_{j=1}^{N}Y^{j}\mid Y\in \mathscr{C}_{%
\mathcal{G}},\mathbb{E}_{\mathbb{P}}\left[ U\left( Z+Y\right) \middle|%
\mathcal{G}\right] \geq B\right\}
\end{equation*}%
and taking essential infima equation \eqref{DynRMcmon} follows.

As to \eqref{DynRMccnvx}, fix $0\leq \lambda \leq 1,\,\lambda \in L^{\infty
}(\mathcal{G})$ and $X,Z\in (L^{\infty }(\mathcal{F}))^{N}$. For $Y,W\in %
\mathscr{C}_{\mathcal{G}}$ such that $\mathbb{E}_{\mathbb{P}}\left[ U\left(
X+Y\right) \middle|\mathcal{G}\right] \geq B$, $\mathbb{E}_{\mathbb{P}}\left[
U\left( Z+W\right) \middle|\mathcal{G}\right] \geq B$ we then have by
concavity of utilities and $\mathcal{G}$-measurability of $\lambda $ 
\begin{equation*}
\mathbb{E}_{\mathbb{P}}\left[ U\left( \lambda X+(1-\lambda )Z+\lambda
Y+(1-\lambda )W\right) \middle|\mathcal{G}\right] =\mathbb{E}_{\mathbb{P}}%
\left[ U\left( \lambda (X+Y)+(1-\lambda )(Z+W)\right) \middle|\mathcal{G}%
\right]
\end{equation*}%
\begin{equation*}
\geq \lambda \mathbb{E}_{\mathbb{P}}\left[ U\left( X+Y\right) \middle|%
\mathcal{G}\right] +(1-\lambda )\mathbb{E}_{\mathbb{P}}\left[ U\left(
Z+W\right) \middle|\mathcal{G}\right] \geq \lambda B+(1-\lambda )B=B\,.
\end{equation*}%
Moreover obviously $\lambda Y+(1-\lambda )W\in \mathscr{C}_{\mathcal{G}}$,
so that by definition%
\begin{equation*}
\rho _{\mathcal{G}}\left( \lambda X+(1-\lambda )Z\right) \leq \lambda
\sum_{j=1}^{N}Y^{j}+(1-\lambda )\sum_{j=1}^{N}W^{j}\,.
\end{equation*}%
Taking essential infima in RHS over $Y$ and $W$ yields equation %
\eqref{DynRMccnvx}.

Finally we come to \eqref{DynRMccmt}. 
For $Y\in (L^{\infty }(\mathcal{G}))^{N}$ the assumption $\mathscr{B}_{%
\mathcal{G}}+(L^{0}(\mathcal{G}))^{N}=\mathscr{B}_{\mathcal{G}}$ implies
that $W:=Z+Y\in \mathscr{C}_{\mathcal{G}}$ for all $Z\in \mathscr{C}_{%
\mathcal{G}}$. Hence%
\begin{align*}
\rho _{\mathcal{G}}\left( X+Y\right) & =\essinf\left\{
\sum_{j=1}^{N}Z^{j}\mid Z\in \mathscr{C}_{\mathcal{G}},\,\mathbb{E}_{\mathbb{%
P}}\left[ U(X+Y+Z)\middle|\mathcal{G}\right] \geq B\right\} \\
& =\essinf\left\{ \sum_{j=1}^{N}(W^{j}-Y^{j})\mid W\in \mathscr{C}_{\mathcal{%
G}},\,\mathbb{E}_{\mathbb{P}}\left[ U(X+W)\middle|\mathcal{G}\right] \geq
B\right\} =\rho _{\mathcal{G}}\left( X\right) -\sum_{j=1}^{N}Y^{j}.
\end{align*}%
\end{proof}

\begin{claim}
\label{DynRMpropinftycoincide} We have that $\rho^\infty_{\mathcal{G}%
}\left(X\right)=\rho_{\mathcal{G}}\left(X\right)$ for every $X\in (L^\infty(%
\mathcal{F}))^N$.
\end{claim}

\begin{proof}
It is clear that 
\begin{equation*}
\rho_{\mathcal{G}}\left(X\right)\leq \essinf\left\{\sum_{j=1}^N Y^j\mid Y\in%
\mathscr{C}_\mathcal{G}\cap \left(L^\infty(\mathcal{F})\right)^N, \mathbb{E}_%
\mathbb{P} \left[U\left(X+Y\right)\middle| \mathcal{G}\right]\geq B \right\}
\end{equation*}
since the infimum on RHS is taken over a smaller set.

We prove now the reverse inequality: by Claim \ref{DynRMthmalloc} an
allocation exists, call it $Y\in \mathscr{C}_{\mathcal{G}}$. Use closedness
under truncation to see that for $k\geq k_{Y}$ $Y_{(k)}\in \mathscr{C}_{%
\mathcal{G}}$ where $Y_{(k)}$, defined as in \eqref{DynRMdeftruncated},
satisfies $Y_{(k)}\rightarrow _{k}Y\text{ a.s.}$. We want to show that the
convergence $U\left( X+Y_{(k)}+\varepsilon \mathbf{1}\right) \rightarrow
_{k}U\left( X+Y+\varepsilon \mathbf{1}\right) $ is dominated, where $\mathbf{%
1}$ is the $N-$components vector with all components equal to $1$. To see
this observe that $\left\vert U\left( X+Y+\varepsilon \mathbf{1}\right)
\right\vert $ and $\left\vert U\left( X+Z_{Y}+\varepsilon \mathbf{1}\right)
\right\vert $ are integrable: 
\begin{equation*}
L^{1}(\mathcal{F})\ni a\left( \sum_{j=1}^{N}(X^{j}+Y^{j})\right)
+aN\varepsilon +b\overset{\text{Lemma.}\ref{DynRMlemmacontrolwithline}.(ii)}{%
\geq }U\left( X+Y+\varepsilon \mathbf{1}\right) \geq U(X+Y)\in L^{1}(%
\mathcal{F})
\end{equation*}%
while integrability of $\left\vert U\left( X+Z_{Y}+\varepsilon \mathbf{1}%
\right) \right\vert $ is trivial by boundedness of the vectors $X,Z_{Y}$ and
continuity of $U$. Moreover 
\begin{align*}
\left\vert U\left( X+Y_{(k)}+\varepsilon \mathbf{1}\right) \right\vert &
=\left\vert U\left( X+Y+\varepsilon \mathbf{1}\right)
1_{\bigcap_{j}\{\left\vert Y^{j}\right\vert < k\}}+U\left(
X+Z_{Y}+\varepsilon \mathbf{1}\right) 1_{\bigcup_{j}\{\left\vert
Y^{j}\right\vert \geq k\}}\right\vert \\
& \leq \max \left( \left\vert U\left( X+Y+\varepsilon \mathbf{1}\right)
\right\vert ,\,\left\vert U\left( X+Z_{Y}+\varepsilon \mathbf{1}\right)
\right\vert \right) \\
& \leq \left\vert U\left( X+Y+\varepsilon \mathbf{1}\right) \right\vert
+\left\vert U\left( X+Z_{Y}+\varepsilon \mathbf{1}\right) \right\vert \,.
\end{align*}%
Applying (cDOM) we then get that for all $\varepsilon >0$ 
\begin{equation*}
\mathbb{E}_{\mathbb{P}}\left[ U\left( X+Y_{(k)}+\varepsilon \mathbf{1}%
\right) \middle|\mathcal{G}\right] \rightarrow _{k}\mathbb{E}_{\mathbb{P}}%
\left[ U\left( X+Y+\varepsilon \mathbf{1}\right) \middle|\mathcal{G}\right]
>B\,.
\end{equation*}%
From the last expression we infer that%
\begin{equation}
{\mathbb{P}}\left( \Gamma _{K}:=\bigcap_{k\geq K}\left\{ \mathbb{E}_{\mathbb{%
P}}\left[ U\left( X+Y_{(k)}+\varepsilon \mathbf{1}\right) \middle|\mathcal{G}%
\right] \geq B\right\} \right) \uparrow _{K}1\,.
\label{DynRMdefinitelygreat}
\end{equation}%
Fix $K$ and take $\alpha _{K}\in {\mathbb{R}}^{N}$ with 
\begin{equation*}
U\left( -\left\Vert X\right\Vert _{\infty }-\left\Vert Y_{(K)}\right\Vert
_{\infty }+\varepsilon \mathbf{1}+\alpha _{K}\right) \geq \esssup{(B)}
\end{equation*}%
where again $\left\Vert X\right\Vert _{\infty }$ denotes the vector $%
[\left\Vert X^{1}\right\Vert _{\infty },\dots ,\left\Vert X^{N}\right\Vert
_{\infty }]$ and similar notation is used for $\left\Vert Y_{(K)}\right\Vert
_{\infty }$. Notice that such an $\alpha _{K}$ exists since $\sup_{z\in {%
\mathbb{R}}^{N}}U(z)>\esssup{(B)}$. Define $Z_{K}$ by $%
Z_{K}^{j}:=Y_{(K)}^{j}+\varepsilon +\alpha^j _{K}1_{\Gamma
_{K}^{c}},\,j=1,\dots ,N$ and observe that since $\Gamma _{K}\in \mathcal{G}$%
, $Z_{K}\in \mathscr{C}_{\mathcal{G}}\cap \left( L^{\infty }(\mathcal{F}%
)\right) ^{N}$. Furthermore 
\begin{equation*}
\mathbb{E}_{\mathbb{P}}\left[ U\left( X+Z_{K}\right) \middle|\mathcal{G}%
\right] =\mathbb{E}_{\mathbb{P}}\left[ U\left( X+Z_{K}\right) \middle|%
\mathcal{G}\right] 1_{\Gamma _{K}}+\mathbb{E}_{\mathbb{P}}\left[ U\left(
X+Z_{K}\right) \middle|\mathcal{G}\right] 1_{\Gamma _{K}^{c}}
\end{equation*}%
and 
\begin{equation*}
\mathbb{E}_{\mathbb{P}}\left[ U\left( X+Z_{K}\right) \middle|\mathcal{G}%
\right] 1_{\Gamma _{K}}=\mathbb{E}_{\mathbb{P}}\left[ U\left(
X+Y_{(K)}+\varepsilon \mathbf{1}\right) \middle|\mathcal{G}\right] 1_{\Gamma
_{K}}\geq B1_{\Gamma _{K}}
\end{equation*}%
by definition of $\Gamma _{K}$ and the fact that $1_{\Gamma _{K}}$ can be
moved inside conditional expectation.

Moreover by definition of $\alpha _{K}$ 
\begin{equation*}
\mathbb{E}_{\mathbb{P}}\left[ U\left( X+Z_{K}\right) \middle|\mathcal{G}%
\right] 1_{\Gamma _{K}^{c}}=\mathbb{E}_{\mathbb{P}}\left[ U\left(
X+Y_{(K)}+\varepsilon \mathbf{1}+\alpha _{K}\right) \middle|\mathcal{G}%
\right] 1_{\Gamma _{K}^{c}}\geq B1_{\Gamma _{K}^{c}}\,.
\end{equation*}%
Hence we have that $Z_{K}\in \mathscr{C}_{\mathcal{G}}\cap \left( L^{\infty
}(\mathcal{F})\right) ^{N}$, $\mathbb{E}_{\mathbb{P}}\left[ U\left(
X+Z_{K}\right) \middle|\mathcal{G}\right] \geq B$, and we conclude that 
\begin{equation}
\rho _{\mathcal{G}}^{\infty }\left( X\right) \leq \sum_{j=1}^{N}Z_{K}^{j}\,.
\label{DynRMineqext}
\end{equation}%
Now, by \eqref{DynRMdefinitelygreat}, for almost all $\omega \in \Omega $
there exists a $K(\omega )\in \mathbb{N}$ such that $\omega \in \Gamma _{K}$
for all $K\geq K(\omega )$, which implies for all $j=1,\dots ,N$ $%
Z_{K}^{j}(\omega )=Y_{(K)}^{j}(\omega )+\varepsilon \,\,\,\forall \,K\geq
K(\omega )$.

By definition $Y_{(K)}\rightarrow _{K}Y\text{ a.s.}$, so that by %
\eqref{DynRMineqext} we can write for almost all $\omega \in \Omega $: 
\begin{equation*}
\begin{split}
\rho _{\mathcal{G}}^{\infty }\left( X\right) & \leq \liminf_{K\rightarrow
+\infty }\sum_{j=1}^{N}Z_{K}^{j}=\liminf_{K\rightarrow +\infty }\left(
\sum_{j=1}^{N}\left( Y_{(K)}^{j}+\varepsilon \right) \right) \\
& =\lim_{K\rightarrow +\infty }\sum_{j=1}^{N}Y_{(K)}^{j}+N\varepsilon =\rho
_{\mathcal{G}}\left( X\right) +N\varepsilon \,.
\end{split}%
\end{equation*}%
Hence $\rho _{\mathcal{G}}^{\infty }\left( X\right) \leq \rho _{\mathcal{G}%
}\left( X\right) \text{ }\mathbb{P}\text{-a.s.}$, which implies $\rho _{%
\mathcal{G}}^{\infty }\left( X\right) =\rho _{\mathcal{G}}\left( X\right) 
\text{ }\mathbb{P}\text{-a.s.}$.
\end{proof}

\begin{claim}
\label{DynRMpropconbelow} The CSRM $\rho _{\mathcal{G}}$ on $(L^{\infty }(%
\mathcal{F}))^{N}$ is continuous from below and from above, in the sense of
Definition \ref{DynRMdefcont}.
\end{claim}

\begin{proof}
Consider a sequence $X_n\uparrow_n X$ and take any $Y\in\mathscr{C}_{%
\mathcal{G}}\cap (L^\infty)^N$ such that $\mathbb{E}_\mathbb{P} \left[%
U\left(X+Y\right)\middle| \mathcal{G}\right]\geq B$. Then for any $%
\varepsilon>0$

\begin{equation*}
B<\mathbb{E}_{\mathbb{P}}\left[ U\left( X+Y+\varepsilon \mathbf{1}\right) %
\middle|\mathcal{G}\right] \overset{\text{(cMON)}}{=}\lim_{n}\mathbb{E}_{%
\mathbb{P}}\left[ U\left( X_{n}+Y+\varepsilon \mathbf{1}\right) \middle|%
\mathcal{G}\right]
\end{equation*}%
Hence the sequence $(A_{K})_{K}$, where 
\begin{equation*}
A_{K}:=\left\{ \mathbb{E}_{\mathbb{P}}\left[ U\left( X_{n}+Y+\varepsilon 
\mathbf{1}\right) \middle|\mathcal{G}\right] \geq B,\,\,\forall \,n\geq
K\right\}
\end{equation*}%
satisfies ${\mathbb{P}}(A_{K})\uparrow _{K}1$. Take $\alpha _{K}\in {\mathbb{%
R}}^{N}$ such that 
\begin{equation*}
U(-\left\Vert X_{n}\right\Vert _{\infty }-\left\Vert Y\right\Vert _{\infty
}+\varepsilon \mathbf{1}+\alpha _{K})\geq \esssup{(B)}\,\,\forall \,n\geq K
\end{equation*}%
where the notation for $\left\Vert X_{n}\right\Vert _{\infty }$ and $%
\left\Vert Y\right\Vert _{\infty }$ is the same as in the proof of Claim \ref%
{DynRMpropinftycoincide}. Define $Z_{K}\in (L^{\infty }(\mathcal{F}))^{N}$
by $Z_{K}^{j}:=Y^{j}+\varepsilon \mathbf{1}+\alpha _{K}1_{A_{K}^{c}}$ for $%
j=1,\dots ,N$. Since $A_{K}\in \mathcal{G}$ we have $Z\in \mathscr{C}_{%
\mathcal{G}}$. Furthermore for all $n\geq K$ 
\begin{align*}
\mathbb{E}_{\mathbb{P}}\left[ U\left( X_{n}+Z_{K}\right) \middle|\mathcal{G}%
\right] &=\mathbb{E}_{\mathbb{P}}\left[ U\left( X_{n}+Z_{K}\right) \middle|%
\mathcal{G}\right] 1_{A_{K}}+\mathbb{E}_{\mathbb{P}}\left[ U\left(
X_{n}+Z_{K}\right) \middle|\mathcal{G}\right] 1_{A_{K}^{c}} \\
&\geq B1_{A_{K}}+\esssup{(B)}1_{A_{K}^{c}}\geq B\,.
\end{align*}%
Hence by definition of $\rho _{\mathcal{G}}\left( X_{n}\right) $ 
\begin{align*}
\rho _{\mathcal{G}}\left( X_{n}\right) & \leq
\sum_{j=1}^{N}Z_{K}^{j}=\sum_{j=1}^{N}Y^{j}+N\varepsilon
+\sum_{j=1}^{N}\alpha _{K}^{j}1_{A_{K}^{c}}, \\
\lim_{n}\rho _{\mathcal{G}}\left( X_{n}\right) & \leq \liminf_{K}\left(
\sum_{j=1}^{N}Y^{j}+N\varepsilon +\sum_{j=1}^{N}\alpha
_{K}^{j}1_{A_{K}^{c}}\right) .
\end{align*}%
Recall now that ${\mathbb{P}}(A_{K})\rightarrow _{K}1$ and $A_{K}\subseteq
A_{K+1}$. Hence almost all $\omega \in \Omega $ are such that $%
1_{A_{K}^{c}}(\omega )=0$ definitely in $K$. As a consequence 
\begin{equation*}
\liminf_{K}\left( \sum_{j=1}^{N}Y^{j}+N\varepsilon +\sum_{j=1}^{N}\alpha
_{K}^{j}1_{A_{K}^{c}}\right) =\liminf_{K}\left(
\sum_{j=1}^{N}Y^{j}+N\varepsilon \right) =\sum_{j=1}^{N}Y^{j}+N\varepsilon\,.
\end{equation*}%
It follows that 
\begin{equation*}
\lim_{n}\rho _{\mathcal{G}}\left( X_{n}\right) \leq \sum_{j=1}^{N}Y^{j}\,\,{%
\mathbb{P}}-\text{a.s.}
\end{equation*}%
and this holds for all $Y\in \mathscr{C}_{\mathcal{G}}$ such that $\mathbb{E}%
_{\mathbb{P}}\left[ U\left( X+Y\right) \middle|\mathcal{G}\right] \geq B$.
Taking essential infimum on RHS for $Y\in \mathscr{C}_{\mathcal{G}}\cap
(L^{\infty }(\mathcal{F}))^{N}$, $\mathbb{E}_{\mathbb{P}}\left[ U\left(
X+Y\right) \middle|\mathcal{G}\right] \geq B$ by Claim \ref%
{DynRMpropinftycoincide} we obtain 
\begin{equation*}
\lim_{n}\rho _{\mathcal{G}}\left( X_{n}\right) \leq \rho^\infty _{\mathcal{G}%
}\left( X\right)=\rho _{\mathcal{G}}\left( X\right) \overset{%
\eqref{DynRMcmon}}{\leq }\lim_{n}\rho _{\mathcal{G}}\left( X_{n}\right) ,
\end{equation*}%
which shows continuity from below. By monotone convergence, the continuity
from below of $\rho _{\mathcal{G}}$ yields the continuity from below of $%
\rho _{0}(\cdot ):=\mathbb{E}_{\mathbb{P}}\left[ \rho _{\mathcal{G}}\left(
\cdot \right) \right] :(L^{\infty }(\mathcal{F}))^{N}\rightarrow {\mathbb{R}}
$ so that Theorem \ref{DynRMthmdual} item i) shows that $\rho _{0}$ is
nicely representable. The continuity from above then follows from Theorem %
\ref{DynRMcorcashadd}.
\end{proof}

\label{DynRMsecdualrepshortfall}We now study the dual representation of the
CSRM $\rho _{\mathcal{G}}^{\infty }$. Notice that we just showed that
Theorem \ref{DynRMcorcashadd} applies and so it yields the dual
representation (\ref{DynRMdualreprgeneralcor}) for $\rho _{\mathcal{G}%
}^{\infty },$ using $L_{\mathcal{F}}:=(L^{\infty }(\mathcal{F}))^{N}$ and $%
L^{\ast }:=(L^{1}(\mathcal{F}))^{N}$. However, in view of Claim \ref%
{DynRMpropinftycoincide}, we can apply an argument inspired by \cite%
{Scandolo} Proposition 3.6 to get a more specific dual representation.
Observe that the set $\mathscr{Q}_{\mathcal{G}}$ defined in (\ref%
{DynRMequatintrointro1}) takes the form%
\begin{equation}
\mathscr{Q}_{\mathcal{G}}:=\left\{ {\mathbb{Q}}\ll {\mathbb{P}}:\,\frac{%
\mathrm{d}{\mathbb{Q}}}{\mathrm{d}{\mathbb{P}}}\in (L^{1}(\mathcal{F}%
))^{N},\,\mathbb{E}_{\mathbb{P}}\left[ \frac{\mathrm{d}{\mathbb{Q}}^{j}}{%
\mathrm{d}{\mathbb{P}}}\middle|\mathcal{G}\right] =1\,\,\forall \,j=1,\dots
,N\right\}  \label{DynRMdefQgforshortfall}
\end{equation}%
and let 
\begin{equation}
\rho _{\mathcal{G}}^{\ast }(Y):=\esssup_{X\in L_{\mathcal{F}}}\left\{
\sum_{j=1}^{N}\mathbb{E}_{{\mathbb{P}}}\left[ X^{j}Y^{j}\middle|\mathcal{G}%
\right] -\rho _{\mathcal{G}}\left( X\right) \right\} \text{,\quad }{Y\in
L^{\ast }}.  \label{DynRmdefdualgeneral}
\end{equation}

\begin{claim}
\label{DynRMthmdualreprshorfall} Let $\rho _{\mathcal{G}}:(L^{\infty }(%
\mathcal{F}))^{N}\rightarrow L^{\infty }(\mathcal{G})$ be defined by %
\eqref{DynRMdefrcondinfty} and take $\alpha ^{1}(\cdot )$ as in %
\eqref{DynRMdefalpha1}. Then the following are equivalent for fixed $p\in
\{0,1\}$ and ${\mathbb{Q}}\in \mathscr{Q}_{\mathcal{G}}$:

\begin{enumerate}
\item $\rho^*_\mathcal{G}\left(-\frac{\mathrm{d}{\mathbb{Q}}}{\mathrm{d}{%
\mathbb{P}}}\right)\in L^p(\mathcal{G})$.

\item $\alpha({\mathbb{Q}})\in L^p(\mathcal{G})$, where $\alpha$ is defined
in \eqref{DynRMequatintro2} for $L_\mathcal{F}=(L^\infty(\mathcal{F}))^N$
and $L^*=(L^1(\mathcal{F}))^N$.

\item $\alpha^1({\mathbb{Q}})\in L^p(\mathcal{G})$ and $\sum_{j=1}^N\mathbb{E%
}_{\mathbb{Q}^j} \left[Y^j|\mathcal{G}\right]\leq \sum_{j=1}^N Y^j$ for all $%
Y\in\mathscr{C}_\mathcal{G}\cap (L^\infty(\mathcal{F}))^N$.
\end{enumerate}

Moreover $\rho _{\mathcal{G}}$ admits the dual representation in %
\eqref{DynRMeqdualreprshorfall} 
and, for every $X\in (L^{\infty }(\mathcal{F}))^{N}$, there exists $\widehat{%
{\mathbb{Q}}}\in \mathscr{Q}_{\mathcal{G}}^{1}$ such that $\rho _{\mathcal{G}%
}\left( X\right) =\sum_{j=1}^{N}\mathbb{E}_{\widehat{{\mathbb{Q}}}^{j}}\left[
-X^{j}\middle|\mathcal{G}\right] -\alpha ^{1}(\widehat{{\mathbb{Q}}})$.
\end{claim}

\begin{proof}
Recall that in this specific setup we have by Claim \ref{DynRMthmalloc} that 
$\rho _{\mathcal{G}}\in L^{\infty }(\mathcal{G})$ and that $\rho _{\mathcal{G%
}}=\rho _{\mathcal{G}}^{\infty }$ on $(L^{\infty }(\mathcal{F}))^{N}$ by
Claim \ref{DynRMpropinftycoincide}. We argued before that Theorem \ref%
{DynRMcorcashadd} applies here. From its proof, more precisely from STEP 5,
selecting the setup $L_{\mathcal{F}}=(L^{\infty }(\mathcal{F}))^{N}$ and $%
L^{\ast }=(L^{1}(\mathcal{F}))^{N}$ we have: 
\begin{equation}
\alpha ({\mathbb{Q}})=\rho _{\mathcal{G}}^{\ast }\left( -\frac{\mathrm{d}{%
\mathbb{Q}}}{\mathrm{d}{\mathbb{P}}}\right)  \label{DynRMlinkalpharhostar}
\end{equation}%
for all ${\mathbb{Q}}\in \mathscr{Q}_{\mathcal{G}}$. Moreover we have: 
\begin{align*}
\rho _{\mathcal{G}}^{\ast }\left( -\frac{\mathrm{d}{\mathbb{Q}}}{\mathrm{d}{%
\mathbb{P}}}\right) & =\esssup_{X\in (L^{\infty }(\mathcal{F}))^{N}}\left(
\sum_{j=1}^{N}\mathbb{E}_{\mathbb{Q}^{j}}\left[ -X^{j}\middle|\mathcal{G}%
\right] -\rho _{\mathcal{G}}\left( X\right) \right) \\
& \overset{\eqref{DynRMdefrcondinfty}}{=}\esssup_{X\in (L^{\infty }(\mathcal{%
F}))^{N}}\left( \sum_{j=1}^{N}\mathbb{E}_{\mathbb{Q}^{j}}\left[ -X^{j}\middle%
|\mathcal{G}\right] -\essinf_{\substack{ Y\in \mathscr{C}_{\mathcal{G}}\cap
(L^{\infty }(\mathcal{F}))^{N}  \\ \mathbb{E}_{\mathbb{P}}\left[ U\left(
X+Y\right) \middle|\mathcal{G}\right] \geq B}}\left(
\sum_{j=1}^{N}Y^{j}\right) \right) \\
& =\esssup_{\substack{ X,Y\in (L^{\infty }(\mathcal{F}))^{N}  \\ Y\in %
\mathscr{C}_{\mathcal{G}},\mathbb{E}_{\mathbb{P}}\left[ U\left( X+Y\right) %
\middle|\mathcal{G}\right] \geq B}}\left( \sum_{j=1}^{N}\mathbb{E}_{\mathbb{Q%
}^{j}}\left[ -X^{j}\middle|\mathcal{G}\right] -\sum_{j=1}^{N}Y^{j}\right) \\
& =\esssup_{\substack{ Z,Y\in (L^{\infty }(\mathcal{F}))^{N},  \\ Y\in %
\mathscr{C}_{\mathcal{G}},\mathbb{E}_{\mathbb{P}}\left[ U\left( Z\right) %
\middle|\mathcal{G}\right] \geq B}}\left( \sum_{j=1}^{N}\mathbb{E}_{\mathbb{Q%
}^{j}}\left[ -(Z^{j}-Y^{j})\middle|\mathcal{G}\right] -\sum_{j=1}^{N}Y^{j}%
\right) .
\end{align*}%
We conclude that 
\begin{equation}
\alpha ({\mathbb{Q}})=\esssup_{\substack{ Z\in (L^{\infty }(\mathcal{F}%
))^{N},  \\ \mathbb{E}_{\mathbb{P}}\left[ U\left( Z\right) \middle|\mathcal{G%
}\right] \geq B}}\left( \sum_{j=1}^{N}\mathbb{E}_{\mathbb{Q}^{j}}\left[
-Z^{j}\middle|\mathcal{G}\right] \right) +\esssup_{Y\in \mathscr{C}_{%
\mathcal{G}}\cap (L^{\infty }(\mathcal{F}))^{N}}\left( \sum_{j=1}^{N}\mathbb{%
E}_{\mathbb{Q}^{j}}\left[ Y^{j}\middle|\mathcal{G}\right] -%
\sum_{j=1}^{N}Y^{j}\right) \,.  \label{DynRMdualrepalphaspecific}
\end{equation}

The equivalence among Items 1-2-3 is now clear, once we observe that for
every ${\mathbb{Q}}\in\mathscr{Q}_\mathcal{G}$ such that $\alpha({\mathbb{Q}}%
)\in L^0(\mathcal{G})$ we must have $\sum_{j=1}^N\mathbb{E}_{\mathbb{Q}^j} %
\left[Y^j\middle|\mathcal{G}\right]-\sum_{j=1}^NY^j\leq 0\,\,{\mathbb{P}}-$%
a.s. since $\mathscr{C}_\mathcal{G}\cap (L^\infty(\mathcal{F}))^N$ is a
conditional cone. 

All the claims then follow from Theorem \ref{DynRMcorcashadd}, observing
that for the optimum $\widehat{{\mathbb{Q}}}$ provided there we must have $%
\alpha({\mathbb{Q}})\in L^1(\mathcal{G})$ (since $\rho_{\mathcal{G}%
}\left(X\right)\in L^\infty(\mathcal{G})$).
\end{proof}

\begin{remark}
We stress the fact that by Claim \ref{DynRMthmdualreprshorfall} we have for
every $Y\in\mathscr{C}_\mathcal{G}\cap(L^\infty(\mathcal{F}))^N$ 
\begin{equation}  \label{DynRMfairsums}
\sum_{j=1}^N\mathbb{E}_{\mathbb{Q}^j} \left[Y^j\middle|\mathcal{G}\right]%
\leq \sum_{j=1}^NY^j\,\,\,\,{\mathbb{P}}-\text{a.s.}\,\,\,\,\,\text{ for all 
}{\mathbb{Q}}\in\mathscr{Q}_\mathcal{G}\text{ such that }\alpha({\mathbb{Q}}%
)\in L^0(\mathcal{G})\,.
\end{equation}
\end{remark}


\subsection{Uniqueness and integrability of optima of $\protect\rho _{%
\mathcal{G}}$}

In this Section, under suitable additional assumptions on $U$ we prove
uniqueness for primal optimal allocations of $\rho_\mathcal{G}$. We also
provide a fairness condition in the form \eqref{DynRMfairsums} for such
optima $\widehat{Y}$ and measures in $\mathscr{Q}_\mathcal{G}^1$ (notice
that, a priori, at the moment we do not even know if $\widehat{Y}$ is
integrable under the various measures ${\mathbb{Q}}\in\mathscr{Q}^1_\mathcal{%
G}$).

\subsubsection{Uniqueness}

\begin{assumption}
\label{DynRMreqint} The function $U:{\mathbb{R}}^{N}\rightarrow {\mathbb{R}}$
satisfies:{\small 
\begin{equation}
X\in (L^{1}(\Omega ,\mathcal{F},{\mathbb{P}}))^{N},(U(X))^{-}\in
L^{1}(\Omega ,\mathcal{F},{\mathbb{P}})\Rightarrow \exists \,\delta >0\text{
s.t. }(U(X-\varepsilon \mathbf{1}))^{-}\in L^{1}(\Omega ,\mathcal{F},{%
\mathbb{P}})\,\forall \,0\leq \varepsilon <\delta \,.  \label{DynRMeqreqint}
\end{equation}%
}
\end{assumption}

Observe that for example taking $\alpha_1,\dots,\alpha_N,\beta_1,\dots,%
\beta_N>0$ the function 
\begin{equation*}
U(x):=\sum_{j=1}^N\left(1-\exp\left(-\alpha_j
x_j\right)\right)+\left(1-\exp\left(-\sum_{j=1}^N\beta_j x_j\right)\right)
\end{equation*}
satisfies Assumption \ref{DynRMreqint}.

\begin{proposition}
Under Assumption \ref{DynRMreqint} $\rho _{\mathcal{G}}\left( X\right) $
defined in (\ref{DynRMdefrcond}) admits a unique optimum in $\mathscr{C}_{%
\mathcal{G}}$ for every $X\in (L^{\infty }(\mathcal{F}))^{N}$.
\end{proposition}

\begin{proof}
Suppose $\widehat{Y}_{1}\neq \widehat{Y}_{2}$ were two optima. Then clearly
so is $\widehat{Y}_{\lambda }=\lambda \widehat{Y}_{1}+(1-\lambda )\widehat{Y}%
_{2}$ for $\lambda \in {\mathbb{R}},0<\lambda <1$ by concavity of $U$. At
the same time, we have that $\Gamma :=\{\mathbb{E}_{\mathbb{P}}\left[ U(X+%
\widehat{Y}_{\lambda })\middle|\mathcal{G}\right] >\lambda \mathbb{E}_{%
\mathbb{P}}\left[ U(X+\widehat{Y}_{1})\middle|\mathcal{G}\right] +(1-\lambda
)\mathbb{E}_{\mathbb{P}}\left[ U(X+\widehat{Y}_{2})\middle|\mathcal{G}\right]
\}\in \mathcal{G}$ satisfies ${\mathbb{P}}(\Gamma )=1$ by strict concavity:
if this were not the case, from concavity and 
\begin{equation*}
\mathbb{E}_{\mathbb{P}}\left[ U(X+\widehat{Y}_{\lambda })1_{\Gamma ^{c}}%
\right] =\lambda \mathbb{E}_{\mathbb{P}}\left[ U(X+\widehat{Y}_{1})1_{\Gamma
^{c}}\right] +(1-\lambda )\mathbb{E}_{\mathbb{P}}\left[ U(X+\widehat{Y}%
_{2})1_{\Gamma ^{c}}\right]
\end{equation*}%
we would get that on $\Gamma ^{c}$, which has positive probability, $U(X+%
\widehat{Y}_{\lambda })=\lambda U(X+\widehat{Y}_{1})+(1-\lambda )U(X+%
\widehat{Y}_{2})$ which contradicts strict concavity of $U$.

Fix now for some $\varepsilon >0$. Recall that we showed 
\begin{equation*}
\mathbb{E}_{\mathbb{P}}\left[ U(X+\widehat{Y}_{\lambda })\middle|\mathcal{G}%
\right] >\lambda \mathbb{E}_{\mathbb{P}}\left[ U(X+\widehat{Y}_{1})\middle|%
\mathcal{G}\right] +(1-\lambda )\mathbb{E}_{\mathbb{P}}\left[ U(X+\widehat{Y}%
_{2})\middle|\mathcal{G}\right] \geq B\,\,\,{\mathbb{P}}-\text{a.s.}.
\end{equation*}%
By monotonicity of $U$ and Assumption \ref{DynRMreqint} we have the
convergence 
\begin{equation*}
\mathbb{E}_{\mathbb{P}}\left[ U\left( X+\widehat{Y}_{\lambda }-\frac{1}{H}%
\mathbf{1}\right) \middle|\mathcal{G}\right] \uparrow _{H}\mathbb{E}_{%
\mathbb{P}}\left[ U(X+\widehat{Y}_{\lambda })\middle|\mathcal{G}\right]
>B\,\,\,{\mathbb{P}}-\text{a.s.}
\end{equation*}%
where $\mathbf{1}:=[1,\dots ,1]\in {\mathbb{R}}^{N}$. By Egorov Theorem \ref%
{DynRMegorovthm}, there exists a $\Xi \in \mathcal{G}$, with ${\mathbb{P}}%
(\Xi )>0$, such that on $\Xi $ both the following conditions hold: $\mathbb{E%
}_{\mathbb{P}}\left[ U\left( X+\widehat{Y}_{\lambda }\right) \middle|%
\mathcal{G}\right] \geq B+\varepsilon $ and the convergence 
\begin{equation*}
\mathbb{E}_{\mathbb{P}}\left[ U\left( X+\widehat{Y}_{\lambda }-\frac{1}{H}%
\mathbf{1}\right) \middle|\mathcal{G}\right] \uparrow _{H}\mathbb{E}_{%
\mathbb{P}}\left[ U(X+\widehat{Y}_{\lambda })\middle|\mathcal{G}\right]
\end{equation*}%
is uniform. Hence, definitely in $H\in \mathbb{N}$, 
\begin{equation*}
\mathbb{E}_{\mathbb{P}}\left[ U\left( X+\widehat{Y}_{\lambda }-\frac{1}{H}%
1_{\Xi }\mathbf{1}\right) \middle|\mathcal{G}\right] \geq B\,.
\end{equation*}%
At the same time, $\widehat{Y}_{\lambda }-\frac{1}{H}1_{\Xi }\mathbf{1}\in %
\mathscr{C}_{\mathcal{G}}$ and by definition $\rho _{\mathcal{G}}\left(
X\right) 1_{\Xi }\leq \left( \sum_{j=1}^{N}\widehat{Y}_{\lambda }^{j}-N\frac{%
1}{H}\right) 1_{\Xi }<\left( \sum_{j=1}^{N}\widehat{Y}_{\lambda }^{j}\right)
1_{\Xi }$ which contradicts the optimality of $\widehat{Y}$, as ${\mathbb{P}}%
(\Xi )>0$.
\end{proof}

\subsubsection{Integrability}

\label{DynRMsecintegr}

We now wish to establish a conditional fairness property for any optimum $%
\widehat{Y}$ from Theorem \ref{DynRMmainthmrhoinfty}, namely we aim to prove
(see Proposition \ref{DynRMpropoptisint}): $\sum_{j=1}^{N}\mathbb{E}_{%
\mathbb{Q}^{j}}\left[ \widehat{Y}^{j}\middle|\mathcal{G}\right] \leq
\sum_{j=1}^{N}\widehat{Y}^{j}$ for any ${\mathbb{Q}}\in \mathscr{Q}_{%
\mathcal{G}}^{1}$ (defined in \eqref{DynRMDefq1cg}). Notice that this is not
automatic from the fairness condition \eqref{DynRMfairsums} coming from the
definition of ${\mathbb{Q}}\in \mathscr{Q}_{\mathcal{G}}^{1}$, since we do
not know in general if $\widehat{Y}\in \mathscr{C}_{\mathcal{G}}\cap
(L^{\infty }(\mathcal{F}))^{N}$ (we only know, from Theorem \ref%
{DynRMmainthmrhoinfty}, that $\widehat{Y}\in \mathscr{C}_{\mathcal{G}}$). We
point out that Proposition \ref{DynRMpropoptisint} will be also needed in
the proof of Theorem \ref{DynRMthmallocissorte}. In order to show such a
fairness property, we need to establish an integrability result for such a $%
\widehat{Y}$, and the theory of multivariate Orlicz spaces will come in
handy for this purpose. To each univariate Young function $\phi :{\mathbb{R}}%
_{+}\rightarrow {\mathbb{R}}$ we can associate its conjugate $\phi ^{\ast
}(y):=\sup_{x\in {\mathbb{R}}}\left( xy-\phi (x)\right) $. As in \cite%
{RaoRen}, we can associate to both $\phi $ and $\phi ^{\ast }$ the Orlicz
spaces and Hearts $L^{\phi },M^{\phi },L^{\phi ^{\ast }},M^{\phi ^{\ast }}$.
Univariate Young functions naturally arise from univariate utility functions 
$u$, setting $\phi (x):=u(0)-u(-x),\,x\geq 0$. We now recall from \cite{DF19}
how to produce multivariate Orlicz functions and spaces from multivariate
utility functions, inspired by \cite{Drapeau}, Appendix B. Indeed, for a
multivariate utility function $U$ specified \ in \eqref{DynRMdefU}, we
define the function $\Phi $ on $({\mathbb{R}}_{+})^{N}$ by 
\begin{equation}
\Phi (y):=U(0)-U(-y)\,.  \label{DynRMassocorlicz}
\end{equation}%
and 
\begin{equation}
\phi _{j}(z):=u_{j}(0)-u_{j}(-z),\,\,z\in {\mathbb{R}}_{+}
\label{DynRMassocorliczj}
\end{equation}%
as the (univariate) functions associated to the univariate utilities $%
u_{1},\dots ,u_{N}$. Note that $\Phi $ defined in \eqref{DynRMassocorlicz}
is a multivariate Orlicz function (\cite{DF19} Lemma 3.5.(i) observing that $%
U$ satisfying Standing Assumption I in this paper is well controlled
according to \cite{DF19} Definition 3.4, as shown in \cite{DF19} Proposition
7.1). As such, $\Phi $ in \eqref{DynRMassocorlicz} generates a multivariate
Orlicz space and a multivariate Orlicz Heart:%
\begin{align}
L^{\Phi }& :=\left\{ X\in L^{0}\left( (\Omega ,\mathcal{F},{\mathbb{P}}%
);[-\infty ,+\infty ]^{N}\right) \mid \,\exists \,\lambda \in (0,+\infty )%
\text{ s.t. }\mathbb{E}_{\mathbb{P}}\left[ \Phi (\lambda \left\vert
X\right\vert )\right] <+\infty \right\}\,,  \notag \\
M^{\Phi }& :=\left\{ X\in L^{0}\left( (\Omega ,\mathcal{F},{\mathbb{P}}%
);[-\infty ,+\infty ]^{N}\right) \mid \,\forall \,\lambda \in (0,+\infty )%
\text{ }\mathbb{E}_{\mathbb{P}}\left[ \Phi (\lambda \left\vert X\right\vert )%
\right] <+\infty \right\} \,.  \label{DynRMMphi}
\end{align}%
Moreover, $\phi _{1},\dots ,\phi _{N}$ are univariate Orlicz functions.

The K\"{o}the dual $K_{\Phi }$ of the space $L^{\Phi }$ is defined as 
\begin{equation}
K_{\Phi }:=\left\{ Z\in L^{0}\left( (\Omega ,\mathcal{F},{\mathbb{P}}%
);[-\infty ,+\infty ]^{N}\right) \mid \sum_{j=1}^{N}X^{j}Z^{j}\in
L^{1}(\Omega ,\mathcal{F},{\mathbb{P}})\quad \forall \,X\in L^{\Phi
}\right\} \,.  \label{DynRMKphi}
\end{equation}

Section 2.1 in \cite{DF19} collects some useful properties on multivariate
Orlicz spaces, Orlicz Hearts and K\"{o}the duals.

\begin{assumption}
\label{DynRMA2} $L^{\Phi }=L^{\phi _{1}}\times \dots \times L^{\phi _{N}}$.
\end{assumption}

Assumption \ref{DynRMA2} is a request on the utility functions we allow for.
It can be rephrased as: if for $X\in (L^{0}\left( (\Omega ,\mathcal{F},{%
\mathbb{P}});[-\infty ,+\infty ]\right) )^{N}$ there exist $\lambda
_{1},\dots ,\lambda _{N}>0$ such that $\mathbb{E}_{\mathbb{P}}\left[
u_{j}(-\lambda _{j}\left\vert X^{j}\right\vert )\right] >-\infty $, then
there exists $\alpha >0$ such that $\mathbb{E}_{\mathbb{P}}\left[ \Lambda
(-\alpha \left\vert X\right\vert )\right] >-\infty $. This request is rather
weak and there are many examples of choices of $U$ and $\Lambda $ that
guarantee this condition is met, see \cite{DF19} Proposition 7.3. Note
however that this is not a request on the topological spaces, but just an
integrability requirement, and it is automatically satisfied if $\Lambda =0.$
%
%

\begin{proposition}
\label{DynRMpropoptisint} Suppose Assumption \ref{DynRMA2} is fulfilled.
Then for any ${\mathbb{Q}}\in \mathscr{Q}_{\mathcal{G}}^{1}$ (defined in %
\eqref{DynRMDefq1cg}), any optimum $\widehat{Y}$ from Theorem \ref%
{DynRMmainthmrhoinfty} satisfies $\widehat{Y}\in L^{1}((\Omega ,\mathcal{F},{%
\mathbb{Q}}^{1})\times \dots \times L^{1}((\Omega ,\mathcal{F},{\mathbb{Q}}%
^{N})$ and 
\begin{equation*}
\sum_{j=1}^{N}\mathbb{E}_{\mathbb{Q}^{j}}\left[ \widehat{Y}^{j}\middle|%
\mathcal{G}\right] \leq \sum_{j=1}^{N}\widehat{Y}^{j}\,.
\end{equation*}
\end{proposition}

\begin{proof}
Postponed to Section \ref{DynRmsecproofs11}.
\end{proof}


\subsection{Optimization with a fixed measure ${\mathbb{Q}}\in\mathscr{Q}_%
\mathcal{G}^1$}

\label{DynRMsecoptwithameasure} We will now study a counterpart of $\rho _{%
\mathcal{G}}$ which can be obtained by fixing a vector of pricing measures ${%
\mathbb{Q}}=[{\mathbb{Q}}^{1},\dots ,{\mathbb{Q}}^{N}]$. For ${\mathbb{Q}}%
\in \mathscr{Q}_{\mathcal{G}}^{1}$ we define the following optimization
problem:

\begin{equation}
\rho _{\mathcal{G}}^{\mathbb{Q}}(X):=\essinf\left\{ \sum_{j=1}^{N}\mathbb{E}%
_{\mathbb{Q}^{j}}\left[ Y^{j}\middle|\mathcal{G}\right] \mid Y\in (L^{\infty
}(\Omega ,\mathcal{F},{\mathbb{P}}))^{N}\text{ and }\mathbb{E}_{\mathbb{P}}%
\left[ U\left( X+Y\right) \middle|\mathcal{G}\right] \geq B\right\} .
\label{DynRMdefrhoq}
\end{equation}

Notice that in the problem \eqref{DynRMdefrhoq} the constraint $%
\sum_{j=1}^{N}Y^{j}\in L^{\infty }(\mathcal{G})$ does not appear anymore.
The problem still makes sense, however, since a valuation of $Y$ is now
assigned by the pricing vector ${\mathbb{Q}}$. By the fairness condition %
\eqref{DynRMfairsums}, it is easy to verify that $\rho _{\mathcal{G}}(X)\leq
\rho _{\mathcal{G}}^{\mathbb{Q}}(X)$ for all ${\mathbb{Q}}\in \mathscr{Q}_{%
\mathcal{G}}^{1}$ and $X\in (L^{\infty }(\mathcal{F}))^{N}$. Such a bound is
actually tight, as the following Theorem \ref{DynRMoptforrhoq} shows. The
result we are to present will also be useful when studying equilibrium
properties for primal and dual optima of $\rho _{\mathcal{G}}$ in Section %
\ref{DynRMsecmSORTE}. In order to state the result more easily, and since we
will need to change underlying probability measures, some additional
notation is in place: from now on, given a vector of probability measures ${%
\mathbb{Q}}=[{\mathbb{Q}}^{1},\dots ,{\mathbb{Q}}^{N}]$ and a number $p\in
\{0,1\}$, we set 
\begin{equation*}
L^{p}(\mathcal{F},{\mathbb{Q}}):=L^{p}(\Omega ,\mathcal{F},{\mathbb{Q}}%
^{1})\times \dots \times L^{p}(\Omega ,\mathcal{F},{\mathbb{Q}}^{N})\,.
\end{equation*}%
Similarly, when some confusion might arise, we will write explicitly also
the measure ${\mathbb{P}}$, that is we will use $L^{p}(\mathcal{F},{\mathbb{P%
}}):=L^{p}(\Omega ,\mathcal{F},{\mathbb{P}})$ in place of the shortened $%
L^{p}(\mathcal{F})$.

\begin{theorem}
\label{DynRMoptforrhoq} Suppose Assumption \ref{DynRMA2} is fulfilled and
let $X\in (L^{\infty }(\mathcal{F},{\mathbb{P}}))^{N}$. Then for any optimum 
$\widehat{{\mathbb{Q}}}\in\mathscr{Q}_\mathcal{G}^1$ of %
\eqref{DynRMeqdualreprshorfall}, the optimum $\widehat{Y}\in \mathscr{C}_{%
\mathcal{G}}$ of Theorem \ref{DynRMmainthmrhoinfty} is an optimum for $\rho
_{\mathcal{G}}^{\widehat{{\mathbb{Q}}}}(X)$ in the following extended sense: 
$\widehat{Y}\in (L^{1}(\mathcal{F},{\mathbb{P}}))^{N}\cap \bigcap_{{\mathbb{Q%
}}\in \mathscr{Q}_{\mathcal{G}}^{1}}L^{1}(\mathcal{F},{\mathbb{Q}}),$ $%
\mathbb{E}_{\mathbb{P}}\left[ U\left( X+\widehat{Y}\right) \middle|\mathcal{G%
}\right] \geq B$ and 
\begin{align}
\rho _{\mathcal{G}}^{\widehat{{\mathbb{Q}}}}(X) &=\essinf\left\{
\sum_{j=1}^{N}\mathbb{E}_{\widehat{{\mathbb{Q}}}^{j}}\left[ Y^{j}\middle|%
\mathcal{G}\right] \mid Y\in (L^{1}(\mathcal{F},{\mathbb{P}}))^{N}\cap
\bigcap_{{\mathbb{Q}}\in \mathscr{Q}_{\mathcal{G}}^{1}}L^{1}(\mathcal{F},{%
\mathbb{Q}}),\,\mathbb{E}_{\mathbb{P}}\left[ U\left( X+Y\right) \middle|%
\mathcal{G}\right] \geq B\right\}  \label{new} \\
&=\sum_{j=1}^{N}\mathbb{E}_{\widehat{{\mathbb{Q}}}^{j}}\left[ \widehat{Y}^{j}%
\middle|\mathcal{G}\right] =\rho _{\mathcal{G}}\left( X\right) .  \notag
\end{align}%
%
%
%
%
%
%
\end{theorem}

\begin{proof}
See Section \ref{DynRmsecproofs11}.
\end{proof}

\section{The exponential case}

\label{DynRMsectionexp}

A specific, rather canonical choice for the multivariate utility $U$ is the
aggregation of univariate exponential utility functions for single agents.
It allows for obtaining explicit formulas for $\rho_\mathcal{G}$, as well
for the corresponding optima.


%
Throughout the whole Section \ref{DynRMsectionexp} we take $%
u_{j}(x)=-e^{-\alpha _{j}x},j=1,\dots ,N$ for real numbers $\alpha
_{1},\dots ,\alpha _{N}>0$ and $\Lambda =0$. We set:%
\begin{equation}
\overline{X}:=\sum_{j=1}^{N}X^{j}\text{; \ }\beta :=\sum_{j=1}^{N}\frac{1}{%
\alpha _{j}};\,\,\,A_{j}:=\frac{1}{\alpha _{j}}\log \left( \frac{1}{\alpha
_{j}}\right) ;\,\,\,A:=\sum_{j=1}^{N}A_{j}\,.  \label{DynRMconstantsexp}
\end{equation}%
We consider only the case $\mathscr{B}_{\mathcal{G}}=\mathscr{D}_{\mathcal{G}%
}$ (recall \eqref{DynRMdefDG} and \eqref{DynRMdefCG}), which corresponds to
the case of full sharing among all agents in the system (i.e., the extreme
case of one single group, as described in Item (i) on page \pageref%
{DynRMC0ginfty}).

\begin{theorem}
\label{DynRMthmformulasgeneral}Consider a general sub $\sigma $-algebra $%
\mathcal{G}\subseteq \mathcal{F}$, $X\in (L^{\infty }(\mathcal{F}))^{N}$, $%
B\in L^{\infty }(\mathcal{G}),$ with $\esssup{(B)}<0$ and $\rho _{\mathcal{G}%
}$ defined in (\ref{DynRMdefrcond}). Then 
\begin{equation*}
\rho _{\mathcal{G}}\left( X\right) =\beta \log \left( -\frac{\beta }{B}%
\mathbb{E}_{\mathbb{P}}\left[ \exp \left( -\frac{\overline{X}}{\beta }%
\right) \middle|\mathcal{G}\right] \right) -A;
\end{equation*}%
$\widehat{Y}=[\widehat{Y}^{1},...,\widehat{Y}^{N}]\in (L^{\infty }(\mathcal{F%
}))^{N}$ is an optimal allocation for $\rho _{\mathcal{G}}\left( X\right) $
and $\widehat{{\mathbb{Q}}}=[\widehat{{\mathbb{Q}}}^{1},...,\widehat{{%
\mathbb{Q}}}^{N}]$ is an optimum for the dual representation of $\rho _{%
\mathcal{G}}\left( X\right) $, where for $j=1,\dots ,N$ 
\begin{align}
\widehat{Y}^{j}& :=-X^{j}+\frac{1}{\beta \alpha _{j}}\left( \overline{X}%
+\rho _{\mathcal{G}}\left( X\right) +A\right) -A_{j},  \label{YY} \\
\frac{\mathrm{d}\widehat{{\mathbb{Q}}}^{j}}{\mathrm{d}{\mathbb{P}}}=\frac{%
\mathrm{d}\widehat{{\mathbb{Q}}}}{\mathrm{d}{\mathbb{P}}}& :=\frac{\exp
\left( -\frac{\overline{X}}{\beta }\right) }{\mathbb{E}_{\mathbb{P}}\left[
\exp \left( -\frac{\overline{X}}{\beta }\right) \middle|\mathcal{G}\right] }%
\,.  \label{QQ}
\end{align}
\end{theorem}

\begin{proof}
See Section \ref{DynRMsecproofsexp}.
\end{proof}

\begin{remark}
In order to avoid more complex notations and lengthy proofs, we provided the
explicit formulas only for the case $\mathscr{B}_{\mathcal{G}}=\mathscr{D}_{%
\mathcal{G}}$. The reader may obtain similar formulas for the cluster cases
in Example \ref{DynRMexChcond}, using the corresponding deterministic
formulas in Biagini et al, "On fairness of Systemic Risk Measures",
arXiv:1803.09898v3, 2018, Section 6.
\end{remark}

\subsection{Time consistency}

A rather natural issue we now wish to tackle is whether a consistency or
concatenation property can be associated to the Conditional Shortfall
Systemic Risk Measures, at least in the exponential case where explicit
computations are feasible. More precisely, we consider now two sub $\sigma $%
-algebras $\mathcal{H}\subseteq \mathcal{G}\subseteq \mathcal{F}$ and study
the relations between primal and dual allocations corresponding to the
various possible risk measurements (from $\mathcal{F}$ to $\mathcal{H}$,
from $\mathcal{G}$ to $\mathcal{H}$, and combinations). In this subsection
we will need to exploit explicitly the dependence of optimal allocations and
minimax measures given by Theorem \ref{DynRMthmformulasgeneral} on initial
datum and sub $\sigma $-algebras. To fix the notation, given $X\in
(L^{\infty }(\mathcal{F}))^{N}$ and $\mathcal{G}\subseteq \mathcal{F}$ we
define for each $k=1,\dots ,N$:

\begin{align}
\widehat{Y}^{k}\left( \mathcal{G},\,X\right)& :=-X^{k}+\frac{1}{\beta \alpha
_{k}}\left( \overline{X}+\rho _{\mathcal{G}}\left( X\right) +A\right)
-A_{k}\,,  \label{DynRMdefopty} \\
\frac{\mathrm{d}\widehat{{\mathbb{Q}}}^{k}\left( \mathcal{G},\,X\right) }{%
\mathrm{d}{\mathbb{P}}}&:=\frac{\exp \left( -\frac{\overline{X}}{\beta }%
\right) }{\mathbb{E}_{\mathbb{P}}\left[ \exp \left( -\frac{\overline{X}}{%
\beta }\right) \middle|\mathcal{G}\right] }\, ,  \label{DynRMdefoptq} \\
\widehat{a}^{k}\left( \mathcal{G},\,X\right) &:=\mathbb{E}_{\widehat{{%
\mathbb{Q}}}^{k}\left( \mathcal{G},\,X\right) }[\widehat{Y}^{k}\left( 
\mathcal{G},\,X\right) |\mathcal{G}]\,,  \label{DynRMdefopta}
\end{align}

\begin{equation}  \label{DynRMdefrcondexp}
\rho_{\mathcal{G}}\left(X\right)=\beta \log\left(-\frac{\beta}{B}\mathbb{E}_%
\mathbb{P} \left[\exp\left(-\frac{\overline{X}}{\beta}\right)\middle|%
\mathcal{G}\right]\right)-A=\sum_{j=1}^N\widehat{Y}^{j}\left(\mathcal{G}%
,\,X\right)=\sum_{j=1}^N\widehat{a}^{j}\left(\mathcal{G},\,X\right)\,.
\end{equation}

\begin{theorem}
\label{DynRMthmconsistency}Let $X\in (L^{\infty }(\mathcal{F}))^{N}$. The
following time consistency property holds whenever $B\in L^{\infty }(%
\mathcal{H})$ is given: for every $k=1,\dots ,N$ 
\begin{align}
\widehat{Y}^{k}\left( \mathcal{H},\,-\widehat{Y}^{{}}\left( \mathcal{G}%
,\,X\right) \right) & =\widehat{Y}^{k}\left( \mathcal{H},\,X\right) +%
\widehat{Y}^{k}\left( \mathcal{H},\,0\right) \,,  \label{DynRMconsisty} \\
\frac{\mathrm{d}\widehat{{\mathbb{Q}}}^{k}}{\mathrm{d}{\mathbb{P}}}\left( 
\mathcal{G},\,X\right) \frac{\mathrm{d}\widehat{{\mathbb{Q}}}^{k}}{\mathrm{d}%
{\mathbb{P}}}\left( \mathcal{H},\,-\widehat{Y}^{{}}\left( \mathcal{G}%
,\,X\right) \right) & =\frac{\mathrm{d}\widehat{{\mathbb{Q}}}^{k}}{\mathrm{d}%
{\mathbb{P}}}\left( \mathcal{G},\,X\right) \frac{\mathrm{d}\widehat{{\mathbb{%
Q}}}^{k}}{\mathrm{d}{\mathbb{P}}}\left( \mathcal{H},\,-\widehat{a}%
^{{}}\left( \mathcal{G},\,X\right) \right) =\frac{\mathrm{d}\widehat{{%
\mathbb{Q}}}^{k}}{\mathrm{d}{\mathbb{P}}}\left( \mathcal{H},\,X\right) \,,
\label{DynRMconsistq} \\
\widehat{a}^{k}\left( \mathcal{H},\,-\widehat{a}^{{}}\left( \mathcal{G}%
,\,X\right) \right) & =\widehat{a}^{k}\left( \mathcal{H},\,X\right) +%
\widehat{a}^{k}\left( \mathcal{H},\,0\right) \,.  \label{DynRMconsista}
\end{align}
\end{theorem}

\begin{proof}
See Section \ref{DynRMsecproofsexp}.
\end{proof}

\begin{remark}
\label{DynRMremarksetvalued2} Various concepts of time consistency have
already been explored in the literature for the dynamic and set-valued case.
In particular, those of \textquotedblleft time consistency\textquotedblright
in \cite{TaharLepinette14} and \textquotedblleft multiportfolio time
consistency\textquotedblright in \cite{ChenHu18} and \cite%
{FeinsteinRudloff15a}, were shown to be equivalent under mild assumptions in 
\cite{FeinsteinRudloff15}. Let us point out that in these approaches
consistency is required for the whole set of eligible portfolio that covers
the risk of $X$. Instead, we only request consistency of particular
allocations, enjoying some peculiar optimality property, as well as for the
dual optimizers. Furthermore, as mentioned already after Theorem \ref%
{DynRMmainthmrhoinfty} and adopting the same notation introduced there for $%
\mathcal{F}_{t}$, we use terminal-time, random allocations for securing the
system. This enlightens a further difference from the aforementioned works
where the whole set of allocations is required to be $\mathcal{F}_{t}$%
-measurable, making the properties hardly comparable. Possible links might
become clearer with a more detailed inspections of the properties of our
allocations $\widehat{a}$, and we wish to leave this topic for further
research.
\end{remark}

The proof of Theorem \ref{DynRMthmconsistency} is entirely based on the
availability of explicit formulas, as well as on the nice combinations of
logarithms and exponentials one sees also in the univariate case.

\section{Conditional Shortfall Systemic Risk Measures and equilibrium:
dynamic mSORTE}

\label{DynRMsecmSORTE} In the papers \cite{BDFFM} and \cite{DF19} the
equilibrium concepts of Systemic Optimal Risk Transfer Equilibrium (SORTE)
and of its multivariate extension Multivariate Systemic Optimal Risk
Transfer Equilibrium (mSORTE) were introduced and analyzed in a static
setup. We refer the reader to these papers for the economic motivation,
applications and unexplained notation. Here we show that a generalization to
the conditional setting is possible and prove the existence of a time
consistent family of dynamic mSORTE in the exponential setup. Consider a
multivariate utility function $U$. For each $j=1,...,N$ consider a vector
subspace $L_{\mathcal{F}}^{j}$ with $L^{\infty }(\mathcal{F})\subseteq L_{%
\mathcal{F}}^{j}\subseteq L^{0}(\mathcal{F}\mathbf{)}$ and set 
\begin{equation*}
L_{\mathcal{F}}\mathbf{:=}L_{\mathcal{F}}^{1}\times ...\times L_{\mathcal{F}%
}^{N}\mathbf{\subseteq }(L^{0}(\Omega ,\mathcal{F},{\mathbb{P}}))^{N}.
\end{equation*}%
With 
\begin{equation*}
\mathcal{M}\subseteq \mathscr{Q}_{\mathcal{G}}
\end{equation*}%
we will denote a subset of probability vectors.

\begin{remark}
We impose the condition $\mathcal{M}\subseteq \mathscr{Q}_{\mathcal{G}}$,
which implies that for every ${\mathbb{Q}}\in \mathscr{Q}_{\mathcal{G}}$ and
for every $j=1,\dots ,N$, ${\mathbb{Q}}^{j}$ coincides with ${\mathbb{P}}$
on $\mathcal{G}$ and 
\begin{equation*}
(L^{1}(\mathcal{G},{\mathbb{P}}))^{N}=L^{1}(\mathcal{G},{\mathbb{Q}}%
):=L^{1}(\Omega ,\mathcal{F},{\mathbb{Q}}^{1})\times \dots \times
L^{1}(\Omega ,\mathcal{F},{\mathbb{Q}}^{N})\,.
\end{equation*}
\end{remark}

For $(Y,{\mathbb{Q}},\,\alpha ,\,A)\in (L_{\mathcal{F}}\cap L^{1}(\mathcal{F}%
,{\mathbb{Q}}))\times \mathcal{M}\mathbf{\times }(L^{1}(\mathcal{G},{\mathbb{%
P}}))^{N}\mathbf{\times }L^{\infty }(\mathcal{G})$ define for $j=1,\dots ,N$ 
\begin{align}
Y^{[-j]}& :=[Y^{1},\dots ,Y^{j-1},Y^{j+1},\dots ,Y^{N}]\in L^{0}(\mathcal{F},%
{\mathbb{P}})^{N-1}\,,  \notag \\
\lbrack Y^{[-j]},Z]& :=[Y^{1},\dots ,Y^{j-1},Z,Y^{j+1},\dots
,Y^{N}],\,\,\,Z\in L^{0}(\mathcal{F},{\mathbb{P}})\,,  \notag \\
U_{j}^{Y^{[-j]}}(Z)& :=\mathbb{E}\left[ u_{j}(X^{j}+Z)\middle|\mathcal{G}%
\right] +\mathbb{E}\left[ \Lambda (X+[Y^{[-j]},Z])\middle|\mathcal{G}\right]
,\,\,\,Z\in L^{0}(\mathcal{F},{\mathbb{P}})\,,  \label{DynRMdefUj} \\
\mathbb{U}_{j}^{\,{\mathbb{Q}}^{j},Y^{[-j]}}(\alpha ^{j})& :=\esssup\left\{
U_{j}^{Y^{[-j]}}(Z)\mid Z\in L_{\mathcal{F}}^{j}\cap L^{1}(\Omega ,\mathcal{F%
},{\mathbb{Q}}^{j})\text{, }E_{{\mathbb{Q}}^{j}}[Z|\mathcal{G}]\leq \alpha
^{j}\right\} \,,  \label{DynRMUQ}
\end{align}%
and 
\begin{align}
T^{\mathbb{Q}}(\alpha )& :=\esssup\left\{ \mathbb{E}_{\mathbb{P}}\left[
U(X+Y)\middle|\mathcal{G}\right] \mid Y\in L_{\mathcal{F}}\cap L^{1}(%
\mathcal{F},{\mathbb{Q}}),\,E_{{\mathbb{Q}}^{j}}\left[ Y^{j}|\mathcal{G}%
\right] \leq \alpha ^{j},\,\forall j\right\} \,,  \label{DynRMTQ} \\
S^{\mathbb{Q}}(A)& :=\esssup\left\{ T^{\mathbb{Q}}(\alpha )\mid \alpha \in
(L^{1}(\mathcal{G},{\mathbb{P}}))^{N},\,\,\sum_{j=1}^{N}\alpha ^{j}\leq
A\right\} \,.  \label{DynRMSQ}
\end{align}%
Obviously, all such quantities depend also on $X\in (L^{\infty }(\mathcal{F}%
))^{N}$, but as $X$ will be kept fixed throughout the analysis, we may avoid
to explicitly specify this dependence in the notations. As $u_{1},\dots
,u_{N},\Lambda ,U$ are increasing we can replace, in the definitions %
\eqref{DynRMUQ}, \eqref{DynRMTQ}, \eqref{DynRMSQ}, the inequalities in the
budget constraints with equalities.

\begin{definition}
\label{DynRMstrongmsorte} The triple $(Y_{X},{\mathbb{Q}}_{X},\alpha
_{X})\in L_{\mathcal{F}}\mathbf{\times }\mathcal{M}\mathbf{\times }(L^{1}(%
\mathcal{G},{\mathbb{P}}))^{N}$ with $Y\in L^{1}({\mathcal{F},\mathbb{Q}}%
_{X})$ is a \textbf{Dynamic Multivariate Systemic Optimal Risk Transfer
Equilibrium} (Dynamic mSORTE) with budget $A\in L^{\infty }(\mathcal{G})$ if

\begin{enumerate}
\item $(Y_{X},\alpha _{X})$ is an optimum for 
\begin{equation}
\esssup_{\substack{ \alpha \in (L^{1}({\mathbb{P}},\mathcal{G}))^{N}  \\ %
\sum_{j=1}^{N}\alpha _{j}=A}}\left\{ \esssup\left\{ \mathbb{E}_{\mathbb{P}}%
\left[ U(X+Y)|\mathcal{G}\right] \mid Y\in L_{\mathcal{F}}\cap L^{1}(%
\mathcal{F},{\mathbb{Q}}_{X}),\,\mathbb{E}_{\mathbb{Q}_{X}^{j}}\left[ Y^{j}|%
\mathcal{G}\right] \leq \alpha ^{j},\,\,\forall j\right\} \right\} \,;
\label{DynRMeqmsorteparticular}
\end{equation}

\item $Y_{X}\in \mathscr{C}_\mathcal{G}$ and $\sum_{j=1}^{N}Y_{X}^{j}=A\,\, {%
\mathbb{P}}$-a.s..
\end{enumerate}
\end{definition}

\begin{theorem}
\label{DynRMthmallocissorte} Suppose Assumption \ref{DynRMreqint} and
Assumption \ref{DynRMA2} hold and let $X\in (L^{\infty }(\mathcal{F}))^{N}$.
Let $\widehat{Y}$ be the optimum of \ $\rho _{\mathcal{G}}$ in (\ref%
{DynRMdefrcond}) and let $\widehat{{\mathbb{Q}}}$ be an optimum of %
\eqref{DynRMeqdualreprshorfall}. Define $\widehat{\alpha }^{j}:=\mathbb{E}_{%
\widehat{{\mathbb{Q}}}^{j}}[\widehat{Y}^{j}|\mathcal{G}]$. Then $(\widehat{Y}%
,\widehat{{\mathbb{Q}}},\widehat{\alpha })$ is a Dynamic mSORTE for $%
\mathcal{M}:=\mathscr{Q}_{\mathcal{G}}^{1}$, $L_{\mathcal{F}}:=(L^{1}(%
\mathcal{F},{\mathbb{P}}))^{N}\cap \bigcap_{{\mathbb{Q}}\in \mathscr{Q}_{%
\mathcal{G}}^{1}}L^{1}(\mathcal{F},{\mathbb{Q}})$, $A:=\rho _{\mathcal{G}%
}\left( X\right) $.
\end{theorem}

\begin{proof}
$\,$

\textbf{STEP 1}: Item 2 of Definition \ref{DynRMstrongmsorte} is satisfied.
We start observing that by Theorem \ref{DynRMmainthmrhoinfty}, $\widehat{Y}%
\in \mathscr{C}_{\mathcal{G}}$ and trivially being an optimum it satisfies $%
\sum_{j=1}^{N}\widehat{Y}^{j}=\rho _{\mathcal{G}}\left( X\right) =:A$.

\textbf{STEP 2}: we prove that for any optimum $\widehat{{\mathbb{Q}}}\in%
\mathscr{Q}_\mathcal{G}^1$ of \eqref{DynRMeqdualreprshorfall}, the
optimization problem

\begin{equation}
\pi _{A}^{\mathcal{G},\widehat{{\mathbb{Q}}}}(X):=\esssup\left\{ \mathbb{E}_{%
\mathbb{P}}\left[ U\left( X+Y\right) \middle|\mathcal{G}\right] \,\,\middle%
|\,\,\begin{aligned} &Y\in (L^1(\mathcal{F},{\mathbb
P}))^N\cap\bigcap_{\probq\in\mathcal{Q}_\cG^1}L^1(\mathcal{F},{{\mathbb
Q}})\,\,\text{ and}\\
&\sum_{j=1}^N\mathbb{E}_{\widehat{\probq}^j}[Y^j|\mathcal{G}]\leq
A\end{aligned}\right\} \,  \label{DynRMpiextendedq}
\end{equation}

satisfies $\pi_A^{\mathcal{G},\widehat{{\mathbb{Q}}}}(X)=B$.

We start showing that the optimal allocation $\widehat{Y}$ for $\rho _{%
\mathcal{G}}\left( X\right) $ provided by Theorem \ref{DynRMmainthmrhoinfty}
satisfies $\sum_{j=1}^{N}\mathbb{E}_{\widehat{{\mathbb{Q}}}^{j}}[Y^{j}|%
\mathcal{G}]=A$ (directly from Theorem \ref{DynRMoptforrhoq}) and $\mathbb{E}%
_{\mathbb{P}}\left[ U\left( X+\widehat{Y}\right) \middle|\mathcal{G}\right]
=B$. To see the latter equality, observe that we already know that $\mathbb{E%
}_{\mathbb{P}}\left[ U\left( X+\widehat{Y}\right) \middle|\mathcal{G}\right]
\geq B$. If on a set $\Xi $ of positive measure we had that the inequality
is strict, we would have for some $N>0$ that $\Xi _{N}:=\{\mathbb{E}_{%
\mathbb{P}}\left[ U\left( X+\widehat{Y}\right) \middle|\mathcal{G}\right] >B+%
\frac{1}{N}\}\in \mathcal{G}$ has positive probability. By Assumption \ref%
{DynRMreqint} and (cDOM) we have 
\begin{equation*}
\mathbb{E}_{\mathbb{P}}\left[ U\left( X+\widehat{Y}-\frac{1}{H}\mathbf{1}%
\right) \middle|\mathcal{G}\right] \uparrow _{H}\mathbb{E}_{\mathbb{P}}\left[
U\left( X+\widehat{Y}\right) \middle|\mathcal{G}\right] \text{ on }\Xi
_{N}\,.
\end{equation*}
By Egorov Theorem \ref{DynRMegorovthm} there exists a $\Theta _{N}\in 
\mathcal{G}$, with both $\Theta _{N}\subseteq\Xi _{N}$ and ${\mathbb{P}}%
(\Theta _{N})>0$, on which the convergence above is uniform (in $H$).

Hence, definitely in $H$, $\mathbb{E}_{\mathbb{P}}\left[ U\left( X+\widehat{Y%
}-\frac{1}{H}1_{\Theta _{N}}\mathbf{1}\right) \middle|\mathcal{G}\right]
\geq B$. Putting things together, we have then 
\begin{equation*}
\widehat{Y}-\frac{1}{H}1_{\Theta _{N}}\mathbf{1}\in \mathscr{C}_{\mathcal{G}%
}\,,\,\,\,\,\,\,\,\mathbb{E}_{\mathbb{P}}\left[ U\left( X+\widehat{Y}%
-\frac1H 1_{\Theta _{N}}\mathbf{1}\right) \middle|\mathcal{G}\right] \geq
B\,.
\end{equation*}

Clearly then $\rho _{\mathcal{G}}\left( X\right) \leq \sum_{j=1}^{N}\widehat{%
Y}^{j}-\frac{N}{H}1_{\Theta _{N}}$. This in turns gives a contradiction,
since $\widehat{Y}$ is supposed to be an optimum for $\rho _{\mathcal{G}%
}\left( X\right) $.

Now we prove that $\pi _{A}^{\mathcal{G},\widehat{{\mathbb{Q}}}}(X)=B$. Take 
$\widehat{Y}$ as before. We stress that it satisfies $\widehat{Y}\in (L^{1}(%
\mathcal{F},{\mathbb{P}}))^{N}$ by Theorem \ref{DynRMmainthmrhoinfty} and $%
\widehat{Y}\in L^{1}(\mathcal{F},{\mathbb{Q}})$ for every ${\mathbb{Q}}\in 
\mathcal{Q}_{\mathcal{G}}^{1}$ by Proposition \ref{DynRMpropoptisint}. We
showed above that $\sum_{j=1}^{N}\mathbb{E}_{\widehat{{\mathbb{Q}}}^{j}}[%
\widehat{Y}^{j}|\mathcal{G}]=A$ and $\mathbb{E}_{\mathbb{P}}\left[ U\left( X+%
\widehat{Y}\right) \middle|\mathcal{G}\right] =B$. Hence by %
\eqref{DynRMpiextendedq} we have $\pi _{A}^{\mathcal{G},\widehat{{\mathbb{Q}}%
}}(X)\geq B$. 
Since the set over which we take the essential supremum to define $\pi _{A}^{%
\mathcal{G},\widehat{{\mathbb{Q}}}}(X)$ is upward directed, we can take a
maximizing sequence for $\pi _{A}^{\mathcal{G},\widehat{{\mathbb{Q}}}}(X)$,
call it $(Y_{n})_{n}$. W.l.o.g. we may suppose that $\sum_{j=1}^{N}\mathbb{E}%
_{\widehat{{\mathbb{Q}}}^{j}}[Y_{n}^{j}|\mathcal{G}]=A$. Suppose for $\Delta
:=\{\pi _{A}^{\mathcal{G},\widehat{{\mathbb{Q}}}}(X)>B\}$ we had ${\mathbb{P}%
}(\Delta )>0$. Then setting $\Delta _{N}:=\{\pi _{A}^{\mathcal{G},\widehat{{%
\mathbb{Q}}}}(X)>B+\frac{1}{N}\}\in \mathcal{G}$ we have ${\mathbb{P}}%
(\Delta _{N})>0$ for some $N$ big enough. By Egorov Theorem \ref%
{DynRMegorovthm}, we have that on a subset $\widetilde{\Delta }_{N}$ of $%
\Delta _{N}$, having positive probability, the pointwise convergence of $%
\mathbb{E}_{\mathbb{P}}\left[ U\left( X+Y_{n}\right) \middle|\mathcal{G}%
\right] $ to the essential supremum is uniform. Hence given $\varepsilon >0$
small enough, for $n$ big enough and for $\widetilde{Y}_{n}=Y_{n}-%
\varepsilon 1_{\widetilde{\Delta }_{N}}\mathbf{1}\in \mathscr{C}_{\mathcal{G}%
}\cap (L^{\infty }({\mathbb{P}},\mathcal{F}))^{N}$ we have $\mathbb{E}_{%
\mathbb{P}}\left[ U\left( X+\widetilde{Y}_{n}\right) \middle|\mathcal{G}%
\right] \geq B$. Clearly $\sum_{j=1}^{N}\mathbb{E}_{\widehat{{\mathbb{Q}}}%
^{j}}[\widetilde{Y}_{n}^{j}|\mathcal{G}]<A$ with positive probability, by
definition of $\widetilde{Y}_{n}$. By the definition of $\rho _{\mathcal{G}%
}^{\widehat{{\mathbb{Q}}}}$ in (\ref{DynRMdefrhoq}), we obtain that with
positive probability (i.e. on $\widetilde{\Delta }_{N}$) 
\begin{equation*}
\rho _{\mathcal{G}}^{\widehat{{\mathbb{Q}}}}(X)\leq \sum_{j=1}^{N}\mathbb{E}%
_{\widehat{{\mathbb{Q}}}^{j}}[\widetilde{Y}_{n}^{j}|\mathcal{G}]<A\,.
\end{equation*}%
We then get a contradiction to $A:=\rho _{\mathcal{G}}\left( X\right) $,
since by Theorem \ref{DynRMoptforrhoq} we have $\rho _{\mathcal{G}}\left(
X\right) =\rho _{\mathcal{G}}^{\widehat{{\mathbb{Q}}}}(X)$.

\textbf{STEP 3}: the optimal allocation $\widehat{Y}$ of $\rho _{\mathcal{G}%
}\left( X\right) $ given in Theorem \ref{DynRMmainthmrhoinfty} (which is an
optimum by Theorem \ref{DynRMoptforrhoq}) is an optimum for the RHS of %
\eqref{DynRMpiextendedq}. This follows trivially from the arguments in the
previous steps, as $\widehat{Y}\in (L^{1}(\mathcal{F},{\mathbb{P}}))^{N}$ by
Theorem \ref{DynRMmainthmrhoinfty}, $\widehat{Y}\in L^{1}(\mathcal{F},{%
\mathbb{Q}})$ for every ${\mathbb{Q}}\in \mathcal{Q}_{\mathcal{G}}^{1}$ by
Proposition \ref{DynRMpropoptisint}, and $\sum_{j=1}^{N}\mathbb{E}_{\widehat{%
{\mathbb{Q}}}^{j}}[Y^{j}|\mathcal{G}]=A$. Thus $\widehat{Y}$ satisfies the
constraints of RHS of \eqref{DynRMpiextendedq}. Moreover we proved in STEP 1
that $\mathbb{E}_{\mathbb{P}}\left[ U\left( X+\widehat{Y}\right) \middle|%
\mathcal{G}\right] =B=\pi _{A}^{\mathcal{G},\widehat{{\mathbb{Q}}}}(X)$.

\textbf{STEP 4}: conclusion. We easily see that $\widehat{Y}$ is an optimum
for 
\begin{equation*}
\esssup\left\{ \mathbb{E}_{\mathbb{P}}\left[ U\left( X+Y\right) \middle|%
\mathcal{G}\right] \mid Y\in L_{\mathcal{F}},\sum_{j=1}^{N}\mathbb{E}_{%
\widehat{{\mathbb{Q}}}^{j}}[Y^{j}|\mathcal{G}]\leq A\right\} =
\end{equation*}%
\begin{equation*}
\esssup_{\substack{ \alpha \in (L^{1}(\mathcal{G},{\mathbb{P}}))^{N}  \\ %
\sum_{j=1}^{N}\alpha _{j}=A}}\left\{ \esssup\left\{ \mathbb{E}_{\mathbb{P}}%
\left[ U\left( X+Y\right) \middle|\mathcal{G}\right] \mid Y\in L_{\mathcal{F}%
},\mathbb{E}_{\widehat{{\mathbb{Q}}}^{j}}[Y^{j}|\mathcal{G}]\leq \alpha
^{j}\,\forall \,j=1,\dots ,N\right\} \right\} \,.
\end{equation*}

Hence $(\widehat{Y},\widehat{\alpha})$ are optimum for %
\eqref{DynRMeqmsorteparticular}, and also Item 1 of Definition \ref%
{DynRMstrongmsorte} is satisfied. This completes the proof.
\end{proof}

\begin{corollary}
For the exponential utilities there exists a time consistent family of
Dynamic mSORTEs.
\end{corollary}

\begin{proof}
Follows from Theorem \ref{DynRMthmconsistency} and Theorem \ref%
{DynRMthmallocissorte}.
\end{proof}

\appendix

\section{Appendix}

\subsection{Miscellaneous results}

We recall the original Koml\'{o}s Theorem:

\begin{theorem}[Koml\'{o}s]
\label{DynRMkomlosoriginal} Let $(f_{n})_{n}\subseteq L^{1}(\Omega ,\mathcal{%
F},{\mathbb{P}})$ be a sequence with bounded $L^{1}$ norms. Then there
exists a subsequence $(f_{n_{k}})_{k}$ and a random variable $g$ in $%
L^{1}(\Omega ,\mathcal{F},{\mathbb{P}})$ such that for any further
subsequence the Cesaro means satisfy: 
\begin{equation*}
\frac{1}{N}\sum_{i\leq N}f_{n_{k_{i}}}\rightarrow _{N}g\,\,\,{\mathbb{P}}%
\text{- a.s.}.
\end{equation*}
\end{theorem}

\begin{corollary}
\label{DynRMcorkomplosmultidim}Consider probability measures ${\mathbb{P}}%
_{1},\dots ,{\mathbb{P}}_{N}\ll {\mathbb{P}}$. Let a sequence $(X_{n})_{n}$
be given in $L^{1}(\Omega ,\mathcal{F},{\mathbb{P}}_{1})\times \dots \times
L^{1}(\Omega ,\mathcal{F},{\mathbb{P}}_{N})$ such that 
\begin{equation*}
\sup_{n}\sum_{j=1}^{N}\mathbb{E}_{\mathbb{P}}\left[ \left\vert
X_{n}^{j}\right\vert \frac{\mathrm{d}{\mathbb{P}}_{j}}{\mathrm{d}{\mathbb{P}}%
}\right] <\infty \,.
\end{equation*}%
Then there exists a subsequence $(X_{n_{h}})_{h}$ and an $Y\in L^{1}(\Omega ,%
\mathcal{F},{\mathbb{P}}_{1})\times \dots \times L^{1}(\Omega ,\mathcal{F},{%
\mathbb{P}}_{N})$ such that every further subsequence $(X_{n_{h_{k}}})_{k}$
satisfies 
\begin{equation*}
\frac{1}{K}\sum_{k=1}^{K}X_{n_{h_{k}}}^{j}\xrightarrow[K\rightarrow+\infty]{%
\probp_j-\text{a.s.}}Y^{j}\,\,\,\,\forall \,\,j=1,\dots N\,.
\end{equation*}
\end{corollary}

\begin{proof}
See \cite{DF19} Corollary A.12.
\end{proof}

\begin{theorem}[Egorov]
\label{DynRMegorovthm} Let $(X_{n})_{n}$ be a sequence in $L^{0}(\Omega ,%
\mathcal{F},{\mathbb{P}})$ almost surely converging to $X\in L^{0}(\Omega ,%
\mathcal{F},{\mathbb{P}})$. For every $\varepsilon >0$ there exists $%
A_{\varepsilon }\in \mathcal{F}$ with ${\mathbb{P}}(A_{\varepsilon
})<\varepsilon $ satisfying 
\begin{equation*}
\left\Vert \left( X_{n}-X\right) 1_{(A_{\varepsilon })^{c}}\right\Vert
_{\infty }\rightarrow _{n}0.
\end{equation*}
\end{theorem}

\begin{proof}
See \cite{Aliprantis} Theorem 10.38.
\end{proof}

%

\begin{remark}
\label{DynRMremcesaro} Observe that for any sequence of real numbers $%
(a_{n})_{n}$ converging to an $a\in {\mathbb{R}}$ and for any sequence $%
(N_{h})_{h}\uparrow +\infty $ we have $\frac{1}{N_{h}}\sum_{j\leq
N_{h}}a_{j}\rightarrow _{h}a$. This can be seen as follows: for $\varepsilon
>0$ fixed, take $K$ s.t. $\left\vert a_{j}-a\right\vert \leq \varepsilon $
for all $j\geq K$. Take $h$ big enough to have $N_{h}>K$. Then 
\begin{equation*}
\left\vert \frac{1}{N_{h}}\sum_{j\leq N_{h}}a_{j}-a\right\vert \leq \frac{1}{%
N_{h}}\sum_{j\leq N_{h}}\left\vert a_{j}-a\right\vert \leq \frac{K}{N_{h}}%
\sup_{j\leq K}\left\vert a_{j}-a\right\vert +\frac{N_{h}-K}{N_{h}}\varepsilon
\end{equation*}%
and we can send $h$ to infinity.
\end{remark}

\subsubsection{Essential suprema and infima}

\label{DynRMsecessupandinf} In this Section \ref{DynRMsecessupandinf} we
write $L^0(\Omega,\mathcal{F},{\mathbb{P}};[-\infty,+\infty ])$ for the set
of (equivalence classes of) $[-\infty,+\infty]$-valued random variables. $%
L^0(\Omega,\mathcal{F},{\mathbb{P}};[0,+\infty ))$ is defined analogously.

\begin{proposition}
\label{DynRMrefesssup} Let $\mathcal{A},\mathcal{B}\subseteq L^{0}(\Omega ,%
\mathcal{F},{\mathbb{P}};[-\infty ,+\infty ])$ be nonempty, $\lambda \in
L^{0}(\Omega ,\mathcal{F},{\mathbb{P}};[0,+\infty ))$, $f:\mathcal{A}\times 
\mathcal{B}\rightarrow L^{0}(\Omega ,\mathcal{F},{\mathbb{P}};[-\infty
,+\infty ])$, $g:\mathcal{A}\rightarrow L^{0}(\Omega ,\mathcal{F},{\mathbb{P}%
};[-\infty ,+\infty ])$, a sequence $(\alpha _{n})_{n}\subseteq \mathcal{A}$
be given. Then: 
\begin{equation*}
\esssup_{(\alpha ,\beta )\in \mathcal{A}\times \mathcal{B}}(\alpha +\beta )=%
\esssup_{\alpha \in \mathcal{A}}\alpha +\esssup_{\beta \in \mathcal{B}}\beta
=\esssup_{\alpha \in \mathcal{A}}\left( \alpha +\esssup_{\beta \in \mathcal{B%
}}\beta \right) \,;
\end{equation*}%
\begin{equation*}
\esssup_{(\alpha ,\beta )\in \mathcal{A}\times \mathcal{B}}f(\alpha ,\beta )=%
\esssup_{\alpha \in \mathcal{A}}\esssup_{\beta \in \mathcal{B}}f(\alpha
,\beta )=\esssup_{\beta \in \mathcal{B}}\esssup_{\alpha \in \mathcal{A}%
}f(\alpha ,\beta )\,;
\end{equation*}%
\begin{equation*}
\esssup_{\alpha \in \mathcal{A}}\lambda g(\alpha )=\lambda \esssup_{\alpha
\in \mathcal{A}}g(\alpha );\quad \esssup_{\alpha \in \mathcal{A}}\alpha \geq
\limsup_{n}\alpha _{n}\,;\quad \esssup_{\alpha \in \mathcal{A}}-g(\alpha )=-%
\essinf_{\alpha \in \mathcal{A}}g(\alpha ).
\end{equation*}
\end{proposition}

\subsubsection{Additional properties of multivariate utility functions}

\label{DynRMaddmultut}Recall that we are working under Standing Assumption
I. \label{DynRMsecmultiut}

\begin{lemma}
\label{DynRMlemmacontrolwithline}$\,$

(i) There exist $a>0$, $b\in {\mathbb{R}}$ such that 
\begin{equation*}
U(x)\leq
a\sum_{j=1}^{N}x^{j}+a\sum_{j=1}^{N}(-(x^{j})^{-})+b\,\,\,\,\,\forall \,x\in 
{\mathbb{R}}^{N}\,.
\end{equation*}

(ii) There exist $a>0$, $b\in {\mathbb{R}}$ such that 
\begin{equation*}
U(x)\leq a\sum_{j=1}^{N}x^{j}+b\,\,\,\,\,\forall \,x\in {\mathbb{R}}^{N}\,.
\end{equation*}

(iii) For every $\varepsilon >0$ there exist a constant $b_{\varepsilon }$
such that 
\begin{equation*}
U(x)\leq \varepsilon \sum_{j=1}^{N}(x^{j})^{+}+b_{\varepsilon
}\,\,\,\,\,\forall x\in {\mathbb{R}}^{N}.
\end{equation*}
\end{lemma}

\begin{proof}
(i) See \cite{DF19} Lemma 3.5 (as pointed out before, a function $U$
satisfying Standing Assumption I in this paper is well controlled according
to \cite{DF19} Definition 3.4, as shown in \cite{DF19} Proposition 7.1).

(ii) Use Item (i) and observe that $a\sum_{j=1}^{N}x^{j}+a%
\sum_{j=1}^{N}(-(x^{j})^{-})+b\leq a\sum_{j=1}^{N}x^{j}+b$.

(iii) See \cite{DF19} Lemma 3.5.
\end{proof}

\begin{lemma}
\label{DynRMlemmakomlos} Let $(Z_{n})\in (L^{0}(\Omega ,\mathcal{F},{\mathbb{%
P}}))^{N}$ satisfy $\mathbb{E}_{\mathbb{P}}\left[ U(Z_{n})\right] \geq B$
for all $n\in \mathbb{N}$ and for some constant $B\in {\mathbb{R}}$. {If $%
\sup_{n}\left\vert \sum_{j=1}^{N}\mathbb{E}_{\mathbb{P}}\left[ Z_{n}^{j}%
\right] \right\vert <+\infty $ then $\sup_{n}\sum_{j=1}^{N}\mathbb{E}_{%
\mathbb{P}}\left[ \left\vert Z_{n}^{j}\right\vert \right] <\infty $.}

{\label{DynRMlemmabdd1}}
\end{lemma}

\begin{proof}
See \cite{DF19} Lemma A.1.
\end{proof}

\begin{lemma}
\label{DynRMproputilfat} Suppose $(Z_n)_n$ is a sequence in $(L^1(\Omega, 
\mathcal{F},{\mathbb{P}}))^N$. Suppose furthermore that the following
conditions are met for some $B\in L^\infty(\Omega, \mathcal{G},{\mathbb{P}})$%
:

\begin{enumerate}
\item \label{DynRMpointwisebdd} $\sup_{n}\left\vert \sum_{j=1}^{N}\mathbb{E}%
_{\mathbb{P}}\left[ Z_{n}^{j}\middle|\mathcal{G}\right] \right\vert <+\infty 
$ ${\mathbb{P}}$-a.s.;

\item \label{DynRMbudgetc}{$\inf_{n}\mathbb{E}_{\mathbb{P}}\left[ U\left(
Z_{n}\right) \middle|\mathcal{G}\right] \geq B$ ${\mathbb{P}}$-a.s.;}

\item \label{DynRMconverge}$Z_n\rightarrow_nZ\,{\mathbb{P}}$-a.s..
\end{enumerate}

Then $\mathbb{E}_\mathbb{P} \left[U\left(Z\right)\middle| \mathcal{G}\right]%
\geq B$ a.s..
\end{lemma}

\begin{proof}
$\,$

\textbf{STEP 1}: $\sup_{n}\left( \sum_{j=1}^{N}\mathbb{E}_{\mathbb{P}}\left[
(Z_{n}^{j})^{+}\middle|\mathcal{G}\right] \right) <+\infty \,\,{\mathbb{P}}-$%
a.s..

Define the sets 
\begin{equation*}
A^{+}:=\left\{ \sup_{n}\sum_{j=1}^{N}\mathbb{E}_{\mathbb{P}}\left[
(Z_{n}^{j})^{+}\middle|\mathcal{G}\right] =+\infty \right\}
\,\,\,\,\,A^{-}:=\left\{ \sup_{n}\sum_{j=1}^{N}\mathbb{E}_{\mathbb{P}}\left[
(Z_{n}^{j})^{-}\middle|\mathcal{G}\right] =+\infty \right\} \,.
\end{equation*}

We prove that ${\mathbb{P}}(A^{-})=0$ : suppose by contradiction that ${%
\mathbb{P}}(A^{-})>0$. Apply Item \ref{DynRMbudgetc} together with the fact
that $A^{-}$ is $\mathcal{G}$ measurable to see that for some $a>0,b\in {%
\mathbb{R}}$ 
\begin{equation*}
B1_{A^{-}}\leq \mathbb{E}_{\mathbb{P}}\left[ U\left( Z_{n}\right) \middle|%
\mathcal{G}\right] 1_{A^{-}}\overset{\text{Lemma }\ref%
{DynRMlemmacontrolwithline}(i)}{\leq }\left( a\sum_{j=1}^{N}\mathbb{E}_{%
\mathbb{P}}\left[ Z_{n}^{j}\mid \mathcal{G}\right] +a%
\sum_{j=1}^{N}(Z_{n}^{j})^{-}+b\right) 1_{A^{-}}
\end{equation*}%
which is a contradiction, by definition of $A^{-}$ and boundedness of $B$.
Hence ${\mathbb{P}}(A^{-})=0$. By Item \ref{DynRMpointwisebdd}, together
with 
\begin{equation*}
\sum_{j=1}^{N}\mathbb{E}_{\mathbb{P}}\left[ Z_{n}^{j}\middle|\mathcal{G}%
\right] =\sum_{j=1}^{N}\mathbb{E}_{\mathbb{P}}\left[ (Z_{n}^{j})^{+}\middle|%
\mathcal{G}\right] -\sum_{j=1}^{N}\mathbb{E}_{\mathbb{P}}\left[
(Z_{n}^{j})^{-}\middle|\mathcal{G}\right]
\end{equation*}%
we have that the symmetric difference $A^{+}\Delta A^{-}$ is ${\mathbb{P}}$%
-null (equivalently $1_{A^{+}}=1_{A^{-}}$), so that from ${\mathbb{P}}%
(A^{-})=0$ we get ${\mathbb{P}}(A^{+})=0$, and the claim follows.

\medskip

\textbf{STEP 2}: Fatou Lemma and conclusion.

By Lemma \ref{DynRMlemmacontrolwithline} (iii) for every $\varepsilon >0$
there exist $b_{\varepsilon }>0$ such that $\Gamma _{\varepsilon
}(x):=-U(x)+\varepsilon \sum_{j=1}^{N}(x^{j})^{+}+b_{\varepsilon }\geq 0$
for all $x\in {\mathbb{R}}^{N}$. By Fatou Lemma ($\Gamma _{\varepsilon }$ is
continuous) we have that 
\begin{equation*}
-\mathbb{E}_{\mathbb{P}}\left[ U\left( Z\right) \middle|\mathcal{G}\right]
+\varepsilon \sum_{j=1}^{N}\mathbb{E}_{\mathbb{P}}\left[ (Z^{j})^{+}\middle|%
\mathcal{G}\right] +b_{\varepsilon }=\mathbb{E}_{\mathbb{P}}\left[ \Gamma
_{\varepsilon }(Z)\middle|\mathcal{G}\right] =\mathbb{E}_{\mathbb{P}}\left[
\liminf_{n}\Gamma _{\varepsilon }(Z_{n})\middle|\mathcal{G}\right]
\end{equation*}%
\begin{equation*}
\leq \liminf_{n}\mathbb{E}_{\mathbb{P}}\left[ \Gamma _{\varepsilon }(Z_{n})%
\middle|\mathcal{G}\right] =\liminf_{n}\left( -\mathbb{E}_{\mathbb{P}}\left[
U\left( Z_{n}\right) \middle|\mathcal{G}\right] +\varepsilon \sum_{j=1}^{N}%
\mathbb{E}_{\mathbb{P}}\left[ (Z_{n}^{j})^{+}\middle|\mathcal{G}\right]
+b_{\varepsilon }\right) \,.
\end{equation*}%
This chain of inequalities yields, since $b_{\varepsilon }$ disappears on
both sides: 
\begin{equation*}
-\mathbb{E}_{\mathbb{P}}\left[ U\left( Z\right) \middle|\mathcal{G}\right]
+\varepsilon \sum_{j=1}^{N}\mathbb{E}_{\mathbb{P}}\left[ (Z^{j})^{+}\middle|%
\mathcal{G}\right] \leq \liminf_{n}\left( -\mathbb{E}_{\mathbb{P}}\left[
U\left( Z_{n}\right) \middle|\mathcal{G}\right] +\varepsilon \sum_{j=1}^{N}%
\mathbb{E}_{\mathbb{P}}\left[ (Z_{n}^{j})^{+}\middle|\mathcal{G}\right]
\right) \,.
\end{equation*}%
We can thus exploit Item \ref{DynRMbudgetc} in RHS to get 
\begin{equation*}
-\mathbb{E}_{\mathbb{P}}\left[ U\left( Z\right) \middle|\mathcal{G}\right]
+\varepsilon \sum_{j=1}^{N}\mathbb{E}_{\mathbb{P}}\left[ (Z^{j})^{+}\middle|%
\mathcal{G}\right] \leq \left( -B+\varepsilon \sup_{n}\left( \sum_{j=1}^{N}%
\mathbb{E}_{\mathbb{P}}\left[ (Z_{n}^{j})^{+}\middle|\mathcal{G}\right]
\right) \right) \,.
\end{equation*}%
From the latter inequality we deduce 
\begin{equation*}
-\mathbb{E}_{\mathbb{P}}\left[ U\left( Z\right) \middle|\mathcal{G}\right]
\leq -B+\varepsilon \sup_{n}\left( \sum_{j=1}^{N}\mathbb{E}_{\mathbb{P}}%
\left[ (Z_{n}^{j})^{+}\middle|\mathcal{G}\right] \right)
\end{equation*}%
which holds ${\mathbb{P}}-$a.s. for all $\varepsilon >0$. By STEP 1, $%
\sup_{n}\left( \sum_{j=1}^{N}\mathbb{E}_{\mathbb{P}}\left[ (Z_{n}^{j})^{+}%
\middle|\mathcal{G}\right] \right) <\infty $ and therefore $-\mathbb{E}_{%
\mathbb{P}}\left[ U\left( Z^{j}\right) \middle|\mathcal{G}\right] \leq -B$.
\end{proof}

%
%
%
%
%
The following result is of technical nature and is used in the proof of the
existence of an optimal allocation $\widehat{Y}$ in Claim \ref{DynRMthmalloc}%
, STEP 2.

\begin{proposition}
\label{DynRMsuminl0gimpliesitslinftyg} Suppose the vectors $X\in (L^\infty(%
\mathcal{F}))^N$ and $Y\in (L^1(\mathcal{F}))^N$ satisfy $\sum_{j=1}^NY^j\in
L^\infty(\mathcal{G})$ and 
\begin{equation*}
\mathbb{E}_\mathbb{P} \left[U(X+Y)\mid\mathcal{G}\right]\geq B\,.
\end{equation*}
Suppose $\sum_{i\in I}Y^i \in L^0(\mathcal{G})$ for some family of indexes $%
I\subseteq\{1,\dots,N\}$. Then $\sum_{i\in I}Y^i \in L^\infty(\mathcal{G})$.
\end{proposition}

\begin{proof}
Set $A_{H}:=\{\sum_{i\in I}Y^{i}<-H\}\in \mathcal{G}$ for $H>0$ and suppose $%
{\mathbb{P}}(A_{H})>0$ for all $H>0$. Then we have by Lemma \ref%
{DynRMlemmacontrolwithline} (i) and $\mathbb{E}_{\mathbb{P}}\left[ U\left(
X+Y\right) \middle|\mathcal{G}\right] \geq B$ that 
\begin{equation}
B1_{A_{H}}\leq a\left( \sum_{j=1}^{N}X^{j}+\sum_{j=1}^{N}Y^{j}\right)
1_{A_{H}}-a\left( \sum_{j=1}^{N}(X^{j}+Y^{j})^{-}\right)
1_{A_{H}}+b1_{A_{H}}.  \label{DynRMcontrolline1}
\end{equation}%
At the same time from $\sum_{j=1}^{N}Y^{j}\in L^{\infty }(\mathcal{G})$ we
must have for some index $k\in I$ (depending on $H$) that $%
A_{H}^{k}:=\{Y^{k}<-\frac{1}{N+1}H\}\cap A_{H}\subseteq A_{H}$ satisfies ${%
\mathbb{P}}(A_{H}^{k})>0$ (otherwise we would get that $\sum_{i\in
I}Y^{i}\geq -\frac{N}{N+1}H$ on $A_{H}$, which is a contradiction). From %
\eqref{DynRMcontrolline1} and $H$ big enough we also have 
\begin{equation*}
\begin{split}
B1_{A_{H}^{k}}& \leq a\left( \sum_{j=1}^{N}X^{j}+\sum_{j=1}^{N}Y^{j}\right)
1_{A_{H}^{k}}+a(-(X^{k}+Y^{k})^{-})1_{A_{H}^{k}}+b1_{A_{H}^{k}} \\
& \leq a\left( \sum_{j=1}^{N}X^{j}+\sum_{j=1}^{N}Y^{j}\right)
1_{A_{H}^{k}}+a(-(\left\Vert X^{k}\right\Vert _{\infty
}+Y^{k})^{-})1_{A_{H}^{k}}+b1_{A_{H}^{k}} \\
& \leq a\left( \left\Vert \sum_{j=1}^{N}X^{j}+\sum_{j=1}^{N}Y^{j}\right\Vert
_{\infty }\right) 1_{A_{H}^{k}}+a\left( \left\Vert X^{k}\right\Vert _{\infty
}-\frac{H}{N+1}\right) 1_{A_{H}^{k}}+b1_{A_{H}^{k}}\,.
\end{split}%
\end{equation*}%
As $B$ is bounded, for an even bigger $H$ in this inequality, we get a
contradiction.

Now set $B_{H}:=\{\sum_{i\in I}Y^{i}>H\}$. Assume that ${\mathbb{P}}%
(B_{H})>0 $ for all $H>0$. Then from $\sum_{j=1}^{N}Y^{j}\in L^{\infty }(%
\mathcal{G})$ we get that 
\begin{equation*}
{\mathbb{P}}\left( \left\{ \sum_{i\notin
I}Y^{i}<\sum_{j=1}^{N}Y^{j}-H\right\} \right) >0\,.
\end{equation*}%
The argument in the first part of the proof can then be replicated, since $%
\sum_{i\notin I}Y^{i}\in L^{0}(\mathcal{G})$, yielding a contradiction.
\end{proof}

\subsection{Proofs for Section \protect\ref{Sec2}}

\label{DynRMsecproofstatic}%

\begin{proof}[Proof of Theorem \protect\ref{DynRMthmdual}]
In Item i) we consider $L_{\mathcal{F}}=L^{\infty }(\Omega ,\mathcal{F},{%
\mathbb{P}})$, $L^{\ast } = (L^{1}(\Omega ,\mathcal{F},{\mathbb{P}}))^{N}.$
We denote by $\text{ba}_{1}$ the set of $N$-dimensional vectors of finitely
additive functionals on $\mathcal{F}$ taking values in $[0,1]$ and taking
value $1$ on $\Omega $. Applying the Namioka-Klee Theorem in \cite{bfnam}
together with a standard argument regarding the monetary property, we see
that%
\begin{equation}
\rho _{0}(X)=\max_{\mu \in \text{ba}_{1}}\left( \sum_{j=1}^{N}\mu
^{j}(-X^{j})-\rho _{0}^{\ast }(-\mu )\right) .
\label{DynRMdualrepstaticnamioka}
\end{equation}%
We now follow the lines of \cite{FollmerSchied2}, Theorem 4.22 and Lemma
4.23.

Take any optimum $\widehat{\mu }=[\widehat{\mu }^{1},\dots ,\widehat{\mu }%
^{N}]$ in the dual representation \eqref{DynRMdualrepstaticnamioka}, so that 
$\rho _{0}^{\ast }(-\widehat{\mu })<+\infty ,$ and select a real number $%
c>-\rho _{0}(0)$ big enough so that $c\geq \rho _{0}^{\ast }(-\widehat{\mu }%
) $. Take any sequence of sets $(A_{n})_{n}$ in $\mathcal{F}$ increasing to $%
\Omega $. We claim that $\widehat{\mu }^{k}(A_{n})\rightarrow _{n}1$ for all 
$k\in \{1,\dots ,N\}$, which allows us to conclude that each $\widehat{\mu }%
^{k}$ is $\sigma $-additive. Hence any optimum of %
\eqref{DynRMdualrepstaticnamioka} belongs to $\mathscr{Q}$, which as a
consequence can replace $\text{ba}_{1}$ in \eqref{DynRMdualrepstaticnamioka}%
, obtaining thus the thesis. To prove the claim fix any $k\in \{1,\dots ,N\}$%
. By definition of $\rho _{0}^{\ast }$ 
\begin{equation*}
c\geq \rho _{0}^{\ast }(-\widehat{\mu })\geq \widehat{\mu }(-\lambda
1_{A_{n}}e^{k})-\rho _{0}(\lambda 1_{A_{n}}e^{k}),
\end{equation*}%
which implies 
\begin{equation*}
\widehat{\mu }(1_{A_{n}}e^{k})=\widehat{\mu }^{k}(A_{n})\geq \frac{1}{%
\lambda }\left( -c-\rho _{0}(\lambda 1_{A_{n}}e^{k})\right) .
\end{equation*}%
Now using continuity from below of $\rho _{0}$ we deduce that for each $%
\lambda >0$:%
\begin{equation*}
\liminf_{n}\left( \widehat{\mu }^{k}(A_{n})\right) \geq \lim_{n}\frac{1}{%
\lambda }\left( -c-\rho _{0}(\lambda 1_{A_{n}}e^{k}\right) =\frac{1}{\lambda 
}\left( -c-\rho _{0}(\lambda e^{k})\right) \overset{\text{Monetary prop.}}{=}%
1-\frac{c+\rho _{0}(0)}{\lambda }\,.
\end{equation*}%
Letting $\lambda \rightarrow +\infty $ we see that 
\begin{equation*}
\widehat{\mu }^{k}(A_{n})\rightarrow _{n}1\,.
\end{equation*}%
The $\sigma ((L^{\infty }(\mathcal{F}))^{N},(L^{1}(\mathcal{F}))^{N})$%
\textbf{-}lower semicontinuity of $\rho _{0}$ follows directly from (\ref%
{DynRMdualrep}) and continuity from above is a consequence of the $\sigma
((L^{\infty }(\mathcal{F}))^{N},(L^{1}(\mathcal{F}))^{N})$\textbf{-}lower
semicontinuity (see \cite{bfnam}).

The proof of Item ii) is a simple consequence of the Namioka-Klee Theorem in 
\cite{bfnam} applied to the dual system $(L_{\mathcal{F}},L^{\ast }).$ In
particular, the thesis follows from Lemma 7 \cite{bfnam} and the application
of the monetary property.
\end{proof}

%


\subsection{Proofs for Section \protect\ref{Sec3}}

\label{DynRMsecauxildualreprcond} %

\begin{proof}[Proof of Theorem \protect\ref{DynRMcorcashadd}]
Set $\rho _{0}(X):=\mathbb{E}_{\mathbb{P}}\left[ \rho _{\mathcal{G}}\left(
X\right) \right] $ and let%
\begin{equation*}
\gamma (\mathbb{Q}):=\esssup_{X\in L_{\mathcal{F}}}\left\{ \sum_{j=1}^{N}%
\mathbb{E}_{\mathbb{Q}^{j}}\left[ -X^{j}\middle|\mathcal{G}\right] -\rho _{%
\mathcal{G}}\left( X\right) \right\} :=\rho _{\mathcal{G}}^{\ast }\left( -%
\frac{\mathrm{d}{\mathbb{Q}}}{\mathrm{d}{\mathbb{P}}}\right) \text{,\quad }{%
\mathbb{Q}}\in \mathscr{Q}_{\mathcal{G}},
\end{equation*}%
where $\rho _{\mathcal{G}}^{\ast }$ was introduced in (\ref%
{DynRmdefdualgeneral}). From the definition of $\gamma (\mathbb{Q})$ we
immediately deduce 
\begin{equation}
\rho _{\mathcal{G}}\left( X\right) \geq \esssup_{{\mathbb{Q}}\in \mathscr{Q}%
_{\mathcal{G}}}\left\{ \sum_{j=1}^{N}\mathbb{E}_{\mathbb{Q}^{j}}\left[ -X^{j}%
\middle|\mathcal{G}\right] -\gamma (\mathbb{Q})\right\} \geq \left\{
\sum_{j=1}^{N}\mathbb{E}_{\mathbb{\widehat{Q}}^{j}}\left[ -X^{j}\middle|%
\mathcal{G}\right] -\gamma (\mathbb{\widehat{Q}})\right\}  \label{newgamma}
\end{equation}%
for any $\mathbb{\widehat{Q}}\in \mathscr{Q}_{\mathcal{G}}$.\newline
\textbf{STEP 1}: $\rho _{\mathcal{G}}$ has the local property, i.e. for any $%
A\in \mathcal{G}$ and $X\in L_{\mathcal{F}}$ $\rho _{\mathcal{G}}\left(
X\right) 1_{A}=\rho _{\mathcal{G}}\left( 1_{A}X\right) 1_{A}$.

Observe that 
\begin{align*}
\rho _{\mathcal{G}}\left( 1_{A}X\right) &\overset{\eqref{DynRMccnvx}}{\leq }%
\rho _{\mathcal{G}}\left( X\right) 1_{A}+\rho _{\mathcal{G}}\left( 0\right)
1_{A^{c}}\overset{\eqref{DynRMccnvx}}{\leq }1_{A}\left( \rho _{\mathcal{G}%
}\left( 1_{A}X\right) 1_{A}+\rho _{\mathcal{G}}\left( X1_{A^{c}}\right)
1_{A^{c}}\right) +\rho _{\mathcal{G}}\left( 0\right) 1_{A^{c}} \\
&=\rho _{\mathcal{G}}\left( 1_{A}X\right) 1_{A}+\rho _{\mathcal{G}}\left(
0\right) 1_{A^{c}}\,,
\end{align*}%
then multiply by $1_{A}$. The local property of $\rho _{\mathcal{G}}$ is
equivalent to: 
\begin{equation}
\rho _{\mathcal{G}}\left( X1_{A}+Z1_{A^{c}}\right) =\rho _{\mathcal{G}%
}\left( X\right) 1_{A}+\rho _{\mathcal{G}}\left( Z\right) 1_{A^{c}}
\label{PLOC}
\end{equation}%
if $A\in \mathcal{G}$ and $X,Z\in L_{\mathcal{F}}$.\newline
\medskip \textbf{STEP 2}: for every ${\mathbb{Q}}\in \mathscr{Q}_{\mathcal{G}%
}$ the set $\{\sum_{j=1}^{N}\mathbb{E}_{\mathbb{Q}^{j}}\left[ -X^{j}\middle|%
\mathcal{G}\right] -\rho _{\mathcal{G}}\left( X\right) ,X\in L_{\mathcal{F}%
}\}$ is upward directed.

For $X,Z\in L_{\mathcal{F}}$ we set 
\begin{equation*}
\xi _{X}:=\sum_{j=1}^{N}\mathbb{E}_{\mathbb{Q}^{j}}\left[ -X^{j}\middle|%
\mathcal{G}\right] -\rho _{\mathcal{G}}\left( X\right) ,\,\,\,\,\,\,\,\,\xi
_{Z}:=\sum_{j=1}^{N}\mathbb{E}_{\mathbb{Q}^{j}}\left[ -Z^{j}\middle|\mathcal{%
G}\right] -\rho _{\mathcal{G}}\left( Z\right) \,,
\end{equation*}%
\begin{equation*}
A:=\left\{ \xi _{X}\geq \xi _{Z}\right\} \in \mathcal{G},\,\,\,\,\,\,\,%
\,W:=X1_{A}+Z1_{A^{c}}\,.
\end{equation*}%
As explained in Remark \ref{DynRMremdecom}, by the decomposability of $L_{%
\mathcal{F}}$ we get $W\in L_{\mathcal{F}}$. 
Then one can check, using \eqref{PLOC}, that 
\begin{equation*}
\sum_{j=1}^{N}\mathbb{E}_{\mathbb{Q}^{j}}\left[ -W^{j}\middle|\mathcal{G}%
\right] -\rho _{\mathcal{G}}\left( W\right) =\xi _{X}1_{A}+\xi
_{Z}1_{A^{c}}=\max (\xi _{X},\xi _{Y})
\end{equation*}%
proving the claim.

\textbf{STEP 3: }For ${\mathbb{Q}}\in \mathscr{Q}_{\mathcal{G}}$, $\gamma
_{0}(\mathbb{Q})=\mathbb{E}_{\mathbb{P}}\left[ \gamma (\mathbb{Q})\right] $
where we set 
\begin{equation*}
\gamma _{0}(\mathbb{Q}):=\sup_{X\in L_{\mathcal{F}}}\left( \sum_{j=1}^{N}%
\mathbb{E}_{\mathbb{Q}^{j}}\left[ -X^{j}\right] -\rho _{0}(X)\right) =\rho
_{0}^{\ast }\left( -\frac{\mathrm{d}{\mathbb{Q}}}{\mathrm{d}{\mathbb{P}}}%
\right) ,\quad {\mathbb{Q}}\in \mathscr{Q}.
\end{equation*}%
Recall that $\mathbb{Q}^{j}=\mathbb{P}$ on $\mathcal{G}$ for all ${\mathbb{Q}%
}\in \mathscr{Q}_{\mathcal{G}}$. By Step 2 we deduce:%
\begin{align*}
\mathbb{E}_{\mathbb{P}}\left[ \gamma (\mathbb{Q})\right] & =\mathbb{E}_{%
\mathbb{P}}\left[ \esssup_{X\in L_{\mathcal{F}}}\left\{ \sum_{j=1}^{N}%
\mathbb{E}_{\mathbb{Q}^{j}}\left[ -X^{j}\middle|\mathcal{G}\right] -\rho _{%
\mathcal{G}}\left( X\right) \right\} \right] \\
& =\sup_{X\in L_{\mathcal{F}}}\mathbb{E}_{\mathbb{P}}\left[ \sum_{j=1}^{N}%
\mathbb{E}_{\mathbb{Q}^{j}}\left[ -X^{j}\middle|\mathcal{G}\right] -\rho _{%
\mathcal{G}}\left( X\right) \right] =\sup_{X\in L_{\mathcal{F}}}\left\{
\sum_{j=1}^{N}\mathbb{E}_{\mathbb{Q}^{j}}\left[ \mathbb{E}_{\mathbb{Q}^{j}}%
\left[ -X^{j}\middle|\mathcal{G}\right] \right] -\mathbb{E}_{\mathbb{P}}%
\left[ \rho _{\mathcal{G}}\left( X\right) \right] \right\} \\
& =\sup_{X\in L_{\mathcal{F}}}\left\{ \sum_{j=1}^{N}\mathbb{E}_{\mathbb{Q}%
^{j}}\left[ -X^{j}\right] -\rho _{0}\left( X\right) \right\} =\gamma _{0}(%
\mathbb{Q})\,.
\end{align*}%
Notice that we have also shown: $\rho _{0}^{\ast }\left( -\frac{\mathrm{d}{%
\mathbb{Q}}}{\mathrm{d}{\mathbb{P}}}\right) =\mathbb{E}_{\mathbb{P}}\left[
\rho _{\mathcal{G}}^{\ast }\left( -\frac{\mathrm{d}{\mathbb{Q}}}{\mathrm{d}{%
\mathbb{P}}}\right) \right] .$

\textbf{STEP 4} Dual Representation and attainment of the supremum.

As in the univariate case, applying the monetary property one may show that $%
\gamma _{0}(\mathbb{Q})=+\infty $ if ${\mathbb{Q}}\in \mathscr{Q}\backslash %
\mathscr{Q}_{\mathcal{G}}.$ Then by nice representability of $\rho _{0}$ 
\begin{align*}
\mathbb{E}_{\mathbb{P}}\left[ \rho _{\mathcal{G}}(X)\right] &=\rho
_{0}(X)=\max_{{\mathbb{Q}}\in \mathscr{Q}}\left( \sum_{j=1}^{N}\mathbb{E}_{{%
\mathbb{Q}}_{j}}\left[ -X^{j}\right] -\gamma _{0}(\mathbb{Q})\right) =\max_{{%
\mathbb{Q}}\in \mathscr{Q}_{\mathcal{G}}}\left( \sum_{j=1}^{N}\mathbb{E}_{{%
\mathbb{Q}}_{j}}\left[ -X^{j}\right] -\gamma _{0}(\mathbb{Q})\right) \\
&=\sup_{{\mathbb{Q}}\in \mathscr{Q}_{\mathcal{G}}}\left( \sum_{j=1}^{N}%
\mathbb{E}_{\mathbb{Q}^{j}}\left[ \mathbb{E}_{\mathbb{Q}^{j}}\left[ -X^{j}%
\middle|\mathcal{G}\right] \right] -\mathbb{E}_{\mathbb{P}}\left[ \gamma (%
\mathbb{Q})\right] \right) \\
&=\sup_{{\mathbb{Q}}\in \mathscr{Q}_{\mathcal{G}}}\mathbb{E}_{\mathbb{P}}%
\left[ \sum_{j=1}^{N}\mathbb{E}_{\mathbb{Q}^{j}}\left[ -X^{j}\middle|%
\mathcal{G}\right] -\gamma (\mathbb{Q})\right] \leq \mathbb{E}_{\mathbb{P}}%
\left[ \esssup_{{\mathbb{Q}}\in \mathscr{Q}_{\mathcal{G}}}\left(
\sum_{j=1}^{N}\mathbb{E}_{\mathbb{Q}^{j}}\left[ -X^{j}\middle|\mathcal{G}%
\right] -\gamma (\mathbb{Q})\right) \right] \,.
\end{align*}%
From this inequality and (\ref{newgamma}) we then deduce the dual
representation:%
\begin{equation}  \label{DynRmneweqrhog}
\rho _{\mathcal{G}}(X)=\esssup_{{\mathbb{Q}}\in \mathscr{Q}_{\mathcal{G}%
}}\left( \sum_{j=1}^{N}\mathbb{E}_{\mathbb{Q}^{j}}\left[ -X^{j}\middle|%
\mathcal{G}\right] -\gamma (\mathbb{Q})\right).
\end{equation}%
Similarly, for the vector $\mathbb{\widehat{Q}}\in \mathscr{Q}_{\mathcal{G}}$
attaining the maximum in $\rho _{0}(X)$ we have:

\begin{equation*}
\mathbb{E}_{\mathbb{P}}\left[ \rho _{\mathcal{G}}(X)\right] =\mathbb{E}_{%
\mathbb{P}}\left[ \sum_{j=1}^{N}\mathbb{E}_{\mathbb{\widehat{Q}}^{j}}\left[
-X^{j}\middle|\mathcal{G}\right] -\gamma (\mathbb{\widehat{Q}})\right]
\end{equation*}%
and then again from (\ref{newgamma}) we deduce 
\begin{equation*}
\rho _{\mathcal{G}}(X)=\sum_{j=1}^{N}\mathbb{E}_{\mathbb{\widehat{Q}}^{j}}%
\left[ -X^{j}\middle|\mathcal{G}\right] -\gamma (\mathbb{\widehat{Q}}).
\end{equation*}%
\textbf{STEP 5 } We prove that $\alpha ({\mathbb{Q}})=\gamma (\mathbb{Q})$
for all ${\mathbb{Q}}\in \mathscr{Q}_{\mathcal{G}}$, where $\alpha $ is
defined in (\ref{DynRMequatintro2}). Then \eqref{DynRMdualreprgeneralcor}
follows from \eqref{DynRmneweqrhog}. Applying the monetary property of $\rho
_{\mathcal{G}},$ for each ${\mathbb{Q}}\in \mathscr{Q}_{\mathcal{G}}$ 
\begin{align*}
\gamma (\mathbb{Q})& =\esssup_{X\in L_{\mathcal{F}}}\left\{ \sum_{j=1}^{N}%
\mathbb{E}_{\mathbb{Q}^{j}}\left[ -X^{j}\middle|\mathcal{G}\right] -\rho _{%
\mathcal{G}}\left( X\right) \right\} =\esssup_{X\in L_{\mathcal{F}}}\left(
\sum_{j=1}^{N}\mathbb{E}_{{\mathbb{Q}}^{j}}\left[ \left( -X^{j}-\frac{1}{N}%
\rho _{\mathcal{G}}\left( X\right) \right) \middle|\mathcal{G}\right] \right)
\\
& \leq \esssup_{Z\in L_{\mathcal{F}},\rho _{\mathcal{G}}\left( Z\right) \leq
0}\left( \sum_{j=1}^{N}\mathbb{E}_{{\mathbb{Q}}^{j}}\left[ -Z^{j}\middle|%
\mathcal{G}\right] \right) =\alpha ({\mathbb{Q}}) \\
& \leq \esssup_{Z\in L_{\mathcal{F}},\rho _{\mathcal{G}}\left( Z\right) \leq
0}\sum_{j=1}^{N}\left( \mathbb{E}_{{\mathbb{Q}}^{j}}\left[ -Z^{j}\middle|%
\mathcal{G}\right] -\frac{1}{N}\rho _{\mathcal{G}}\left( Z\right) \right)
\leq \esssup_{Z\in L_{\mathcal{F}}}\left( \sum_{j=1}^{N}\mathbb{E}_{{\mathbb{%
Q}}^{j}}\left[ -Z^{j}\middle|\mathcal{G}\right] -\rho _{\mathcal{G}}\left(
Z\right) \right) =\gamma (\mathbb{Q})\,.
\end{align*}

\textbf{STEP 6} Continuity from above of $\rho _{\mathcal{G}}$

Continuity from above is an easy consequence of the dual representation (\ref%
{DynRMdualreprgeneralcor}) just proved. Take $X_{n}\downarrow
X\Leftrightarrow -X_{n}\uparrow -X$ and observe that:%
\begin{align*}
\rho _{\mathcal{G}}\left( X\right) & =\esssup_{{\mathbb{Q}}\in \mathscr{Q}_{%
\mathcal{G}}}\left( \sum_{j=1}^{N}\mathbb{E}_{\mathbb{Q}^{j}}\left[ -X^{j}%
\middle|\mathcal{G}\right] -\alpha ({\mathbb{Q}})\right) \\
& \overset{(\text{cMON})}{=}\esssup_{{\mathbb{Q}}\in \mathscr{Q}_{\mathcal{G}%
}}\left( \esssup_{n}\sum_{j=1}^{N}\mathbb{E}_{\mathbb{Q}^{j}}\left[
-X_{n}^{j}\middle|\mathcal{G}\right] -\alpha ({\mathbb{Q}})\right)
\end{align*}
\begin{align*}
& \overset{\text{Prop.}\ref{DynRMrefesssup}}{=}\esssup_{n}\esssup_{{\mathbb{Q%
}}\in \mathscr{Q}_{\mathcal{G}}}\left( \sum_{j=1}^{N}\mathbb{E}_{\mathbb{Q}%
^{j}}\left[ -X_{n}^{j}\middle|\mathcal{G}\right] -\alpha ({\mathbb{Q}}%
)\right) \\
& =\sup_{n}\rho _{\mathcal{G}}\left( X_{n}\right) =\lim_{n}\rho _{\mathcal{G}%
}\left( X_{n}\right) \leq \rho _{\mathcal{G}}\left( X\right) .
\end{align*}

%
%
%
\end{proof}

\subsection{Proofs for Section \protect\ref{Sec5}}

The proof of Proposition \ref{DynRMpropoptisint} needs some preparation.
Recall the definition of $\rho _{\mathcal{G}}^{\ast }$ in (\ref%
{DynRmdefdualgeneral}).

\begin{proposition}
\label{DynRMpropextension} There exists an extension, $\rho_0^\Phi$, of $%
\rho_0(\cdot):=\mathbb{E}_\mathbb{P} \left[\rho_{\mathcal{G}%
}\left(\cdot\right)\right]$ to $M^\Phi$ which is convex, nondecreasing and $%
\left\| \cdot\right\|_\Phi$-continuous.
\end{proposition}

\begin{proof}
Observe that because of the downward directness proved in STEP 1 of Claim %
\ref{DynRMthmalloc}, together with (MON), we have 
\begin{equation*}
\rho _{0}(X)=\inf \left\{ \sum_{j=1}^{N}\mathbb{E}_{\mathbb{P}}\left[ Y^{j}%
\right] \mid Y\in \mathscr{C}_{\mathcal{G}},\,\mathbb{E}_{\mathbb{P}}\left[
U(X+Y)\middle|\mathcal{G}\right] \geq B\right\} \,\,\,\,\,X\in (L^{\infty }(%
\mathcal{F}))^{N}\,.
\end{equation*}%
Define now

\begin{equation}
\rho _{0}^{\Phi }(X):=\inf \left\{ \sum_{j=1}^{N}\mathbb{E}_{\mathbb{P}}%
\left[ Y^{j}\right] \mid Y\in \mathscr{C}_{\mathcal{G}},\,\mathbb{E}_{%
\mathbb{P}}\left[ U(X+Y)\middle|\mathcal{G}\right] \geq B\right\}
\,\,\,\,\,X\in M^{\Phi }\,.  \label{DynRMdefr0phi}
\end{equation}%
We easily see that $\rho _{0}^{\Phi }(X)<+\infty $ for every $X\in M^{\Phi }$
(since the set over which we take infima in \eqref{DynRMdefr0phi} is
nonempty). Moreover $\rho _{0}^{\Phi }(X)>-\infty $ since if this were the
case, for a minimizing sequence $(Y_{n})_{n}$ we would have $%
\inf_{n}\sum_{j=1}^{N}\mathbb{E}_{\mathbb{P}}\left[ Y_{n}^{j}\right]
=-\infty $. Now, by Lemma \ref{DynRMlemmacontrolwithline}.(ii)%
\begin{align*}
\mathbb{E}_{\mathbb{P}}\left[ B\right] & \leq \mathbb{E}_{\mathbb{P}}\left[
U(X+Y_{n})\right] \\
& \leq a\sum_{j=1}^{N}\mathbb{E}_{\mathbb{P}}\left[ X^{j}\right]
+a\sum_{j=1}^{N}\mathbb{E}_{\mathbb{P}}\left[ Y_{n}^{j}\right] +b\leq
const+a\sum_{j=1}^{N}\mathbb{E}_{\mathbb{P}}\left[ Y_{n}^{j}\right] ,
\end{align*}%
which gives a contradiction. Clearly, mimicking what we did in the proof of
Claim \ref{DynRMthmalloc} Step 3, we can check that $\rho _{0}^{\Phi }(\cdot
)$ is also convex and nondecreasing. Now by the Extended Namioka-Klee
Theorem in \cite{bfnam}) it is also norm continuous.
\end{proof}

\begin{lemma}
\label{DynRMpropint} For any $Z\in (L^1(\mathcal{F}))^N$ we have that $\rho_%
\mathcal{G}^*(Z)\in L^1(\Omega,\mathcal{G},{\mathbb{P}})$ if and only if $%
\rho_0^*(Z)<+\infty$ and, if any of the two conditions is met, we have $Z\in
K_\Phi$.
\end{lemma}

\begin{proof}
Notice first that, for any $Z\in (L^{1}(\mathcal{F}))^{N},$ $\rho _{\mathcal{%
G}}^{\ast }(Z)\geq -\rho _{\mathcal{G}}\left( 0\right) \in L^{\infty
}(\Omega ,\mathcal{G},{\mathbb{P}})$ and $\rho _{0}^{\ast }(Z)\geq -\rho
_{0}\left( 0\right) >-\infty $. As in STEP 3 of Theorem \ref{DynRMcorcashadd}%
, one can show that $\rho _{0}^{\ast }(Z)=\mathbb{E}_{\mathbb{P}}\left[ \rho
_{\mathcal{G}}^{\ast }(Z)\right] $ for any $Z\in (L^{1}(\mathcal{F}))^{N}$.
The first claim then follows. Suppose now $\rho _{\mathcal{G}}^{\ast }(Z)\in
L^{1}(\mathcal{G})$. By \cite{Aliprantis} Theorem 5.43 Item 3 $\rho
_{0}^{\Phi }$ is bounded on a ball $B_{\varepsilon }$ (defined using the
norm $\left\Vert \cdot \right\Vert _{\Phi }$) of $M^{\Phi }$ centered at $0$%
. We have as a consequence 
\begin{equation*}
+\infty >\sup_{X\in B_{\varepsilon }}\left( \rho _{0}^{\Phi }(X)+\rho
_{0}^{\ast }(Z)\right) \,.
\end{equation*}%
Now we use the fact that $\rho _{0}^{\Phi }$, when restricted to $(L^{\infty
}(\mathcal{F}))^{N}$, coincides with $\rho _{0}$, and continuity of $\rho
_{0}^{\Phi }$ (by Proposition \ref{DynRMpropextension}) to see that

\begin{align*}
+\infty&>\sup_{X\in
B_\varepsilon}\left(\rho_0^\Phi(X)+\rho_0^*(Z)\right)\geq\sup_{X\in
B_\varepsilon\cap (L^\infty(\mathcal{F}))^N}\left(\rho_0^\Phi(X)+\rho_0^*(Z)%
\right) \\
&=\sup_{X\in B_\varepsilon\cap (L^\infty(\mathcal{F}))^N}\left(\rho_0(X)+%
\rho_0^*(Z)\right)\geq \sup_{X\in B_\varepsilon\cap (L^\infty(\mathcal{F}%
))^N}\left(\sum_{j=1}^N\mathbb{E}_\mathbb{P} \left[X^jZ^j\right]\right)
\end{align*}
where we used Fenchel inequality to obtain the last inequality. Furthermore,
using the fact that given $Z$, for any $X\in B_\varepsilon\cap (L^\infty(%
\mathcal{F}))^N$ the vector $\widehat{X}$ defined by $\widehat{X}%
^j=sgn(Z^j)\left|X^j\right|$ still belongs to $B_\varepsilon\cap (L^\infty(%
\mathcal{F}))^N$, we have 
\begin{equation*}
\sup_{X\in B_\varepsilon\cap (L^\infty(\mathcal{F}))^N}\left(\sum_{j=1}^N%
\mathbb{E}_\mathbb{P} \left[X^jZ^j\right]\right)=\sup_{X\in
B_\varepsilon\cap (L^\infty(\mathcal{F}))^N}\left(\sum_{j=1}^N\mathbb{E}_%
\mathbb{P} \left[\left|X^jZ^j\right|\right]\right)\,.
\end{equation*}
To conclude, we observe that an approximation with simple functions yields: 
\begin{equation*}
\sup_{X\in B_\varepsilon\cap (L^\infty(\mathcal{F}))^N}\left(\sum_{j=1}^N%
\mathbb{E}_\mathbb{P} \left[\left|X^jZ^j\right|\right]\right)=\sup_{X\in
B_\varepsilon}\left(\sum_{j=1}^N\mathbb{E}_\mathbb{P} \left[%
\left|X^jZ^j\right|\right]\right) \,.
\end{equation*}
This completes the proof using \cite{DF19} Proposition 2.5 Item 1.
\end{proof}

\begin{lemma}
\label{DynRMlemmaintnegaprts} Suppose Assumption \ref{DynRMA2} holds. Let $%
Z\in (L^1(\mathcal{F}))^N$ be given and suppose that $\mathbb{E}_\mathbb{P} %
\left[U(Z)\middle|\mathcal{G}\right]\geq B$. Then for any $W\in K_\Phi$ we
have $\sum_{j=1}^N(Z^j)^-W^j\in L^1(\mathcal{F})$.
\end{lemma}

\begin{proof}
Observe that $\mathbb{E}_{\mathbb{P}}\left[ U(Z)\right] =\mathbb{E}_{\mathbb{%
P}}\left[ \mathbb{E}_{\mathbb{P}}\left[ U(Z)\middle|\mathcal{G}\right] %
\right] \geq \mathbb{E}_{\mathbb{P}}\left[ B\right] $. Furthermore 
\begin{equation*}
U(Z)=\sum_{j=1}^{N}u_{j}(Z^{j})+\Lambda
(Z)=\sum_{j=1}^{N}u_{j}((Z^{j})^{+})+\sum_{j=1}^{N}u_{j}(-(Z^{j})^{-})+%
\Lambda (Z)\,.
\end{equation*}%
As $u_{j}(0)=0$ and $u_{j}$ is increasing, for each $j$, this implies 
\begin{equation*}
0\leq -\sum_{j=1}^{N}u_{j}(-(Z^{j})^{-})\leq \max_{j=1,\dots ,N}\left( \frac{%
\mathrm{d}u_{j}}{\mathrm{d}x^{j}}(0)\right)
\sum_{j=1}^{N}(Z^{j})^{+}+\sup_{z\in {\mathbb{R}}^{N}}\Lambda (z)-U(Z)
\end{equation*}%
where in the last inequality we used \eqref{DynRMgradest}. It then follows
that $\sum_{j=1}^{N}\left( -u_{j}(-(Z^{j})^{-})\right) \in L^{1}(\mathcal{F}%
) $, which in turns yields $(Z)^{-}\in L^{\phi _{1}}\times \dots \times
L^{\phi _{N}}$, for $\phi _{j}$ defined in (\ref{DynRMassocorliczj}). Since
by \cite{DF19} Proposition 2.5 Item 3 we have $W\in L^{\phi _{1}^{\ast
}}\times \dots \times L^{\phi _{N}^{\ast }}$, we get by \cite{Edgar}
Proposition 2.2.7 that $(Z^{j})^{-}W^{j}\in L^{1}(\mathcal{F})$ for every $%
j=1,\dots ,N$, and the last claim is proved.
\end{proof}

\begin{proof}[Proof of Proposition \protect\ref{DynRMpropoptisint}]
Observe that ${\mathbb{Q}}\in \mathscr{Q}_{\mathcal{G}}^{1}$ implies $%
\alpha^1({\mathbb{Q}})\in L^1(\mathcal{G})$ by definition of $\mathscr{Q}_{%
\mathcal{G}}^{1}$, which in turns implies that $\rho^*_\mathcal{G}\left(-%
\frac{\mathrm{d}{\mathbb{Q}}}{\mathrm{d}{\mathbb{P}}}\right)\in L^1(\mathcal{%
G})$ by Claim \ref{DynRMthmdualreprshorfall}. By Lemma \ref{DynRMpropint},
then, for any ${\mathbb{Q}}\in \mathscr{Q}_{\mathcal{G}}^{1}$ we have $\frac{%
\mathrm{d}{\mathbb{Q}}}{\mathrm{d}{\mathbb{P}}}\in K_{\Phi }$, so that, by
Lemma \ref{DynRMlemmaintnegaprts}, for any ${\mathbb{Q}}\in \mathscr{Q}_{%
\mathcal{G}}^{1}$ we get $[(\widehat{Y}^{1})^{-},\dots ,(\widehat{Y}%
^{N})^{-}]\in L^{1}({\mathbb{Q}}^{1})\times \dots \times L^{1}({\mathbb{Q}}%
^{N})$. Given $\widehat{Y}$, we use the notation $\widehat{Y}_{(k)}$ and $Z_Y
$ from \eqref{DynRMdeftruncated} for large values $k\geq k_{\widehat{Y}}$.
By Fatou Lemma we have for any ${\mathbb{Q}}\in \mathscr{Q}_{\mathcal{G}}^{1}
$ 
\begin{align*}
\sum_{j=1}^{N}\mathbb{E}_{\mathbb{Q}^{j}}\left[ (\widehat{Y}^{j})^{+}\middle|%
\mathcal{G}\right] & \leq \liminf_{k}\sum_{j=1}^{N}\mathbb{E}_{\mathbb{Q}%
^{j}}\left[ (\widehat{Y}_{(k)}^{j})^{+}\middle|\mathcal{G}\right] \\
& =\liminf_{k}\left( \sum_{j=1}^{N}\mathbb{E}_{\mathbb{Q}^{j}}\left[ 
\widehat{Y}_{(k)}^{j}\middle|\mathcal{G}\right] +\sum_{j=1}^{N}\mathbb{E}_{%
\mathbb{Q}^{j}}\left[ (\widehat{Y}_{(k)}^{j})^{-}\middle|\mathcal{G}\right]
\right) \\
& \overset{{\mathbb{Q}}\in \mathscr{Q}_{\mathcal{G}}}{\leq }%
\liminf_{k}\left( \sum_{j=1}^{N}\widehat{Y}_{(k)}^{j}+\sum_{j=1}^{N}\mathbb{E%
}_{\mathbb{Q}^{j}}\left[ (\widehat{Y}_{(k)}^{j})^{-}\middle|\mathcal{G}%
\right] \right) \,.
\end{align*}%
We conclude that 
\begin{equation*}
\sum_{j=1}^{N}\mathbb{E}_{\mathbb{Q}^{j}}\left[ (\widehat{Y}^{j})^{+}\middle|%
\mathcal{G}\right] \leq \sum_{j=1}^{N}\widehat{Y}^{j}+\lim_{k}\sum_{j=1}^{N}%
\mathbb{E}_{\mathbb{Q}^{j}}\left[ (\widehat{Y}_{(k)}^{j})^{-}\middle|%
\mathcal{G}\right] \overset{\text{(cDOM)}}{=}\sum_{j=1}^{N}\widehat{Y}%
^{j}+\sum_{j=1}^{N}\mathbb{E}_{\mathbb{Q}^{j}}\left[ (\widehat{Y}^{j})^{-}%
\middle|\mathcal{G}\right]
\end{equation*}%
where we used in the last step that $Y_{(k)}\rightarrow \widehat{Y}\,{%
\mathbb{P}}-$ a.s. and that $(\widehat{Y})^{-}\in L^{1}({\mathbb{Q}}%
^{1})\times \dots \times L^{1}({\mathbb{Q}}^{N})$ to apply (cDOM): $(%
\widehat{Y}_{(k)}^{j})^{-}\leq \max \left( (\widehat{Y}%
^{j})^{-},(Z_{Y}^{j})^{-}\right) \in L^{1}({\mathbb{Q}}^{j}),\,j=1,\dots ,N$%
. This yields both integrability and the fact that, rearranging terms, 
\begin{equation*}
\sum_{j=1}^{N}\mathbb{E}_{\mathbb{Q}^{j}}\left[ \widehat{Y}^{j}\middle|%
\mathcal{G}\right] \leq \sum_{j=1}^{N}\widehat{Y}^{j}\,.
\end{equation*}
\end{proof}

\label{DynRmsecproofs11} To prove Theorem \ref{DynRMoptforrhoq} we first
state two preliminary Propositions.

\begin{proposition}
\label{Prop517}Let $X\in (L^{\infty }(\mathcal{F},{\mathbb{P}}))^{N}$. For
any ${\mathbb{Q}}\in \mathscr{Q}_{\mathcal{G}}^{1}$ we have:%
\begin{equation}
\rho _{\mathcal{G}}^{\mathbb{Q}}(X)=\sum_{j=1}^{N}\mathbb{E}_{\mathbb{Q}^{j}}%
\left[ -X^{j}\middle|\mathcal{G}\right] -\alpha ^{1}({\mathbb{Q}})
\label{DynRMeqrhoqalpha}
\end{equation}%
and 
\begin{equation}
\rho _{\mathcal{G}}\left( X\right) =\max_{{\mathbb{Q}}\in \mathscr{Q}_{%
\mathcal{G}}^{1}}\rho _{\mathcal{G}}^{\mathbb{Q}}(X)=\rho _{\mathcal{G}}^{%
\widehat{{\mathbb{Q}}}}(X)\text{ for any optimum }\widehat{{\mathbb{Q}}}%
\text{ of \eqref{DynRMeqdualreprshorfall}.}  \label{DynRMlinkpipiq}
\end{equation}
\end{proposition}

\begin{proof}
\bigskip We first prove \eqref{DynRMeqrhoqalpha}: observe that by %
\eqref{DynRMdefalpha1} we have for any ${\mathbb{Q}}\in \mathscr{Q}_{%
\mathcal{G}}^{1}$%
\begin{align*}
\alpha ^{1}({\mathbb{Q}})& =\esssup\left\{ \sum_{j=1}^{N}\mathbb{E}_{\mathbb{%
Q}^{j}}\left[ -W^{j}\middle|\mathcal{G}\right] \mid W\in (L^{\infty }(%
\mathcal{F},{\mathbb{P}}))^{N},\text{ }\mathbb{E}_{\mathbb{P}}\left[ U\left(
W\right) \middle|\mathcal{G}\right] \geq B\right\} \\
& =-\essinf_{\substack{ W-X\in (L^{\infty }(\mathcal{F},{\mathbb{P}}))^{N} 
\\ \mathbb{E}_{\mathbb{P}}\left[ U\left( X+(W-X)\right) \middle|\mathcal{G}%
\right] \geq B}}\left\{ \sum_{j=1}^{N}\mathbb{E}_{\mathbb{Q}^{j}}\left[
(W^{j}-X^{j})\middle|\mathcal{G}\right] +\sum_{j=1}^{N}\mathbb{E}_{\mathbb{Q}%
^{j}}\left[ X^{j}\middle|\mathcal{G}\right] \right\} \\
& =\sum_{j=1}^{N}\mathbb{E}_{\mathbb{Q}^{j}}\left[ -X^{j}\middle|\mathcal{G}%
\right] -\essinf\left\{ \sum_{j=1}^{N}\mathbb{E}_{\mathbb{Q}^{j}}\left[ Z^{j}%
\middle|\mathcal{G}\right] \mid Z\in (L^{\infty }(\mathcal{F},{\mathbb{P}}%
))^{N},\,\mathbb{E}_{\mathbb{P}}\left[ U\left( X+Z\right) \middle|\mathcal{G}%
\right] \geq B\right\} \\
& =\sum_{j=1}^{N}\mathbb{E}_{\mathbb{Q}^{j}}\left[ -X^{j}\middle|\mathcal{G}%
\right] -\rho _{\mathcal{G}}^{{\mathbb{Q}}}(X)\,.
\end{align*}%
Observe now that by Theorem \ref{DynRMmainthmrhoinfty} Item 3 $\rho _{%
\mathcal{G}}\left( X\right) \geq \sum_{j=1}^{N}\mathbb{E}_{\mathbb{Q}^{j}}%
\left[ -X^{j}\middle|\mathcal{G}\right] -\alpha ^{1}({\mathbb{Q}})$ for
every ${\mathbb{Q}}\in \mathscr{Q}_{\mathcal{G}}^{1}$, and equality holds
for any optimum $\widehat{{\mathbb{Q}}}$ of (\ref{DynRMeqdualreprshorfall}).
Direct substitution yields then \eqref{DynRMlinkpipiq}.
\end{proof}

\begin{proposition}
\label{Prop518}Suppose Assumption \ref{DynRMA2} is fulfilled and let $X\in
(L^{\infty }(\mathcal{F},{\mathbb{P}}))^{N}$. Then for any ${\mathbb{Q}}\in %
\mathscr{Q}_{\mathcal{G}}^{1}$ we have:%
\begin{equation}
\rho _{\mathcal{G}}^{{\mathbb{Q}}}(X)=\essinf\left\{ \sum_{j=1}^{N}\mathbb{E}%
_{\mathbb{Q}^{j}}\left[ Y^{j}\middle|\mathcal{G}\right] \mid Y\in L^{1}(%
\mathcal{F},{\mathbb{Q}})\cap (L^{1}(\mathcal{F},{\mathbb{P}}))^{N}\text{, }%
\mathbb{E}_{\mathbb{P}}\left[ U\left( X+Y\right) \middle|\mathcal{G}\right]
\geq B\right\} \,.  \label{DynRMrhoonl1pq}
\end{equation}
\end{proposition}

\begin{proof}
Clearly the inequality $(\geq )$ is trivial, since we are enlarging the set
over which we take the essential infimum. As to the converse $(\leq )$,
observe that whenever $Y\in L^{1}(\mathcal{F},{\mathbb{Q}})\cap (L^{1}(%
\mathcal{F},{\mathbb{P}}))^{N}$ is given with $\mathbb{E}_{\mathbb{P}}\left[
U\left( X+Y\right) \middle|\mathcal{G}\right] \geq B$, then, for any $%
\varepsilon >0$, $\mathbb{E}_{\mathbb{P}}\left[ U\left( X+Y+\varepsilon 
\mathbf{1}\right) \middle|\mathcal{G}\right] >B\,\,{\mathbb{P}}-$ a.s. by
strict monotonicity of $U$. Hence, given $Y_{(k)}$ as in %
\eqref{DynRMdeftruncated}, $k\geq k_{Y}$, defining 
\begin{equation*}
\Gamma _{K}:=\bigcap_{k\geq K}\left\{ \mathbb{E}_{\mathbb{P}}\left[ U\left(
X+Y_{(k)}+\varepsilon \mathbf{1}\right) \middle|\mathcal{G}\right] \geq
B\right\} \in \mathcal{G}
\end{equation*}%
we have that $\Gamma _{K}\subseteq \Gamma _{K+1}$ and ${\mathbb{P}}(\cup
_{K}\Gamma _{K})=1$. The argument is now similar to the one in the proof of
Claim \ref{DynRMpropinftycoincide}. As a consequence, we have 
\begin{equation}
1_{\Gamma _{K}^{c}}=0\text{ definitely in }K\,\,\,{\mathbb{P}}-\text{a.s.}.
\label{DynRMcondef}
\end{equation}%
For each $K$, take a vector $\alpha _{K}\in {\mathbb{R}}^{N}$ such that 
\begin{equation*}
U(-\left\Vert X\right\Vert _{\infty }-\left\Vert Y_{(K)}\right\Vert _{\infty
}+\varepsilon \mathbf{1}+\alpha _{K})\geq \esssup{(B)}
\end{equation*}%
where the notation for the vectors $\left\Vert X\right\Vert _{\infty
},\left\Vert Y_{(K)}\right\Vert _{\infty }$ is the same used in the proof of
Claim \ref{DynRMpropinftycoincide} and define 
\begin{equation*}
Z_{K}:=Y_{(K)}+\varepsilon \mathbf{1}+\alpha _{K}1_{\Gamma _{K}^{c}}\in
(L^{\infty }(\mathcal{F},{\mathbb{P}}))^{N}\,.
\end{equation*}%
Then clearly 
\begin{equation*}
\mathbb{E}_{\mathbb{P}}\left[ U\left( X+Y_{(K)}+\varepsilon \mathbf{1}%
\right) \middle|\mathcal{G}\right] 1_{\Gamma _{K}}+\mathbb{E}_{\mathbb{P}}%
\left[ U\left( Y_{(K)}+\varepsilon \mathbf{1}+\alpha _{K}1_{\Gamma
_{K}^{c}}\right) \middle|\mathcal{G}\right] 1_{\Gamma _{K}^{c}}\geq B
\end{equation*}%
which implies $\mathbb{E}_{\mathbb{P}}\left[ U\left( X+Z_{K}\right) \middle|%
\mathcal{G}\right] \geq B$. Hence 
\begin{align*}
\rho _{\mathcal{G}}^{{\mathbb{Q}}}(X)& \leq \liminf_{K}\sum_{j=1}^{N}\mathbb{%
E}_{\mathbb{Q}^{j}}\left[ Z_{K}^{j}\middle|\mathcal{G}\right] \\
& =\liminf_{K}\left( \left( \sum_{j=1}^{N}\mathbb{E}_{\mathbb{Q}^{j}}\left[
Y_{(K)}^{j}\middle|\mathcal{G}\right] \right) 1_{\Gamma _{K}}+\left(
\sum_{j=1}^{N}\alpha _{K}^{j}\right) 1_{\Gamma _{K}^{c}}\right) +N\varepsilon
\\
& =\lim_{K}\sum_{j=1}^{N}\mathbb{E}_{\mathbb{Q}^{j}}\left[ Y_{(K)}^{j}\middle%
|\mathcal{G}\right] +\lim_{K}\left( \sum_{j=1}^{N}\alpha _{K}^{j}1_{\Gamma
_{K}^{c}}\right) +N\varepsilon \\
& \overset{\substack{\eqref{DynRMcondef}\\ \text{(cDOM)}}}{=}\sum_{j=1}^{N}%
\mathbb{E}_{\mathbb{Q}^{j}}\left[ Y^{j}\middle|\mathcal{G}\right]
+N\varepsilon \,.
\end{align*}
\end{proof}

\begin{proof}[Proof of Theorem \protect\ref{DynRMoptforrhoq}]
By Proposition \ref{DynRMpropoptisint} we have that $\widehat{Y}\in(L^{1}(%
\mathcal{F},{\mathbb{P}}))^{N} \cap \bigcap_{{\mathbb{Q}}\in\mathscr{Q}_%
\mathcal{G}^1} L^{1}(\mathcal{F},{{\mathbb{Q}}})$. We also know that $%
\mathbb{E}_{\mathbb{P}}\left[ U\left( X+\widehat{Y}\right) \middle|\mathcal{G%
}\right] \geq B$. Hence we have 
\begin{equation*}
\rho _{\mathcal{G}}^{\widehat{{\mathbb{Q}}}}(X)\overset{ %
\eqref{DynRMrhoonl1pq}}{\leq }\sum_{j=1}^{N}\mathbb{E}_{\widehat{{\mathbb{Q}}%
}^{j}}[\widehat{Y}^{j}|\mathcal{G}]\overset{\text{Prop.}\ref%
{DynRMpropoptisint}}{\leq }\sum_{j=1}^{N}\widehat{Y}^{j}\,.
\end{equation*}
Moreover, by optimality of $\widehat{Y}$ for $\rho _{\mathcal{G}}\left(
X\right) $ and \eqref{DynRMlinkpipiq}, we have $\sum_{j=1}^{N}\widehat{Y}%
^{j}=\rho _{\mathcal{G}}\left( X\right) =\rho _{\mathcal{G}}^{\widehat{{%
\mathbb{Q}}}}(X)$. We then conclude jointly optimality in the extended sense
of $\widehat{Y}$ for $\rho _{\mathcal{G}}^{\widehat{{\mathbb{Q}}}}(X)$ and (%
\ref{new}). %
\end{proof}

\subsection{Proofs for Section \protect\ref{DynRMsectionexp}}

\label{DynRMsecproofsexp}

%
%
%

\subsubsection{Proof of Theorem \protect\ref{DynRMthmformulasgeneral}}

Let us rename 
\begin{equation}
H_{\mathcal{G}}(X):=\beta \log \left( -\frac{\beta }{B}\mathbb{E}_{\mathbb{P}%
}\left[ \exp \left( -\frac{\overline{X}}{\beta }\right) \middle|\mathcal{G}%
\right] \right) -A\,\in L^{\infty }(\mathcal{G}).  \label{HH}
\end{equation}%
We consider $\widehat{Y}$ and $\widehat{{\mathbb{Q}}}$ assigned in (\ref{YY}%
) and (\ref{QQ}). One immediately checks that $\widehat{{\mathbb{Q}}}\in %
\mathscr{Q}_{\mathcal{G}}$. As all the components $\widehat{{\mathbb{Q}}}^{j}
$ of $\widehat{{\mathbb{Q}}}$ \ are all equal, to prove that $\widehat{{%
\mathbb{Q}}}\in \mathscr{Q}_{\mathcal{G}}^{1}$ it is sufficient to show that 
$\alpha ^{1}(\widehat{{\mathbb{Q}}})\in L^{1}(\mathcal{G}).$ Let 
\begin{equation*}
V(y):=\sup_{x\in {\mathbb{R}}^{N}}\left\{
U(x)-\sum_{j=1}^{N}x^{j}y^{j}\right\} =\sum_{j=1}^{N}\left( \frac{y^{j}}{%
\alpha _{j}}\log \left( \frac{y^{j}}{\alpha _{j}}\right) -\frac{y^{j}}{%
\alpha _{j}}\right) ,\quad y\in {(0,+\infty )}^{N},
\end{equation*}%
be the the convex conjugate of $U$ and let $\lambda \in L^{\infty }(\mathcal{%
G})$ with $\lambda \geq c$ $\mathbb{P}$-a.s. for some constant $c\in
(0,+\infty )$. Then the Fenchel inequality $U(W)-V(\frac{1}{\lambda }\frac{%
\mathrm{d}\widehat{{\mathbb{Q}}}}{\mathrm{d}{\mathbb{P}}})\leq
\sum_{j=1}^{N}W^{j}\frac{1}{\lambda }\frac{\mathrm{d}\widehat{{\mathbb{Q}}}%
^{j}}{\mathrm{d}{\mathbb{P}}}$ holds $\mathbb{P}$-a.s. for any $W\in
(L^{\infty }(\mathcal{F}))^{N}$. Take any $j\in \left\{ 1,...,N\right\} $
and set $I_{\mathcal{G}}(\widehat{{\mathbb{Q}}},{\mathbb{P}}):=\mathbb{E}_{%
\mathbb{P}}\left[ \frac{\mathrm{d}\widehat{{\mathbb{Q}}}^{j}}{\mathrm{d}{%
\mathbb{P}}}\log \left( \frac{\mathrm{d}\widehat{{\mathbb{Q}}}^{j}}{\mathrm{d%
}{\mathbb{P}}}\right) \middle|\mathcal{G}\right] $. Then 
\begin{equation*}
\alpha ^{1}(\widehat{{\mathbb{Q}}})=\esssup_{\substack{ W\in (L^{\infty }(%
\mathcal{F}))^{N} \\ \mathbb{E}_{\mathbb{P}}\left[ U\left( W\right) \middle|%
\mathcal{G}\right] \geq B}}\sum_{j=1}^{N}\mathbb{E}_{\widehat{\mathbb{Q}}%
^{j}}\left[ -W^{j}\mid \mathcal{G}\right] \overset{\widehat{{\mathbb{Q}}}\in 
\mathcal{Q}_{\mathcal{G}}}{=}\esssup_{\substack{ W\in (L^{\infty }(\mathcal{F%
}))^{N} \\ \mathbb{E}_{\mathbb{P}}\left[ U\left( W\right) \middle|\mathcal{G}%
\right] \geq B}}\sum_{j=1}^{N}\lambda \mathbb{E}_{\mathbb{P}}\left[
-W^{j}\left( \frac{1}{\lambda }\frac{\mathrm{d}\widehat{{\mathbb{Q}}}}{%
\mathrm{d}{\mathbb{P}}}\right) \middle|\mathcal{G}\right] 
\end{equation*}%
\begin{align*}
& \leq \esssup_{\substack{ W\in (L^{\infty }(\mathcal{F}))^{N} \\ \mathbb{E}%
_{\mathbb{P}}\left[ U\left( W\right) \middle|\mathcal{G}\right] \geq B}}%
\left( \lambda \mathbb{E}_{\mathbb{P}}\left[ V\left( \frac{1}{\lambda }\frac{%
\mathrm{d}\widehat{{\mathbb{Q}}}}{\mathrm{d}{\mathbb{P}}}\right) \middle|%
\mathcal{G}\right] -\lambda \mathbb{E}_{\mathbb{P}}\left[ U(W)|\mathcal{G}%
\right] \right) \leq \lambda \mathbb{E}_{\mathbb{P}}\left[ V\left( \frac{1}{%
\lambda }\frac{\mathrm{d}\widehat{{\mathbb{Q}}}}{\mathrm{d}{\mathbb{P}}}%
\right) \middle|\mathcal{G}\right] -\lambda B \\
& =\sum_{j=1}^{N}\frac{1}{\alpha }_{j}\log \left( \frac{1}{\alpha _{j}}%
\right) -\sum_{j=1}^{N}\frac{1}{\alpha _{j}}+\sum_{j=1}^{N}\frac{1}{\alpha
_{j}}I_{\mathcal{G}}(\widehat{{\mathbb{Q}}},{\mathbb{P}})+\sum_{j=1}^{N}%
\frac{1}{\alpha _{j}}\log \left( \frac{1}{\lambda }\right) -\lambda B \\
& =A-\beta +\beta I_{\mathcal{G}}(\widehat{{\mathbb{Q}}},{\mathbb{P}})+\beta
\log \left( \frac{1}{\lambda }\right) -\lambda B \\
& =A+\beta I_{\mathcal{G}}(\widehat{{\mathbb{Q}}},{\mathbb{P}})+\beta \log
\left( -\frac{B}{\beta }\right) \quad \text{if }\lambda :=-\frac{\beta }{B}
\\
& =\sum_{j=1}^{N}\mathbb{E}_{\widehat{{\mathbb{Q}}}^{j}}[-X^{j}|\mathcal{G}%
]-H_{\mathcal{G}}\left( X\right) ,\,
\end{align*}%
where the last equality is obtained by direct computation using (\ref{QQ})
and (\ref{HH}). Hence $\alpha ^{1}(\widehat{{\mathbb{Q}}})\in L^{1}(\mathcal{%
G}),$  $\widehat{{\mathbb{Q}}}\in \mathscr{Q}_{\mathcal{G}}^{1}$ and 
\begin{equation}
H_{\mathcal{G}}(X)\leq \sum_{j=1}^{N}\mathbb{E}_{\widehat{{\mathbb{Q}}}%
^{j}}[-X^{j}|\mathcal{G}]-\alpha ^{1}(\widehat{{\mathbb{Q}}}).  \label{HQ}
\end{equation}
Notice that when $\mathscr{B}_{\mathcal{G}}=\mathscr{D}_{\mathcal{G}}$ then $%
\mathscr{C}_{\mathcal{G}}\cap (L^{\infty }(\mathcal{F}))^{N}=\mathscr{D}_{%
\mathcal{G}}\cap (L^{\infty }(\mathcal{F}))^{N}.$ The Theorem is proved
once we show the following chain of inequalities:%
\begin{align}
\rho _{\mathcal{G}}(X)& =\rho _{\mathcal{G}}^{\infty }(X):=\essinf\left\{
\sum_{j=1}^{N}Y^{j}\mid Y\in \mathscr{D}_{\mathcal{G}}\cap (L^{\infty }(%
\mathcal{F}))^{N},\text{ }\mathbb{E}_{\mathbb{P}}\left[ U\left( X+Y\right) %
\middle|\mathcal{G}\right] \geq B\right\}   \label{11} \\
& \leq \sum_{j=1}^{N}\widehat{Y}^{j}=H_{\mathcal{G}}(X)  \label{22} \\
& \leq \sum_{j=1}^{N}\mathbb{E}_{\widehat{{\mathbb{Q}}}^{j}}[-X^{j}|\mathcal{%
G}]-\alpha ^{1}(\widehat{{\mathbb{Q}}})\leq \esssup_{{\mathbb{Q}}\in %
\mathscr{Q}_{\mathcal{G}}^{1}}\left( \sum_{j=1}^{N}\mathbb{E}_{\mathbb{Q}%
^{j}}\left[ -X^{j}\middle|\mathcal{G}\right] -\alpha ^{1}({\mathbb{Q}}%
)\right) =\rho _{\mathcal{G}}(X).  \label{33}
\end{align}%
The equalities in (\ref{11}) and the last equality in (\ref{33}) follow from
Theorem \ref{DynRMmainthmrhoinfty}. By direct computation, $\widehat{Y}$
satisfies: $\widehat{Y}\in (L^{\infty }(\mathcal{F}))^{N},\sum_{j=1}^{N}%
\mathbb{E}_{\mathbb{P}}\left[ -\exp \left( -\alpha _{j}\left( X^{j}+\widehat{Y%
}^{j}\right) \right) \middle|\mathcal{G}\right] =B,$ $\sum_{j=1}^{N}\widehat{%
Y}^{j}=H_{\mathcal{G}}(X)\in L^{\infty }(\mathcal{G})$. Hence $\widehat{Y}$
satisfies the constraints of $\rho _{\mathcal{G}}^{\infty }(X),$ which
proves the inequality (and the equality) in (\ref{22}). The first inequality
in (\ref{33}) is shown in (\ref{HQ}), while the second one is a direct
consequence of $\widehat{{\mathbb{Q}}}\in \mathscr{Q}_{\mathcal{G}}^{1}\,$.

\subsubsection{Proof of Theorem \protect\ref{DynRMthmconsistency}\label{SecA52}}

\textbf{Equation \eqref{DynRMconsisty}}: we start observing that a
straightforward computation yields 
\begin{equation}
\widehat{Y}^{k}\left( \mathcal{G},\,X\right) =\widehat{Y}^{k}\left( \mathcal{%
H},\,X\right) +\frac{1}{\beta \alpha _{k}}\left( \rho _{\mathcal{G}}\left(
X\right) -\rho _{\mathcal{H}}\left( X\right) \right) \,\,\,\forall
\,k=1,\dots ,N\,.  \label{DynRMyfromgtoh}
\end{equation}

We also have, recalling $\sum_{j=1}^{N}\widehat{Y}^{j}\left( \mathcal{G}%
,\,X\right) =\rho _{\mathcal{G}}\left( X\right) $ and fixing $k$, that 
\begin{align*}
\widehat{Y}^{k}\left( \mathcal{H},\,-\widehat{Y}^{{}}\left( \mathcal{G}%
,\,X\right) \right)&=\widehat{Y}^{k}\left( \mathcal{G},\,X\right) +\frac{1}{%
\beta \alpha _{k}}\left( -\rho _{\mathcal{G}}\left( X\right) +\rho _{%
\mathcal{H}}\left( -\widehat{Y}^{{}}\left( \mathcal{G},\,X\right) \right)
\right) +\frac{1}{\beta \alpha _{k}}A-A_{k} \\
& \overset{\text{Eq.}\eqref{DynRMyfromgtoh}}{=}\widehat{Y}^{k}\left( 
\mathcal{H},\,X\right) +\frac{1}{\beta \alpha _{k}}\left( -\rho _{\mathcal{H}%
}\left( X\right) \right) + \\
& \frac{1}{\beta \alpha _{k}}\left( \rho _{\mathcal{H}}\left( -\widehat{Y}%
^{{}}\left( \mathcal{G},\,X\right) \right) -\rho _{\mathcal{H}}\left(
0\right) \right) +\frac{1}{\beta \alpha _{k}}(\rho _{\mathcal{H}}\left(
0\right) +A)-A_{k}\,.
\end{align*}%
It is then enough to show that $\rho _{\mathcal{H}}\left( -\widehat{Y}%
^{{}}\left( \mathcal{G},\,X\right) \right) =\rho _{\mathcal{H}}\left(
X\right) +\rho _{\mathcal{H}}\left( 0\right) $, since $\widehat{Y}^{k}\left( 
\mathcal{H},\,0\right) =\frac{1}{\beta \alpha _{k}}(\rho _{\mathcal{H}%
}\left( 0\right) +A)-A_{k}$. A direct computation yields 
\begin{align*}
\rho _{\mathcal{H}}\left( -\widehat{Y}^{{}}\left( \mathcal{G},\,X\right)
\right) &=\beta \log \left( -\frac{\beta }{B}\right) -A+\beta \log \left( 
\mathbb{E}_{\mathbb{P}}\left[ \exp \left( -\frac{1}{\beta }\left(
-\sum_{j=1}^{N}\widehat{Y}^{j}\left( \mathcal{G},\,X\right) \right) \right) %
\middle|\mathcal{H}\right] \right) \\
& \overset{\text{Eq.}\eqref{DynRMdefrcondexp}}{=}\beta \log \left( -\frac{%
\beta }{B}\right) -A+\beta \log \left( -\frac{A}{\beta }\right) + \\
& +\beta \log \left( \mathbb{E}_{\mathbb{P}}\left[ \exp \left( \frac{\beta }{%
\beta }\log \left( -\frac{\beta }{B}\mathbb{E}_{\mathbb{P}}\left[ \exp
\left( -\frac{1}{\beta }\overline{X}\right) \middle|\mathcal{G}\right]
\right) \right) \middle|\mathcal{H}\right] \right) \\
& =\rho _{\mathcal{H}}\left( 0\right) -A+\beta \log \left( -\frac{\beta }{B}%
\mathbb{E}_{\mathbb{P}}\left[ \mathbb{E}_{\mathbb{P}}\left[ \exp \left( -%
\frac{1}{\beta }\overline{X}\right) \middle|\mathcal{G}\right] \middle|%
\mathcal{H}\right] \right) \,.
\end{align*}%
Hence we have 
\begin{equation}
\rho _{\mathcal{H}}\left( -\widehat{Y}^{{}}\left( \mathcal{G},\,X\right)
\right) =\rho _{\mathcal{H}}\left( 0\right) +\rho _{\mathcal{H}}\left(
X\right) \,.  \label{DynRMrcondtwice}
\end{equation}

\textbf{Equation \eqref{DynRMconsistq}}: we have by \eqref{DynRMdefoptq} and
using \eqref{DynRMdefrcondexp} that 
\begin{equation}
\frac{\mathrm{d}\widehat{{\mathbb{Q}}}^{k}}{\mathrm{d}{\mathbb{P}}}\left( 
\mathcal{G},\,X\right) \frac{\mathrm{d}\widehat{{\mathbb{Q}}}^{k}}{\mathrm{d}%
{\mathbb{P}}}\left( \mathcal{H},\,-\widehat{Y}^{{}}\left( \mathcal{G}%
,\,X\right) \right) =\frac{\exp \left( -\frac{\overline{X}}{\beta }\right) }{%
\mathbb{E}_{\mathbb{P}}\left[ \exp \left( -\frac{\overline{X}}{\beta }%
\right) \middle|\mathcal{G}\right] }\frac{\exp \left( \frac{\rho _{\mathcal{G%
}}\left( X\right) }{\beta }\right) }{\mathbb{E}_{\mathbb{P}}\left[ \exp
\left( \frac{\rho _{\mathcal{G}}\left( X\right) }{\beta }\right) \middle|%
\mathcal{H}\right] }\,.  \label{DynRMintermconsist1}
\end{equation}%
We now see, just using \eqref{DynRMdefrcondexp}, that 
\begin{equation*}
\exp \left( \frac{\rho _{\mathcal{G}}\left( X\right) }{\beta }\right)
=\left( -\frac{\beta }{B}\right) \exp \left( -\frac{A}{\beta }\right) 
\mathbb{E}_{\mathbb{P}}\left[ \exp \left( -\frac{\overline{X}}{\beta }%
\right) \middle|\mathcal{G}\right] \,,
\end{equation*}%
\begin{equation*}
\mathbb{E}_{\mathbb{P}}\left[ \exp \left( \frac{\rho _{\mathcal{G}}\left(
X\right) }{\beta }\right) \middle|\mathcal{H}\right] =\left( -\frac{\beta }{B%
}\right) \exp \left( -\frac{A}{\beta }\right) \mathbb{E}_{\mathbb{P}}\left[
\exp \left( -\frac{\overline{X}}{\beta }\right) \middle|\mathcal{H}\right]
\,.
\end{equation*}%
Direct substitution in \eqref{DynRMintermconsist1} yields 
\begin{equation*}
\frac{\mathrm{d}\widehat{{\mathbb{Q}}}^{k}}{\mathrm{d}{\mathbb{P}}}\left( 
\mathcal{G},\,X\right) \frac{\mathrm{d}\widehat{{\mathbb{Q}}}^{k}}{\mathrm{d}%
{\mathbb{P}}}\left( \mathcal{H},\,-\widehat{Y}^{{}}\left( \mathcal{G}%
,\,X\right) \right) =\frac{\exp \left( -\frac{\overline{X}}{\beta }\right) }{%
\mathbb{E}_{\mathbb{P}}\left[ \exp \left( -\frac{\overline{X}}{\beta }%
\right) \middle|\mathcal{H}\right] }\overset{\text{Eq.}\eqref{DynRMdefoptq}}{%
=}\frac{\mathrm{d}\widehat{{\mathbb{Q}}}^{k}}{\mathrm{d}{\mathbb{P}}}\left( 
\mathcal{H},\,X\right) \,.
\end{equation*}

\textbf{Equation \eqref{DynRMconsista}}: by definition \eqref{DynRMdefopta}
and using the fact that 
\begin{equation*}
\mathbb{E}_{\mathbb{P}}\left[ \frac{\mathrm{d}\widehat{{\mathbb{Q}}}^{k}}{%
\mathrm{d}{\mathbb{P}}}\left( \mathcal{H},\,-\widehat{Y}^{{}}\left( \mathcal{%
G},\,X\right) \right) \middle|\mathcal{H}\right] =1\,\,\,\forall \,k=1,\dots
,N
\end{equation*}%
we have 
\begin{equation*}
\widehat{a}^{k}\left( \mathcal{H},\,-\widehat{a}^{{}}\left( \mathcal{G}%
,\,X\right) \right) =\mathbb{E}_{\mathbb{P}}\left[ \widehat{Y}^{k}\left( 
\mathcal{H},\,-\widehat{a}^{{}}\left( \mathcal{G},\,X\right) \right) \frac{%
\mathrm{d}\widehat{{\mathbb{Q}}}^{k}}{\mathrm{d}{\mathbb{P}}}\left( \mathcal{%
H},\,-\widehat{a}^{{}}\left( \mathcal{G},\,X\right) \right) \middle|\mathcal{%
H}\right] =E+F+G+H
\end{equation*}%
where 
\begin{align*}
E&:=\mathbb{E}_{\mathbb{P}}\left[ -\left( -\widehat{a}^{{k}}\left( \mathcal{G%
},\,X\right) \right) \frac{\mathrm{d}\widehat{{\mathbb{Q}}}^{k}}{\mathrm{d}{%
\mathbb{P}}}\left( \mathcal{H},\,-\widehat{a}^{{}}\left( \mathcal{G}%
,\,X\right) \right) \middle|\mathcal{H}\right] \,, \\
F&:=\mathbb{E}_{\mathbb{P}}\left[ \frac{1}{\beta \alpha _{k}}%
\sum_{j=1}^{N}\left( -\widehat{a}^{j}\left( \mathcal{G},\,X\right) \right) 
\frac{\mathrm{d}\widehat{{\mathbb{Q}}}^{k}}{\mathrm{d}{\mathbb{P}}}\left( 
\mathcal{H},\,-\widehat{a}^{{}}\left( \mathcal{G},\,X\right) \right) \middle|%
\mathcal{H}\right] \,,
\end{align*}%
\begin{align*}
G&:=\mathbb{E}_{\mathbb{P}}\left[ \frac{1}{\beta \alpha _{k}}\rho _{\mathcal{%
H}}\left( -\widehat{a}^{{}}\left( \mathcal{G},\,X\right) \right) \frac{%
\mathrm{d}\widehat{{\mathbb{Q}}}^{k}}{\mathrm{d}{\mathbb{P}}}\left( \mathcal{%
H},\,-\widehat{a}^{{}}\left( \mathcal{G},\,X\right) \right) \middle|\mathcal{%
H}\right] \,, \\
H&:=\mathbb{E}_{\mathbb{P}}\left[ \left( \frac{1}{\beta \alpha _{k}}%
A-A_{k}\right) \frac{\mathrm{d}\widehat{{\mathbb{Q}}}^{k}}{\mathrm{d}{%
\mathbb{P}}}\left( \mathcal{H},\,-\widehat{a}^{{}}\left( \mathcal{G}%
,\,X\right) \right) \middle|\mathcal{H}\right] =\frac{1}{\beta \alpha _{k}}%
A-A_{k}\,.
\end{align*}%
We now work separately on each of the above random variables:

\begin{itemize}
\item considering $E$, by \eqref{DynRMdefopta}, observing that $\frac{%
\mathrm{d}\widehat{{\mathbb{Q}}}^{k}}{\mathrm{d}{\mathbb{P}}}\left(\mathcal{H%
},\,-\widehat{a}^{}\left(\mathcal{G},\,X\right)\right)\in L^\infty(\mathcal{G%
})$ and %
using the fact that $\rho_{\mathcal{H}}\left(X\right)\in L^\infty(\mathcal{H}%
)$ we get 
\begin{equation}  \label{DynRMeqE}
E=\widehat{a}^{k}\left(\mathcal{H},\,X\right)+\frac{1}{\beta\alpha_k}\mathbb{%
E}_\mathbb{P} \left[\rho_{\mathcal{G}}\left(X\right)\frac{\mathrm{d}\widehat{%
{\mathbb{Q}}}^{k}}{\mathrm{d}{\mathbb{P}}}\left(\mathcal{H},\,X\right)\middle%
|\mathcal{H}\right]-\frac{1}{\beta\alpha_k}\rho_{\mathcal{H}%
}\left(X\right)\,.
\end{equation}

\item We now move to $F$. First, computing $\frac{\mathrm{d}\widehat{{%
\mathbb{Q}}}^{k}}{\mathrm{d}{\mathbb{P}}}\left( \mathcal{H},\,-\widehat{a}%
^{{}}\left( \mathcal{G},\,X\right) \right) $ we get %
\begin{equation}
\frac{\mathrm{d}\widehat{{\mathbb{Q}}}^{k}}{\mathrm{d}{\mathbb{P}}}\left( 
\mathcal{H},\,-\widehat{a}^{{}}\left( \mathcal{G},\,X\right) \right) =\frac{%
\mathbb{E}_{\mathbb{P}}\left[ \exp \left( -\frac{\overline{X}}{\beta }%
\right) \middle|\mathcal{G}\right] }{\mathbb{E}_{\mathbb{P}}\left[ \exp
\left( -\frac{\overline{X}}{\beta }\right) \middle|\mathcal{H}\right] }\,.
\label{DynRMeqrnweird}
\end{equation}%
After some additional tedious computation we obtain 
\begin{equation}
F=-\frac{1}{\beta \alpha _{k}}\mathbb{E}_{\mathbb{P}}\left[ \rho _{\mathcal{G%
}}\left( X\right) \frac{\mathrm{d}\widehat{{\mathbb{Q}}}^{k}}{\mathrm{d}{%
\mathbb{P}}}\left( \mathcal{H},\,X\right) \middle|\mathcal{H}\right] \,.
\label{DynRMeqF}
\end{equation}

\item To compute $G$, we first see that $\rho _{\mathcal{H}}\left( -\widehat{%
a}^{{}}\left( \mathcal{G},\,X\right) \right) =\rho _{\mathcal{H}}\left( -%
\widehat{Y}^{{}}\left( \mathcal{G},\,X\right) \right) $. We can thus exploit %
\eqref{DynRMrcondtwice} to see that $\rho _{\mathcal{H}}\left( -\widehat{a}%
^{{}}\left( \mathcal{G},\,X\right) \right) =\rho _{\mathcal{G}}\left(
0\right) +\rho _{\mathcal{H}}\left( X\right) $. Using also %
\eqref{DynRMeqrnweird} we get 
\begin{equation}
\begin{split}
G& =\frac{1}{\beta \alpha _{k}}\left( \rho _{\mathcal{H}}\left( 0\right) +%
\mathbb{E}_{\mathbb{P}}\left[ \rho _{\mathcal{H}}\left( X\right) \frac{%
\mathbb{E}_{\mathbb{P}}\left[ \exp \left( -\frac{\overline{X}}{\beta }%
\right) \middle|\mathcal{G}\right] }{\mathbb{E}_{\mathbb{P}}\left[ \exp
\left( -\frac{\overline{X}}{\beta }\right) \middle|\mathcal{H}\right] }%
\middle|\mathcal{H}\right] \right) \\
& \overset{\rho _{\mathcal{H}}\left( X\right) \in L^{\infty }(\mathcal{H})}{=%
}\frac{1}{\beta \alpha _{k}}\left( \rho _{\mathcal{H}}\left( 0\right) +\rho
_{\mathcal{H}}\left( X\right) \right) \,.
\end{split}
\label{DynRMeqG}
\end{equation}

\item Recalling \eqref{DynRMdefopta} we have $\widehat{a}^{k}\left( \mathcal{%
H},\,0\right) =H+\frac{1}{\beta \alpha _{k}}\rho _{\mathcal{H}}\left(
0\right) $ hence 
\begin{equation}
H=\widehat{a}^{k}\left( \mathcal{H},\,0\right) -\frac{1}{\beta \alpha _{k}}%
\rho _{\mathcal{H}}\left( 0\right) \,.  \label{DynRMeqH}
\end{equation}
\end{itemize}

More detailed computations can be found in \cite{DoldiThesis21}, proof of
Theorem 3.5.8. Summing \eqref{DynRMeqE}, \eqref{DynRMeqF}, \eqref{DynRMeqG}, %
\eqref{DynRMeqH} most terms simplify and we get 
\begin{equation*}
\widehat{a}^{k}\left( \mathcal{H},\,-\widehat{a}^{{}}\left( \mathcal{G}%
,\,X\right) \right) =E+F+G+H=\widehat{a}^{k}\left( \mathcal{H},\,X\right) +%
\widehat{a}^{k}\left( \mathcal{H},\,0\right) \,\,\,k=1,\dots ,N\,.
\end{equation*}

{\ 
\bibliographystyle{abbrv}
\bibliography{BibAll}
} \newpage

\end{document}